\documentclass{LMCS}

\def\doi{8(4:11)2012}
\lmcsheading%
{\doi}
{1--44}
{}
{}
{Dec.~29, 2011}
{Oct.~23, 2012}
{}
 
\usepackage{enumerate}
\usepackage{hyperref}
\usepackage[off]{auto-pst-pdf}

\usepackage{proof}
\usepackage{subfigure}
\usepackage{epsf}
\usepackage{wrapfig}
\usepackage{mathrsfs}
\usepackage{url}
\usepackage{calc}
\usepackage{cmll}
\usepackage{amsmath,amssymb,amsfonts,mathrsfs,stmaryrd}
\usepackage[all,color,frame]{xy}
\usepackage{pstricks,pst-node,pst-text,pst-3d}
\usepackage{epsfig}
\usepackage{mathabx}
\usepackage{graphics}

\newcommand{\midd}{\; \; \mbox{\Large{$\mid$}}\;\;}

\newcommand{\funone}{f}
\newcommand{\funtwo}{g}

\newcommand{\PCFld}{$\mathsf{d}\ell\mathsf{PCF}$}
\newcommand{\PCF}{$\mathsf{PCF}$}
\newcommand{\DML}{$\mathsf{DML}$}
\newcommand{\BLL}{$\mathsf{BLL}$}

\newcommand{\QBAL}{$\mathsf{QBAL}$}
\newcommand{\Kld}{\mathsf{K}_{\mathsf{PCF}}}

\newcommand{\carone}{\mathcal{S}}
\newcommand{\arone}{\alpha}
\newcommand{\oterms}[2]{\mathcal{O}(#1,#2)}

\newcommand{\mnu}{\dotdiv}
\newcommand{\varsetone}{\mathcal{V}}
\newcommand{\assone}{\rho}
\newcommand{\eqpone}{\mathcal{E}}

\newcommand{\eqpun}{\mathcal{U}}
\newcommand{\eqpless}{\models}
\newcommand{\funsymone}{\texttt{f}}

\newcommand{\emfs}{\texttt{empty}}
\newcommand{\pairfs}{\texttt{pair}}
\newcommand{\evalfs}{\texttt{eval}}
\newcommand{\pairing}{\mathit{pairing}}
\newcommand{\ftermone}{t}
\newcommand{\ftermtwo}{s}
\newcommand{\semu}[1]{\llbracket #1\rrbracket}

\newcommand{\semt}[3]{\llbracket #1\rrbracket^{#3}_{#2}}

\newcommand{\itermone}{\mathrm{I}}
\newcommand{\itermtwo}{\mathrm{J}}
\newcommand{\itermthree}{\mathrm{K}}
\newcommand{\itermfour}{\mathrm{H}}
\newcommand{\itermfive}{\mathrm{L}}
\newcommand{\itermsix}{\mathrm{M}}
\newcommand{\itermseven}{\mathrm{N}}
\newcommand{\itermeight}{\mathrm{P}}
\newcommand{\itermnine}{\mathrm{R}}
\newcommand{\itermten}{\mathrm{Q}}
\newcommand{\itermeleven}{\mathrm{S}}
\newcommand{\itermtwelve}{\mathrm{T}}

\newcommand{\Dfscomb}[4]{\bigotriangleup_{#1}^{#2,#3}#4}

\newcommand{\gn}[2]{\ulcorner #1,#2 \urcorner}

\renewcommand{\conv}[1]{#1\Downarrow}

\newcommand{\ivarone}{a}
\newcommand{\ivartwo}{b}
\newcommand{\ivarthree}{c}
\newcommand{\ivarfour}{d}
\newcommand{\ivarfive}{e}

\newcommand{\iconstraintone}{\Phi}
\newcommand{\iconstrainttwo}{\Psi}
\newcommand{\icontextone}{\phi}

\newcommand{\natone}{n}
\newcommand{\nattwo}{m}
\newcommand{\natthree}{p}

\newcommand{\natit}[1]{\textbf{#1}}

\newcommand{\kleq}{\simeq}

\newcommand{\tree}[1]{\mathbb{T}_{#1}}

\newcommand{\ifadd}{\texttt{add}}
\newcommand{\ifgtz}{\texttt{gt}}
\newcommand{\ifmult}{\texttt{mult}}

\newcommand{\pvdash}{\Vdash}

\newcommand{\typeone}{\sigma}
\newcommand{\typetwo}{\tau}
\newcommand{\typethree}{\mu}
\newcommand{\typefour}{\gamma}
\newcommand{\typefive}{\delta}
\newcommand{\typesix}{\eta}
\newcommand{\typeseven}{\beta}
\newcommand{\gtypeone}{\theta}
\newcommand{\gtypetwo}{\gamma}

\newcommand{\mtypeone}{A}
\newcommand{\mtypetwo}{B}
\newcommand{\mtypethree}{C}
\newcommand{\mtypefour}{D}
\newcommand{\mtsum}[2]{#1\mtsums #2}
\newcommand{\mtsums}{\uplus}

\newcommand{\Nat}{{\tt Nat}}

\newcommand{\tcontextone}{\Gamma}
\newcommand{\tcontexttwo}{\Delta}
\newcommand{\tcontexthree}{\Sigma}
\newcommand{\tcontextfour}{\Theta}
\newcommand{\qlin}[3]{[{#1}]\cdot{#2}\lin #3}
\newcommand{\qbang}[2]{[{#1}]\cdot{#2}}
\newcommand{\qbangwa}[1]{[{#1}]}
\newcommand{\tless}{\sqsubseteq}
\newcommand{\teq}{\cong}
\renewcommand{\sb}[3]{#1\{#3/#2\}}
\newcommand{\sbt}[2]{#1#2}
\newcommand{\sbst}[2]{\{#2/#1\}}
\newcommand{\sbstone}{\theta}

\newcommand{\natleq}{($\Nat$.\mathit{l})}
\newcommand{\linleq}{(\lin.\mathit{l})}
\newcommand{\qbleq}{(\qbang{-}{}.\mathit{l})}
\newcommand{\nattdef}{($\Nat$.\mathit{t})}
\newcommand{\lintdef}{(\lin.\mathit{t})}
\newcommand{\qbtdef}{(\qbang{-}{}.\mathit{t})}
\newcommand{\Axty}{V}
\newcommand{\Lamty}{L}
\newcommand{\Apty}{A}
\newcommand{\Natty}{N}
\newcommand{\Sucty}{S}
\newcommand{\Predty}{P}
\newcommand{\Ifty}{F}
\newcommand{\Recty}{R}

\newcommand{\rt}{\rightarrow}
\newcommand{\rts}{\rightarrow^*}
\newcommand{\ev}{\Downarrow}
\newcommand{\envone}{\rho}
\newcommand{\envtwo}{\mu}
\newcommand{\stone}{\xi}
\newcommand{\sttwo}{\theta}
\newcommand{\pairone}{c}
\newcommand{\casecl}[3]{(#1,#2,#3)}
\newcommand{\confone}{C}
\newcommand{\conftwo}{D}
\newcommand{\confthree}{E}

\newcommand{\prog}{\mathcal{P}}

\newcommand{\lin}{\ensuremath{\multimap}}
\newcommand{\arr}{\rightarrow}

\newcommand{\suc}{{\tt s}}
\newcommand{\pred}{{\tt p}}
\newcommand{\zero}{{\tt 0}}
\newcommand{\num}{{\tt n}}
\newcommand{\val}[1]{\underline{\tt #1}}
\newcommand{\CASE}[3]{{\ \tt ifz\ }#1{\ \tt then\ }#2{\ \tt else\ }#3}
\newcommand{\REC}[2]{{\tt fix}\ #1.#2}

\newcommand{\varone}{x}

\newcommand{\conone}{\Gamma}

\newcommand{\emcon}{\emptyset}

\newcommand{\prov}{\mathrel{\rhd}}
\newcommand{\tdone}{\pi}
\newcommand{\tdtwo}{\rho}
\newcommand{\tdthree}{\nu}
\newcommand{\tdfour}{\mu}
\newcommand{\tdfive}{\xi}

\newcommand{\termone}{t}
\newcommand{\termtwo}{u}
\newcommand{\termthree}{v}
\newcommand{\termfour}{w}

\newcommand{\termdbl}{\mathtt{dbl}}
\newcommand{\termdiv}{\mathtt{omega}}

\newcommand{\sigone}{\Sigma}

\newcommand{\sigun}{\Sigma_\eqpun}

\newcommand{\size}[1]{|#1|}

\newcommand{\NN}{\mathbb{N}}

\newcommand{\TtoNDT}[1]{(\!|#1|\!)}

\newenvironment{varitemize}
{
\begin{list}{\labelitemi}
{\setlength{\itemsep}{0pt}
 \setlength{\topsep}{0pt}
 \setlength{\parsep}{0pt}
 \setlength{\partopsep}{0pt}
 \setlength{\leftmargin}{15pt}
 \setlength{\rightmargin}{0pt}
 \setlength{\itemindent}{0pt}
 \setlength{\labelsep}{5pt}
 \setlength{\labelwidth}{10pt}
}}
{
 \end{list} 
}

\newcounter{numberone}

\newenvironment{varenumerate}
{
\begin{list}{\arabic{numberone}.}
{
  \usecounter{numberone}
  \setlength{\itemsep}{0pt}
  \setlength{\topsep}{0pt}
  \setlength{\parsep}{0pt}
  \setlength{\partopsep}{0pt}
  \setlength{\leftmargin}{15pt}
  \setlength{\rightmargin}{0pt}
  \setlength{\itemindent}{0pt}
  \setlength{\labelsep}{5pt}
  \setlength{\labelwidth}{15pt}
}}
{
\end{list} 
}

\begin{document}

\title[Linear Dependent Types and Relative Completeness]{Linear Dependent Types and Relative Completeness\rsuper*}

\author[U.~Dal Lago]{Ugo Dal Lago\rsuper a}
\address{{\lsuper a}Dipartimento di Scienze dell'Informazione -- Universit\`a di Bologna \\
 EPI FOCUS -- INRIA Sophia Antipolis}
\email{dallago@cs.unibo.it}

\author[M.~Gaboardi]{Marco Gaboardi\rsuper b}	
\address{{\lsuper b}Dipartimento di Scienze dell'Informazione -- Universit\`a di Bologna \\
Computer and Information Science Department --  University of Pennsylvania \\
 EPI FOCUS -- INRIA Sophia Antipolis}
\email{gaboardi@cs.unibo.it}
\thanks{{\lsuper b}Marco Gaboardi was supported by the European Community's Seventh Framework 
  Programme (FP7/2007-2013) under grant agreement n$^\circ$ 272487.}

\titlecomment{{\lsuper *}This work is partially supported by the INRIA 
ARC project ``ETERNAL''. This is a revised and extended version of a paper with the same title which has
appeared in the proceedings of LICS 2011.}

\begin{abstract}
A system of linear dependent types for the $\lambda$-calculus
with full higher-order recursion, called \PCFld, is introduced and proved
sound and relatively complete. Completeness holds in
a strong sense: \PCFld\ is not only able to precisely
capture the functional behavior of \PCF\ programs (i.e. how the output
relates to the input) but also some of their intensional properties, namely
the complexity of evaluating them with Krivine's Machine. \PCFld\
is designed around dependent types and linear logic and
is parametrized on the underlying language of index terms,
which can be tuned so as to sacrifice completeness for tractability.
\end{abstract}

\maketitle

\section{Introduction}
Type systems are powerful tools in the design of programming 
languages. While they have been employed traditionally to guarantee 
weak properties of programs (e.g. ``well-typed programs cannot go wrong''),
it is becoming more and more evident that they can be useful
when stronger properties are needed, such as security~\cite{journals/jcs/VolpanoIS96,SabelfeldMyers03}, 
termination~\cite{conf/csl/BartheGR08}, monadic temporal
properties~\cite{conf/lics/KobayashiO09} or
resource bounds~\cite{HOAA_POPL10,CraryWeirich00}.

One key advantage of type systems seen as formal methods is their
simplicity and their close relationship with programs --- checking whether a program has a type or even inferring the
(most general) type of a program is often decidable. The price to pay
is the incompleteness of most type systems: there are programs satisfying the property at hand
which cannot be given a type. This is in contrast with other formal methods,
like program logics~\cite{ABO09-Book3} where completeness is always a desirable feature, although
it only holds relatively to an oracle. Graphically, the situation is similar to the one in Figure~\ref{fig:graph}:
type systems are bound to be in the lower left corner of the diagram, where both the degree
of completeness and the complexity of the property under consideration is low; program
logics, on the other hand, are confined to the upper-right corner, where checking for 
derivability is almost always undecidable.
\begin{figure}
\begin{center}
  \includegraphics[bb = 0.144000 23.741999 142.667996 213.695993,
    angle = 270]{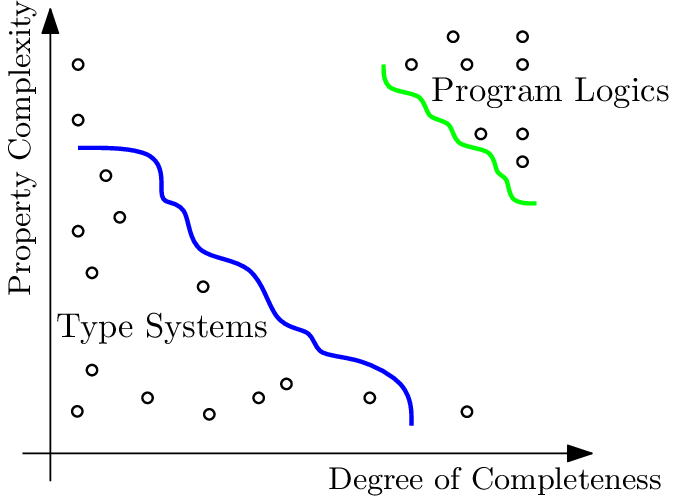}
\end{center}
\caption{Type Systems and Program Logics}
\label{fig:graph}
\end{figure}

One specific research field in which the just-described scenario 
manifests itself is implicit computational complexity, in which one aims
at defining characterizations of complexity classes by programming
languages and logical systems. Many type systems have been introduced
capturing, for instance, the polynomial time computable
functions \cite{conf/lics/Hofmann99a,BaillotTerui,BaillotGaboardiMogbil09esop}. 
All of them, under mild assumptions, can be employed as tools to certify programs as asymptotically time efficient. 
However, a tiny slice of the polytime \emph{programs} are generally typable, since
the underlying complexity class $\mathbf{FP}$ is only characterized in a purely extensional sense ---
for every function in $\mathbf{FP}$ there is \emph{at least one} typable
program computing it.

The main contribution of this paper is a \emph{type system} for the $\lambda$-calculus with
full recursion, called \PCFld, which is sound \emph{and complete}. Types of \PCFld\ are obtained, in the
spirit of \DML~\cite{POPL99*214,DBLP:journals/jfp/Xi07}, by decorating types 
of ordinary \PCF~\cite{plotkin77tcs,gunter92mit} with \emph{index terms}.
These are first-order terms freely generated from variables, function symbols
and a few more term constructs. They are indicated with metavariables like
$\itermone, \itermtwo,\itermthree$.
Type decoration reflects the standard decomposition of types into \emph{linear types} 
(as suggested by linear logic \cite{Girard87tcs}), and is inspired by recent works on 
the expressivity of bounded logics \cite{conf/tlca/LagoH09}.

Index terms and linear types permit to describe program properties
with a fine granularity. More precisely, \PCFld\ enjoys the following two properties:
\begin{varitemize}
\item
  \emph{Soundness}: if $\termone$ is a program
  and $\vdash_\itermthree\termone:\Nat[\itermone,\itermtwo]$,
  then $\termone$ evaluates to a natural number which lies
  between $\itermone$ and $\itermtwo$ and this evaluation takes
  at most $(\itermthree+1)\cdot\size{\termone}$ steps;
\item
  \emph{Completeness}: if $\termone$ is typable in \PCF\ and
  evaluates to a natural number $\natone$ in $\nattwo$ steps, then
  $\vdash_\itermone\termone:\Nat[\natit{\natone},\natit{\natone}]$, 
  where $\itermone\leq\nattwo$.
\end{varitemize}
Completeness of \PCFld\ holds not only for programs (i.e. terms of
ground types) but also for functions on the natural
numbers (see Section~\ref{sec:uniformization} for further details). 
Moreover, typing judgments tell us something 
about the functional behavior of programs
but also about their non-functional one, namely
the number of steps needed to evaluate the term in Krivine's Abstract
Machine.

As the title of this paper suggests, completeness of \PCFld\ holds in a relative
sense. Indeed, the behavior of programs can be precisely captured 
only in presence of a complete oracle for the truth of certain
assumptions in typing rules. This is exactly what happens in program logics such
as Floyd-Hoare's logic, where all true partial correctness assertions can be derived
\emph{provided} one is allowed to use all true sentences of 
first order arithmetic as axioms \cite{Cook78}. In \PCFld, those assumptions
take the form of (in)equalities between index terms, to be verified when
function symbols are interpreted as partial functions on natural numbers according
to an equational program $\eqpone$.
Actually, the whole of \PCFld\ is \emph{parametrized} on such an $\eqpone$,
but while soundness holds independently of the specific 
$\eqpone$, completeness, as is to be expected, holds only if $\eqpone$
is sufficiently powerful to encode all total computable functions (i.e. if
$\eqpone$ is universal). In other words, \PCFld\ can be claimed to be not
\emph{a} type system, but \emph{a family of} type systems obtained by 
taking a specific $\eqpone$ as the underlying ``logic'' of index terms.
The simpler $\eqpone$, the easier type checking and type inference are; the more
complex $\eqpone$, the larger the class of captured programs.

The design of \PCFld\ has been very much influenced by linear
logic~\cite{Girard87tcs}, and in particular by systems of indexed and bounded linear 
logic~\cite{GSS92,conf/tlca/LagoH09}, which have been recently shown to subsume other ICC systems as for the class
of programs they capture~\cite{conf/tlca/LagoH09}. One of the many ways to ``read'' \PCFld\ is as a
variation on the theme of \BLL~\cite{GSS92} obtained by generalizing polynomials to
arbitrary functions. The idea of going beyond a restricted, 
fixed class of bounds comes from Xi's work on \DML~\cite{POPL99*214,DBLP:journals/jfp/Xi07}. Cost 
recurrences for first order \DML\ programs have been studied \cite{conf/icfp/Grobauer01}.
No similar completeness results for dependent types are known, however.
\section{Types and Program Properties: An Informal Account}
Consider the following program:
$$
\termdbl=\REC{\funone}{\lambda\varone.\CASE{\varone}{\val{0}}{\suc(\suc(\funone(\pred(\varone))))}}.
$$
In a monomorphic, traditionally designed type system like \PCF~\cite{plotkin77tcs,gunter92mit}, the term $\termdbl$
receives type $\Nat\arr\Nat$. As a consequence, $\termdbl$ computes a function on natural numbers without ``going wrong'': 
it takes in input a natural number, and produces in output another
natural number (if any).
The type $\Nat\arr\Nat$, however, does not give any information
about \emph{which} specific function on the natural numbers $\termdbl$ computes. 
Indeed, in \PCF\ (and in most real-world programming languages) any program computing a 
function on natural numbers, being it for instance the identity function or
(a unary version of) the Ackermann function, can be typed by $\Nat\arr\Nat$. 

Some modern type systems allow one to construct and use types like $\typetwo=\Nat[\ivarone]\arr\Nat[\natit{2}\times\ivarone]$, which
tell not only what set or domain (the interpretation of) the term belongs to, but also which specific element of the
domain the term actually denotes. The type $\typetwo$, for example, could be attributed only to those programs
computing the function $\natone\mapsto 2\times\natone$, including $\termdbl$. Types of this form can be constructed in
dependent and sized type theories \cite{POPL99*214,conf/csl/BartheGR08}. The type system \PCFld\ introduced in this paper offers
this possibility, too. But, as a first contribution, it further allows to specify \emph{imprecise} types, like
$\Nat[\natit{5},\natit{8}]$, which stands for the type of those natural numbers between $5$ and $8$ (included).

A property of programs which is completely ignored by ordinary type systems is 
termination, at least if full recursion is in the underlying language. 
Typing a term $\termone$  with $\Nat\arr\Nat$ does not guarantee that $\termone$, when
applied to a natural number, terminates. In \PCF\ this is even worse: $\termone$ could possibly diverge \emph{itself}!
Consider, as another example, a slight modification of $\termdbl$, namely
$$
\termdiv=\REC{\funone}{\lambda\varone.\CASE{\varone}{\val{0}}{\suc(\suc(\funone(\varone)))}}.
$$
It behaves as $\termdbl$ when fed with $0$, but it diverges when it receives a positive natural
number as an argument. But look: $\termdiv$ is not so different from $\termdbl$. Indeed, the second can
be obtained from the first by feeding not $\varone$ but $\pred(\varone)$ to $\funone$. And any
type systems in which $\termdbl$ and $\termdiv$ are somehow recognized as being fundamentally different
must be able to detect the presence of $\pred$ in $\termdbl$ and deduct termination
from it. Indeed, sized types \cite{conf/csl/BartheGR08} and dependent types \cite{LICS01*231} are able to do so.

Going further, we could ask the type system to be able not only to guarantee termination, but also
to somehow evaluate the time or space consumption of programs. For example, we could be
interested in knowing that $\termdbl$ takes a polynomial number of steps to be evaluated on
any natural number. This cannot be achieved neither using classical type systems nor
using systems of sized types, at least when traditionally
formulated. 
However, some type systems able to control the complexity
of programs exist. Good examples are type systems
for amortized analysis \cite{HOAA_POPL10,HAH11}
or those using ideas from linear logic \cite{BaillotTerui,BaillotGaboardiMogbil09esop}.
In those type systems, typing judgements carry, 
besides the usual type information, some 
additional information about the resource consumption 
of the underlying program. As an example,
$\termdbl$ could be given a type as follows
$$
\vdash_{\itermone}\termdbl:\Nat\to\Nat
$$
where $\itermone$ is some cost information for $\termdbl$. 
This way, building a type derivation and inferring resource consumption can be done 
at the same time.
 
The type system \PCFld\ we propose in this paper 
makes some further steps in this direction. 
First of all, it combines some of the ideas presented above with the ones of
bounded linear logic.
\BLL\ allows one to explicitly count the number
of times functions
use their arguments (in rough notation, 
$!_{\natit{\nattwo}} \typeone\lin \typetwo$ says that the argument of type 
$\typeone$ is used $\nattwo$ times).
This permits to extract natural cost functions from type
  derivations. 
The cost of evaluating a term will be measured 
by counting how many times function arguments need to be copied 
during evaluation. 
Making this information explicit in types permits to compute the cost step by step during
the type derivation process. 
By the way, previous works by the first author~\cite{TOCL2009b} show that this way
of attributing a cost to (proofs seen as) programs is sound and precise as a way to measure
their time complexity.
Intuitively, typing judgements in \PCFld\ can be thought as:
$$
\vdash_{\itermtwo}\termone\ :\ !_{\natit{\nattwo}}\ \Nat[\ivarone]\lin\Nat[\itermone].
$$
where $\itermone$ and $\itermtwo$ can be derived while building a type derivation, exploiting 
the information carried by the modalities. In fact, the quantitative information in $!_{\natit{\nattwo}}$ 
allows to statically determine the number of times any subterm will 
be copied during evaluation. But this is not sufficient:
analogously to what happens in \BLL, \PCFld\ makes types more parametric.
A rough type as $!_{\natit{\natone}} \typeone\lin \typetwo$ is replaced by the 
more parametric type $\qlin{\ivarone<\natit{\natone}}{\typeone}{\typetwo}$, which tells us
that the argument will be used $\natone$ times, and each instance has type $\typeone$ 
\emph{where, however} the variable $\ivarone$ is instantiated with a value less than $\natone$.
This allows to type each copy of the argument differently but uniformly, since all 
instances of $\typeone$ have the same \PCF\ skeleton.
This form of \emph{uniform linear dependence} is actually crucial in obtaining
the result which makes \PCFld\ different from similar type systems, namely completeness.

Finally, as already stressed in the Introduction, \PCFld\ is also parametric in 
the class of functions (in the form of an equational program $\eqpone$) that can 
be used to reason about types and costs. This permits to tune the type system,
as described in Section~\ref{sec:type-checking} below.

Anticipating on the next section, we can say that $\termdbl$ 
can be typed as follows in \PCFld:
$$
\vdash_{\ivarone}^{\eqpone}\termdbl:\qbang{\ivartwo<\ivarone+\natit{1}}\Nat[\ivarone]\lin\Nat[\natit{2}\times\ivarone].
$$
This tells us that the argument will be used $\ivarone+1$ times by $\termdbl$, and 
that the cost of evaluation will be itself proportional to $\ivarone$. 
\section{\PCFld}
In this section, the language of programs and the type system \PCFld\ for it will be introduced formally.
Some of their basic properties will be described. The type system \PCFld\ is based 
on the notion of an index term whose semantics, in turn, is defined by an equational program.
As a consequence, all these notions must be properly introduced and are the subject of 
Section~\ref{sect:itep} below.
\subsection{Index Terms and Equational Programs}\label{sect:itep}
Syntactically, index terms are built either from function symbols from a given  signature or
by applying any of two special term constructs.

Formally, a \emph{signature} $\sigone$ is a pair $(\carone,\arone)$ where
$\carone$ is a finite set of \emph{function symbols}
and $\arone:\carone\rightarrow\NN$ assigns an \emph{arity}
to every function symbol. Index terms on a given signature $\sigone=(\carone,\arone)$ are generated by the following grammar:
\begin{align*}
  \itermone,\itermtwo,\itermthree &::= \ivarone\midd\funsymone(\itermone_1,\ldots,\itermone_{\arone(\funsymone)}) 
       \midd \displaystyle{\sum_{\ivarone< \itermone}\itermtwo}\midd \displaystyle{\Dfscomb{\ivarone}{\itermone}{\itermtwo}{\itermthree}},
\end{align*}
where $\funsymone\in\carone$ and $\ivarone$ is a variable drawn from a set $\varsetone$ of \emph{index
variables}. We assume the symbols $\natit{0}$, $\natit{1}$ (with arity $0$) and $+$, $\mnu$ (with arity $2$) are always
part of $\sigone$. An index term in the form $\sum_{\ivarone< \itermone}\itermtwo$
is a \emph{bounded sum}, while one in the form $\Dfscomb{\ivarone}{\itermone}{\itermtwo}{\itermthree}$
is a \emph{forest cardinality}. For every natural number $\natone$, the index term $\natit{\natone}$ is
just 
$$
\underbrace{\natit{1}+\natit{1}+\ldots+\natit{1}}_{\mbox{$\natone$ times}}.
$$
Index terms are meant to denote natural numbers, possibly depending on the (unknown) values of variables.
Variables can be instantiated with other index terms, e.g. $\sb{\itermone}{\ivarone}{\itermtwo}$. 
So, index terms can also act as first order functions. 
 What is the meaning of the function symbols
from $\sigone$? It is the one induced by  an equational program $\eqpone$.
Formally, an \emph{equational program} $\eqpone$ over a signature $\sigone$ and
a set of variables $\varsetone$ is a set of equations in the form
$\ftermone=\ftermtwo$ where both $\ftermone$ and $\ftermtwo$
are terms in the free algebra $\oterms{\sigone}{\varsetone}$ over $\sigone$ and $\varsetone$. 
We are interested in equational programs guaranteeing that, whenever
symbols in $\sigone$ are interpreted as partial functions
over $\NN$ and $\natit{0}$, $\natit{1}$, $+$ and $\mnu$ are interpreted in the
usual way, the semantics of any function symbol 
$\funsymone$ can be uniquely determined from $\eqpone$.
This can be guaranteed by, for example, taking $\eqpone$ as
an Herbrand-G\"odel scheme~\cite{Odifreddi} or as an orthogonal constructor
term rewriting system~\cite{BaaderNipkow}. 
One may wonder why the definition of index terms
is parametric on $\sigone$ and $\eqpone$. As we will see in Section~\ref{sec:type-checking},
being parametric this way allows us to tune our concrete type system from
a highly undecidable but truly powerful machinery down to a tractable but
less expressive formal system. An example of an equational program over
the signature $\sigone$ consisting of three function symbols
$\ifgtz$, $\ifadd$ and $\ifmult$ of arity two is the following sequence of equations:
\begin{align*}
\ifgtz(\natit{0},\ivartwo)&=\natit{0};\\
\ifgtz(\ivarone+\natit{1},\natit{0})&=\natit{1};\\
\ifgtz(\ivarone+\natit{1},\ivartwo+\natit{1})&=\ifgtz(\ivarone,\ivartwo);\\[1mm]
\ifadd(\natit{0},\ivartwo)&=\ivartwo;\\
\ifadd(\ivarone+\natit{1},\ivartwo)&=\ifadd(\ivarone,\ivartwo)+\natit{1};\\[1mm]
\ifmult(\natit{0},\ivartwo)&=\natit{0};\\
\ifmult(\ivarone+\natit{1},\ivartwo)&=\ifadd(\ivartwo,\ifmult(\ivarone,\ivartwo)).
\end{align*}

What about the meaning of bounded sums and forest cardinalities? The first is very intuitive:
the value of $\sum_{\ivarone< \itermone}\itermtwo$ is simply the sum of all
possible values of $\itermtwo$ with $\ivarone$ taking the values from $0$ up to
$\itermone$, excluded. Forest cardinalities, on the other hand, require some more effort to be described.
 Informally, 
$\Dfscomb{\ivarone}{\itermone}{\itermtwo}{\itermthree}$ is an
index term denoting the number of nodes in a forest composed of $\itermtwo$ trees described using $\itermthree$.
All the nodes in the forest are (uniquely) identified by natural numbers. These are obtained by consecutively 
visiting each tree in pre-order, starting from $\itermone$. The term $\itermthree$ has the role of 
describing the number of children of each forest node $\natone$ by properly instantiating the variable 
$\ivarone$, e.g the number of children of the root (of the leftmost tree in the forest) is $\sb{\itermthree}{\ivarone}{\natit{0}}$. 
More formally, the meaning of a forest cardinality is defined by the following two equations:
\begin{align}
  \Dfscomb{\ivarone}{\itermone}{\natit{0}}{\itermthree}&=\natit{0}\label{eqn:fcbase};\\
  \Dfscomb{\ivarone}{\itermone}{\itermtwo+\natit{1}}{\itermthree}&=
    \left(\Dfscomb{\ivarone}{\itermone}{\itermtwo}{\itermthree}\right)+\natit{1}+
    \left(\Dfscomb{\ivarone}{\itermone+\natit{1}+\Dfscomb{\ivarone}{\itermone}{\itermtwo}{\itermthree}}
            {\sb{\itermthree}{\ivarone}{\itermone+\Dfscomb{\ivarone}{\itermone}{\itermtwo}{\itermthree}}}{\itermthree}\right).\label{eqn:fcind}
\end{align}
Equation~(\ref{eqn:fcbase}) says that a forest of $0$ trees contains no nodes. Equation~(\ref{eqn:fcind})
tells us that a forest of $\itermtwo+1$ trees contains:
\begin{varitemize} 
\item
  the nodes in the first $\itermtwo$ trees;
\item
  and the nodes in the last tree, which are just one plus the nodes in the immediate 
  subtrees of the root, considered themselves as a forest.
\end{varitemize}
To better understand forest cardinalities, consider the following forest comprising two trees:
$$
\scalebox{0.8}{
\xymatrix{
 & &  0          &   \\
 & & 1\ar@{-}[u] &   \\
\empty&2\ar@{-}[ur] &5\ar@{-}[u] & 6\ar@{-}[ul]\\
3\ar@{-}[ur]& &  4\ar@{-}[ul]&7\ar@{-}[u]
\save "4,1"."3,4"*[F.]\frm{--}\restore 
}

\xymatrix{
 & & 8 &   \\
& 9\ar@{-}[ur] & & 11\ar@{-}[ul]\\
& 10\ar@{-}[u] & & 12\ar@{-}[u] \\
}
}
$$
and consider an index term $\itermthree$ with a free index variable $\ivarone$ such that 
$\sb{\itermthree}{\ivarone}{\natit{\natone}}=3$ for $\natone=1$;
$\sb{\itermthree}{\ivarone}{\natit{\natone}}=2$ for $\natone\in\{2,8\}$;
$\sb{\itermthree}{\ivarone}{\natit{\natone}}=1$ when
$\natone\in\{0,6,9,11\}$;
and $\sb{\itermthree}{\ivarone}{\natit{\natone}}=0$ when
$\natone\in\{3,4,7,10,12\}$. That is, $\itermthree$ describes the number of children of each node in the forest. Then 
$\Dfscomb{\ivarone}{\natit{0}}{\natit{2}}{\itermthree}=\natit{13}$ since it takes into account the entire forest;
$\Dfscomb{\ivarone}{\natit{0}}{\natit{1}}{\itermthree}=\natit{8}$ since it takes into account only the leftmost tree;
$\Dfscomb{\ivarone}{\natit{8}}{\natit{1}}{\itermthree}=\natit{5}$ since it takes into account only the second tree of the forest;
finally, $\Dfscomb{\ivarone}{\natit{2}}{\natit{3}}{\itermthree}=\natit{6}$ since it takes into account only the three trees 
(as a forest) in the dashed rectangle.

One may wonder what is the role of forest cardinalities in the type system. Actually, 
they play a crucial role in the treatment of recursive calls, where the unfolding of recursion
produces a tree-like structure whose size is just the number of times the (recursively
defined) function will be used \emph{globally}. 
Note that the value of a forest cardinality could also be undefined. 
For instance, this happens when infinite trees, corresponding to 
diverging recursive computations, are considered.

The expression $\semt{\itermone}{\assone}{\eqpone}$ denotes the meaning of $\itermone$,
defined by induction along the lines of the previous discussion, 
where $\assone:\varsetone\arr\NN$ is an assignment and
$\eqpone$ is an equational program giving meaning to the function symbols in $\itermone$.
Since $\eqpone$ does not necessarily interpret such symbols as \emph{total} functions, and moreover,
the value of a forest cardinality can be undefined,
$\semt{\itermone}{\assone}{\eqpone}$ can be undefined itself. A \emph{constraint} is an
inequality in the form $\itermone\leq\itermtwo$. A constraint is \emph{true} in an
assignment $\assone$ if $\semt{\itermone}{\assone}{\eqpone}$ and
$\semt{\itermtwo}{\assone}{\eqpone}$ are both defined and the first is smaller
or equal to the latter. 
Now, for a subset $\icontextone$ of $\varsetone$, and for 
a set $\iconstraintone$ of constraints involving variables in $\icontextone$,
the expression
\begin{equation}\label{eqn:sem}
\icontextone;\iconstraintone\models^\eqpone\itermone\leq\itermtwo
\end{equation}
denotes the fact that the truth of $\itermone\leq\itermtwo$ semantically follows
from the truth of the constraints in $\iconstraintone$. 
The expression $\icontextone;\iconstraintone\models^\eqpone\itermone\geq\natit{0}$ indicates
that (the semantics of) $\itermone$ is \emph{defined} for the relevant values of the
variables in $\icontextone$; this is usually written as 
$\icontextone;\iconstraintone\models^\eqpone\conv{\itermone}$. 

Similarly, one can define the meaning of expressions like $\icontextone;\iconstraintone\models^\eqpone\itermone=\itermtwo$
or $\icontextone;\iconstraintone\models^\eqpone\itermone\kleq\itermtwo$, the latter standing for
the equality of $\itermone$ and $\itermtwo$ in the sense of Kleene,
i.e. 
$\icontextone;\iconstraintone\models^\eqpone\conv{\itermone}$ if and
only if 
$\icontextone;\iconstraintone\models^\eqpone\conv{\itermtwo}$, and
if  $\icontextone;\iconstraintone\models^\eqpone\conv{\itermone}$ then 
$\icontextone;\iconstraintone\models^\eqpone\itermone=\itermtwo$. 
When both
$\icontextone$ and $\iconstraintone$ are empty, such expressions can be written in a much more concise
form, e.g. $\itermone\kleq\itermtwo$ stands for $\emcon;\emcon\models^\eqpone\itermone\kleq\itermtwo$.

The following two lemmas about forest cardinalities are useful, and will be crucial when proving the Substitution Lemma.
\begin{lem}\label{lemma:dfsshift}
For every index terms $\itermone,\itermtwo,\itermthree,\itermfour$, we have:
$$
\Dfscomb{\ivarone}{\itermone+\itermtwo}{\itermthree}{\itermfour}\kleq\Dfscomb{\ivarone}{\itermtwo}{\itermthree}{\sb{\itermfour}{\ivarone}{\ivarone+\itermone}}.
$$
\end{lem}
\begin{proof}
The proof is  by coinduction on the definition of $\Dfscomb{\ivarone}{\itermone+\itermtwo}{\itermthree}{\itermfour}$ by distinguishing the cases for the different values of $\itermthree$. For $\itermthree\kleq \natit{0}$ we have both:
$$
\Dfscomb{\ivarone}{\itermone+\itermtwo}{\natit{0}}{\itermfour}\kleq \natit{0};\qquad
\Dfscomb{\ivarone}{\itermtwo}{\natit{0}}{\sb{\itermfour}{\ivarone}{\ivarone+\itermone}}\kleq \natit{0}.
$$
For  $\itermthree\kleq\itermfive+\natit{1}$ we have:
$$
\Dfscomb{\ivarone}{\itermone+\itermtwo}{\itermfive+\natit{1}}{\itermfour}\kleq \Dfscomb{\ivarone}{\itermone+\itermtwo}{\itermfive}{\itermfour}+
   \natit{1}+\Dfscomb{\ivarone}{\itermone+\itermtwo+\natit{1}+\Dfscomb{\ivarone}{\itermone+\itermtwo}{\itermfive}{\itermfour}}
       {\sb{\itermfour}{\ivarone}{\itermone+\itermtwo+\Dfscomb{\ivarone}{\itermone+\itermtwo}{\itermfive}{\itermfour}}}{\itermfour},
$$
and analogously
$$
\Dfscomb{\ivarone}{\itermtwo}{\itermfive+\natit{1}}{\sb{\itermfour}{\ivarone}{\ivarone+\itermone}}\kleq
   \Dfscomb{\ivarone}{\itermtwo}{\itermfive}{\sb{\itermfour}{\ivarone}{\ivarone+\itermone}}+\natit{1}+
   \Dfscomb{\ivarone}{\itermtwo+1+\Dfscomb{\ivarone}{\itermtwo}{\itermfive}{\sb{\itermfour}{\ivarone}{\ivarone+\itermone}}}
       {\sb{\itermfour}{\ivarone}{\itermone+\itermtwo+\Dfscomb{\ivarone}{\itermtwo}{\itermfive}{\sb{\itermfour}{\ivarone}
             {\ivarone+\itermone}}}}{\sb{\itermfour}{\ivarone}{\ivarone+\itermone}}.
$$
This concludes the proof.
\end{proof}

\begin{lem}\label{lemma:dfssum}
For every index term of the shape $\Dfscomb{\ivarone}{\natit{1}}{\itermtwo}{\itermone}$ we have:
$$
\Dfscomb{\ivarone}{\natit{1}}{\itermtwo}{\itermone}\kleq\sum_{\ivartwo<\itermtwo}
\Dfscomb{\ivarone}{\natit{0}}{\natit{1}}{\sb{\itermone}{\ivarone}{\ivarone+\natit{1}+\Dfscomb{\ivarone}{\natit{1}}{\ivartwo}{\itermone}}}.
$$
\end{lem}
\begin{proof}
The proof is  by coinduction on the definition of $\Dfscomb{\ivarone}{\natit{1}}{\itermtwo}{\itermone}$ by distinguishing the 
cases for the different values of $\itermtwo$. For $\itermtwo\kleq\natit{0}$, we have both:
$$
\Dfscomb{\ivarone}{\natit{1}}{\natit{0}}{\itermone}\kleq \natit{0};
\qquad
\sum_{\ivartwo<\natit{0}}\Dfscomb{\ivarone}{\natit{0}}{\natit{1}}{\sb{\itermone}{\ivarone}{\ivarone+\natit{1}+\Dfscomb{\ivarone}{\natit{1}}{\ivartwo}{\itermone}}}\kleq \natit{0}.
$$
For $\itermtwo\kleq\itermfive+\natit{1}$ we have
$$
\Dfscomb{\ivarone}{\natit{1}}{\itermfive+\natit{1}}{\itermone}\kleq\itermthree+\natit{1}+\Dfscomb{\ivarone}{\itermthree+\natit{2}}{\sb{\itermone}{\ivarone}
  {\itermthree+\natit{1}}}{\itermone}
$$
and 
$$
\sum_{\ivartwo<\itermfive+\natit{1}}\Dfscomb{\ivarone}{\natit{0}}{\natit{1}}{\sb{\itermone}{\ivarone}{\ivarone+\natit{1}+\Dfscomb{\ivarone}{\natit{1}}{\ivartwo}{\itermone}}}
  \kleq\itermfour+\Dfscomb{\ivarone}{\natit{0}}{\natit{1}}{\sb{\itermone}{\ivarone}{\ivarone+\natit{1}+\Dfscomb{\ivarone}{\natit{1}}{\itermfive}{\itermone}}},
$$
where $\itermthree$ is $\Dfscomb{\ivarone}{\natit{1}}{\itermfive}{\itermone}$ 
and $\itermfour$ is $\sum_{\ivartwo<\itermfive}\Dfscomb{\ivarone}{\natit{0}}{\natit{1}}{\sb{\itermone}{\ivarone}
{\ivarone+\natit{1}+\Dfscomb{\ivarone}{\natit{1}}{\ivartwo}{\itermone}}}$.
Now, by definition and by Lemma~\ref{lemma:dfsshift}, we have 
$$
  \Dfscomb{\ivarone}{\natit{0}}{\natit{1}}{\sb{\itermone}{\ivarone}{\ivarone+\natit{1}+\Dfscomb{\ivarone}{\natit{1}}{\itermfive}{\itermone}}}\kleq
    \natit{1}+\Dfscomb{\ivarone}{\natit{1}}{\sb{\itermone}{\ivarone}{\itermthree+\natit{1}}}{\sb{\itermone}{\ivarone}{\ivarone+\natit{1}+\itermthree}}
  \kleq \natit{1}+\Dfscomb{\ivarone}{\itermthree+\natit{2}}{\sb{\itermone}{\ivarone}{\itermthree+\natit{1}}}{\itermone}.
$$
This concludes the proof.
\end{proof}

Before embarking in the description of the type system, 
a further remark on the role of
index terms could be useful. 
Index terms are not meant to be part
of \emph{programs} but of \emph{types}. As a consequence, computation will not be
carried out on index terms but on proper terms, which are the subject of 
Section~\ref{sect:tysys} below.
\subsection{The Type System}\label{sect:tysys}
Terms are generated by the following grammar:
\begin{align*}
\termone ::=& \varone \midd \val{n} \midd \suc(\termone)  \midd \pred(\termone) \midd \lambda \varone.\termone\midd\termone\termtwo\\
          &\midd\CASE{\termone}{\termtwo}{\termthree} \midd\REC{\varone}{\termone}
\end{align*}
where $\tt n$ ranges over natural numbers and $\varone$ ranges over a set of \emph{variables}.
As usual, terms which are equal modulo $\alpha$-conversion are considered equal. This, in turn,
allows to define the notion of substitution in the standard way. The set of \emph{head subterms} of any term 
$\termone$ can be defined easily by induction on the structure of $\termone$,
e.g. the head subterms of $\termone=\termtwo\termthree$ are $\termone$ itself and the head subterms of $\termtwo$
(but not those of $\termthree$).

A notion of \emph{size} $\size{\termone}$ 
for a term $\termone$ will be useful in the sequel. This can be defined as follows:
\begin{align*}
\size{\varone}&=1;&\size{\lambda \varone.\termone}&=\size{\termone}+1;\\
\size{\val{\natone}}&=1;&\size{\termone\termtwo}&=\size{\termone}+\size{\termtwo}+1;\\ 
\size{\suc(\termone)}&=\size{\termone}+2;&\size{\CASE{\termone}{\termtwo}{\termthree}}&=\size{\termone}+\size{\termtwo}+\size{\termthree}+1;\\
\size{\pred(\termone)}&=\size{\termone}+2;&\size{\REC{\varone}{\termone}}&=\size{\termone}+1.
\end{align*}
Notice that for technical reasons size is defined in a slightly nonstandard way: every
integer constant has size $1$.
\begin{lem}
If $\termone$ is a term and $\termtwo$ is a subterm of $\termone$, then $\size{\termtwo}\leq\size{\termone}$.
\end{lem}
Terms can be typed by a well-known type system called \PCF.
Types are those generated by the basic type $\Nat$ and the binary type constructor $\arr$. Typing
rules are in Figure~\ref{fig:pcf}. 
\begin{figure*}
\fbox{
\begin{minipage}[c]{.97\textwidth}
\begin{center}
$$
\begin{array}{ccccc}
\infer
  {\tcontextone,\varone:\typeone\vdash\varone:\typeone}{}
&
\hspace{5pt}
&
\infer
  {\tcontextone \vdash\lambda \varone.\termone:\typeone\arr\typetwo}
  {\tcontextone,\varone:\typeone\vdash\termone:\typetwo}
&
\hspace{5pt}
&
\infer
  {\tcontextone\vdash\termone\termtwo:\typetwo}
  {\tcontextone\vdash\termone:\typeone\arr\typetwo & \tcontextone\vdash\termtwo:\typeone}
\end{array}
$$
$$
\begin{array}{ccccc}
\infer
  {\tcontextone\vdash\val{n}:\Nat}{}
&
\hspace{5pt}
&
\infer
  {\tcontextone\vdash\suc(\termone):\Nat}
  {\tcontextone\vdash\termone:\Nat}
&
\hspace{5pt}
&
\infer
  {\tcontextone\vdash\pred(\termone):\Nat}
  {\tcontextone\vdash\termone:\Nat}
\end{array}
$$
$$
\begin{array}{ccc}
\infer
  {\tcontextone\vdash\CASE{\termone}{\termtwo}{\termthree}:\typeone}
  {
      \tcontextone\vdash\termone: \Nat
      &
      \tcontextone\vdash\termtwo: \typeone
      &
      \tcontextone\vdash\termthree: \typeone
  }
&
\hspace{5pt}
&
\infer
  {\tcontextone\vdash\REC{\varone}{\termone}:\typeone}   
  {\tcontextone,\varone:\typeone\vdash\termone:\typeone}
\end{array}
$$
\vspace{0pt}  
\end{center}
\end{minipage}}
\caption{The \PCF\ Type System.}\label{fig:pcf}
\end{figure*}
A notion of weak-head reduction $\to$ can be easily defined: see Figure~\ref{fig:whred}.
\begin{figure*}
\fbox{
\begin{minipage}[c]{.97\textwidth}
\begin{center}
$$
\begin{array}{ccccccc}
\infer
  {(\lambda\varone.\termone)\termtwo\to\sb{\termone}{\varone}{\termtwo}}{}
&
\hspace{5pt}
&
\infer
  {\suc(\val{\natone})\to \val{\natone+1}}{}
&
\hspace{5pt}
&
\infer
  {\pred(\val{\natone+1})\to \val{\natone}}{}
&
\hspace{5pt}
&
\infer
  {\pred(\val{0})\to \val{0}}{}
\end{array}
$$
$$
\begin{array}{ccccc}
\infer
  {\CASE{\val{0}}{\termtwo}{\termthree} \to \termtwo}{}
&
\hspace{5pt}
&
\infer
  {\CASE{\val{n+1}}{\termtwo}{\termthree} \to \termthree}{}
\end{array}
$$
$$
\begin{array}{ccccccc}
\infer
  {\REC{\varone}{\termone} \to \sb{\termone}{\varone}{\REC{\varone}{\termone}}}{}
&
\hspace{5pt}
&
\infer
  {\suc(\termone)\to\suc(\termtwo)}
  {\termone\to\termtwo}
&
\hspace{5pt}
&
\infer
  {\pred(\termone)\to\pred(\termtwo)}
  {\termone\to\termtwo}
&
\hspace{5pt}
&
\infer
  {\termone\termtwo\to\termthree\termtwo}
  {\termone\to\termthree}
\end{array}
$$
$$
\infer
  {\CASE{\termone}{\termtwo}{\termthree}\to\CASE{\termfour}{\termtwo}{\termthree}}
  {\termone\to\termfour}
$$
\end{center}
\vspace{3pt}
\end{minipage}}
\caption{Weak-head Reduction}\label{fig:whred}
\end{figure*}
A term $\termone$ is said to be a \emph{program} if it can be given the \PCF\ type $\Nat$ in the empty context. 

Almost all the definitions about \PCFld\ in this and the next sections should be understood as parametric on an 
equational program $\eqpone$ over a signature $\sigone$. For the sake of simplicity, however, we will often 
avoid to explicitly mention $\eqpone$ and leave it implicit.

\PCFld\ can be seen as a refinement of \PCF\ obtained by a linear decoration of its type derivations. 
Basic and modal types are defined as follows:
\begin{align*}
\typeone,\typetwo &::= \Nat[\itermone,\itermtwo] \midd\mtypeone\lin\typeone; &\mbox{basic types}\\
\mtypeone,\mtypetwo &::= \qbang{\ivarone<\itermone}{\typeone}; &\mbox{modal types}
\end{align*}
where $\itermone,\itermtwo$ range over index terms and $\ivarone$ ranges over index variables.
$\Nat[\itermone]$ is syntactic sugar for $\Nat[\itermone,\itermone]$.
Modal types need some comments. As a first approximation, they can be thought of as 
quantifiers over type variables.
So, a type like $\mtypeone=\qbang{\ivarone<\itermone}{\typeone}$ acts as a binder for the 
index variable $\ivarone$ in the basic type $\typeone$. 
Moreover, the condition $\ivarone<\itermone$ says that 
$\mtypeone$ consists of all the instances of the basic type $\typeone$ where the variable $\ivarone$ 
is successively instantiated with the values from $\natit{0}$ to (the value of) $\itermone\mnu\natit{1}$, 
i.e. $\sb{\typeone}{\ivarone}{\natit{0}},\ldots,\sb{\typeone}{\ivarone}{\itermone\mnu\natit{1}}$ .
For those readers who are familiar with linear logic, 
and in particular with \BLL, the modal type
 $\qbang{\ivarone<\itermone}{\typeone}$ is a generalization of the \BLL\ formula
$!_{\ivarone< p}\typeone$ to arbitrary index terms. As such it can be thought of as representing the type 
$\sb{\typeone}{\ivarone}{\natit{0}}\otimes\cdots\otimes\sb{\typeone}{\ivarone}{\itermone\mnu\natit{1}}$.
In analogy to what happens in the standard linear logic decomposition of the 
intuitionistic arrow, i.e. $!A\lin B=A\Rightarrow B$, it is sufficient to restrict
to modal types appearing in negative position.
Finally, for those readers with some knowledge of \DML, modal types
are in a way similar to \DML's subset sort constructions~\cite{DBLP:journals/jfp/Xi07}.

We always assume that index terms appearing inside types are defined for all the
relevant values of the variables in $\icontextone$. This is captured by 
the judgement $\icontextone;\iconstraintone\vdash^\eqpone\conv{\typeone}$, whose
rules are in Figure~\ref{fig:wdtypes}.
\begin{figure*}
\fbox{
\begin{minipage}[c]{.97\textwidth}
\begin{center}
$$
\begin{array}{ccc} 
\infer[\nattdef]
  {\icontextone;\iconstraintone\vdash^{\eqpone}\conv{\Nat[\itermone,\itermtwo]}}
  {
    \begin{array}{c}
      \icontextone;\iconstraintone\models^\eqpone\conv{\itermone}\\
      \icontextone;\iconstraintone\models^\eqpone\conv{\itermtwo}\\
    \end{array}
  }
&
\hspace{10pt}
&
\infer[\lintdef]
  {\icontextone;\iconstraintone\vdash^{\eqpone}\conv{\mtypeone\lin\typeone}}
  {
    \begin{array}{c}
      \icontextone;\iconstraintone\vdash^{\eqpone}\conv{\mtypeone}\\
      \icontextone;\iconstraintone\vdash^{\eqpone}\conv{\typeone}
    \end{array}
  }
\end{array}
$$
\vspace{5pt}
$$
\infer[\qbtdef]
  {\icontextone;\iconstraintone\vdash^{\eqpone}\conv{\qbang{\ivarone<\itermone}{\typeone}}}
  {
    \begin{array}{c}
      \icontextone,\ivarone;\iconstraintone,\ivarone<\itermone\vdash^{\eqpone}\conv{\typeone}\\
      \icontextone;\iconstraintone\models^\eqpone\conv{\itermone}\\
    \end{array}
  }
$$
\vspace{0pt}
\end{center}
\end{minipage}}
\caption{Well-defined Types}\label{fig:wdtypes}
\end{figure*}

In the typing rules, modal types need to be manipulated in an algebraic way. 
For this reason, two operations on modal types need to be introduced.
The first one is a binary operation $\mtsums$ on modal types. Suppose
that $\mtypeone=\qbang{\ivarone<\itermone}{\sb{\typethree}{\ivarthree}{\ivarone}}$
and that $\mtypetwo=\qbang{\ivartwo<\itermtwo}{\sb{\typethree}{\ivarthree}{\itermone+\ivartwo}}$.
In other words, $\mtypeone$ consists of the first $\itermone$ instances of $\typethree$,
i.e. $\sb{\typethree}{\ivarthree}{\natit{0}},\ldots,\sb{\typethree}{\ivarthree}{\itermone\mnu\natit{1}}$
while $\mtypetwo$ consists of the next $\itermtwo$ instances of $\typethree$,
i.e. $\sb{\typethree}{\ivarthree}{\itermone+\natit{0}},\ldots,\sb{\typethree}{\ivarthree}{\itermone+\itermtwo\mnu\natit{1}}$. 
Their \emph{sum} $\mtsum{\mtypeone}{\mtypetwo}$ is naturally defined
as a modal type consisting of the first $\itermone+\itermtwo$ instances of $\typethree$, i.e.
$\qbang{\ivarthree<\itermone+\itermtwo}{\typethree}$.
An operation of bounded sum on modal types can be defined by generalizing the idea above.
Suppose that $\mtypeone=\qbang{\ivartwo< \itermtwo}{\sb{\typeone}{\ivarthree}{\sum_{\ivarfour<\ivarone}
\sb{\itermtwo}{\ivarone}{\ivarfour}+\ivartwo}}$.
Then its \emph{bounded sum} $\sum_{\ivarone<\itermone}\mtypeone$
is $\qbang{\ivarthree< \sum_{\ivarone< \itermone}\itermtwo }{\typeone}$.

To every type $\typeone$ corresponds a type $\TtoNDT{\typeone}$ of ordinary \PCF, namely a type
built from the basic type $\Nat$ and the arrow operator $\arr$:
\begin{align*}
\TtoNDT{\Nat[\itermone,\itermtwo]}&=\Nat;\\
\TtoNDT{\qbang{\ivarone<\itermone}{\typeone}\lin\typetwo}&=\TtoNDT{\typeone}\arr\TtoNDT{\typetwo}.
\end{align*}

Central to \PCFld\ is the notion of subtyping.
An inequality relation $\tless$ between (basic and modal) types can be defined by way of the formal system
in Figure~\ref{fig:tsub}. This relation corresponds to lifting index inequalities at the type level.
\begin{figure*}
\fbox{
\begin{minipage}[c]{.97\textwidth}
\begin{center}
$$
\begin{array}{ccccc} 
\infer[\natleq]
  {\icontextone;\iconstraintone\vdash^{\eqpone}\Nat[\itermone,\itermtwo]\tless\Nat[\itermthree,\itermfour]}
  {
    \begin{array}{c}
      \icontextone;\iconstraintone\models^\eqpone\itermthree\leq\itermone\\
      \icontextone;\iconstraintone\models^\eqpone\itermtwo\leq\itermfour\\
    \end{array}
  }
&
\hspace{0pt}
&
\infer[\linleq]
  {\icontextone;\iconstraintone\vdash^{\eqpone}\mtypeone\lin\typeone\tless\mtypetwo\lin\typetwo}
  {
    \begin{array}{c}
      \icontextone;\iconstraintone\vdash^{\eqpone}\mtypetwo\tless\mtypeone\\
      \icontextone;\iconstraintone\vdash^{\eqpone}\typeone\tless\typetwo
    \end{array}
  }
&
\end{array}
$$
$$
\infer[\qbleq]
  {\icontextone;\iconstraintone\vdash^{\eqpone}\qbang{\ivarone<\itermone}{\typeone}\tless\qbang{\ivarone<\itermtwo}{\typetwo}}
  {
    \begin{array}{c}
      \icontextone,\ivarone;\iconstraintone,\ivarone<\itermone\vdash^{\eqpone}\typeone\tless\typetwo\\
      \icontextone;\iconstraintone\models^\eqpone\itermtwo\leq\itermone\\
    \end{array}
  }
$$
\end{center}
\vspace{2pt}
\end{minipage}}
\caption{The Subtyping Relation}\label{fig:tsub}
\end{figure*}
The equivalence $\icontextone;\iconstraintone\vdash\typeone\teq\typetwo$ holds
when both $\icontextone;\iconstraintone\vdash\typeone\tless\typetwo$
and $\icontextone;\iconstraintone\vdash\typetwo\tless\typeone$
can be derived from the rules in Figure~\ref{fig:tsub}. $\icontextone;\iconstraintone\vdash\conv{\typeone}$
is syntactic sugar for $\icontextone;\iconstraintone\vdash\typeone\tless\typeone$.

It is now time to introduce the main object of this paper, namely the type system
\PCFld. \emph{Typing judgements} of \PCFld\ are expressions in the form
\begin{equation}\label{eq:tj}
\icontextone;\iconstraintone;\tcontextone\vdash_{\itermone}^{\eqpone}\termone:\typeone,
\end{equation}
where $\tcontextone$ is a \emph{typing context}, that is, 
a set of term variable assignments of the shape $\varone:\mtypeone$ where each variable $\varone$ occurs at most once. 
The expression (\ref{eq:tj}) can be informally read as follows: for every
values of the index variables in $\icontextone$ satisfying $\iconstraintone$,
$\termone$ can be given type $\typeone$ and \emph{cost} $\itermone$ once its free term variables have types
as in $\tcontextone$. In proving this, equations from $\eqpone$ can play a role.

Typing rules are in Figure~\ref{fig:trule}, where binary and bounded sums are used in their
natural generalization to contexts. A \emph{type derivation} is nothing more than a tree built
according to typing rules. A \emph{precise type derivation} is a type derivation such that all premises
in the form $\typeone\tless\typetwo$ (respectively, in the form $\itermone\leq\itermtwo$) are required
to be in the form $\typeone\cong\typetwo$ (respectively, $\itermone=\itermtwo$).
\begin{figure*}
\fbox{
\begin{minipage}[c]{.97\textwidth}
\begin{center}
$$
\begin{array}{ccc}
\infer[\Axty]
  {\icontextone;\iconstraintone;\tcontextone,\varone:\qbang{\ivarone<\itermone}{\typeone}\vdash_{\itermtwo}^{\eqpone} \varone:\typetwo}
  {
    \begin{array}{c}
      \icontextone;\iconstraintone\models^{\eqpone} \natit{0}\leq \itermtwo \qquad
      \icontextone;\iconstraintone\models^{\eqpone} \natit{1} \leq \itermone\\
      \icontextone;\iconstraintone\vdash^{\eqpone}\sb{\typeone}{\ivarone}{\natit{0}}\tless{\typetwo}\\
      \icontextone;\iconstraintone\vdash^{\eqpone}\conv{(\qbang{\ivarone<\itermone}{\typeone})}\qquad
      \icontextone;\iconstraintone\vdash^{\eqpone}\conv{\tcontextone}
    \end{array}
  }
&&
\infer[\Lamty]
  {\icontextone;\iconstraintone;\tcontextone \vdash^{\eqpone}_{\itermtwo} \lambda \varone.\termone:\qlin{\ivarone<\itermone}{\typeone}{\typetwo}}
  {
    \icontextone;\iconstraintone;\tcontextone,\varone :\qbang{\ivarone<\itermone}{\typeone}\vdash^{\eqpone}_{\itermtwo} \termone:\typetwo
  }
\end{array}
$$
\vspace{7pt}
$$
\begin{array}{ccc}
\infer[\Apty]
  {\icontextone;\iconstraintone;\tcontexthree\vdash^{\eqpone}_{\itermfour} \termone\termtwo:\typetwo}
  {
    \begin{array}{c}
      \icontextone;\iconstraintone;\tcontextone\vdash^{\eqpone}_{\itermtwo}\termone:\qlin{\ivarone<\itermone}{\typeone}{\typetwo}\\
      \icontextone,\ivarone;\iconstraintone,\ivarone<\itermone ;\tcontexttwo\vdash^{\eqpone}_{\itermthree} \termtwo:\typeone\\
      \icontextone;\iconstraintone\vdash^{\eqpone}\tcontexthree \tless \tcontextone\uplus\sum_{\ivarone<\itermone}{\tcontexttwo}\\
      \icontextone;\iconstraintone\models^{\eqpone}\itermfour\geq\itermtwo+\itermone+{\sum_{\ivarone<\itermone}\itermthree}
    \end{array}
  }
&
\hspace{7pt}
&
\infer[\Sucty]
  {\icontextone;\iconstraintone;\tcontextone \vdash^{\eqpone}_{\itermfive} \suc(\termone) :\Nat[\itermthree,\itermfour]}
  {
  \begin{array}{c}
    \icontextone;\iconstraintone\vdash^{\eqpone}\Nat[\itermone+\natit{1},\itermtwo+\natit{1}]\tless\Nat[\itermthree,\itermfour]\\
    \icontextone;\iconstraintone;\tcontextone\vdash^{\eqpone}_{\itermfive} \termone: \Nat[\itermone,\itermtwo]
  \end{array}
  }
\end{array}
$$
\vspace{7pt}
$$
\begin{array}{ccc}
\infer[\Natty]
  {\icontextone;\iconstraintone;\tcontextone\vdash^{\eqpone}_{\itermthree} \val{n}:\Nat[\itermone,\itermtwo]}
  { \begin{array}{c}
      \icontextone;\iconstraintone\models^{\eqpone} \itermthree\geq \natit{0}\\
      \icontextone;\iconstraintone\models^{\eqpone} \itermone\leq\natit{n}\\
      \icontextone;\iconstraintone\models^{\eqpone} \natit{n}\leq \itermtwo\\
      \icontextone;\iconstraintone\vdash^{\eqpone}\conv{\tcontextone}
    \end{array}}
&&
\infer[\Predty]
  {\icontextone;\iconstraintone;\tcontextone \vdash^{\eqpone}_{\itermfive} \pred(\termone) :\Nat[\itermthree,\itermfour]}
  {
  \begin{array}{c}
    \icontextone;\iconstraintone\vdash^{\eqpone}\Nat[\itermone\mnu \natit{1},\itermtwo\mnu \natit{1}]\tless \Nat[\itermthree,\itermfour]\\
    \icontextone;\iconstraintone;\tcontextone\vdash^{\eqpone}_{\itermfive} \termone: \Nat[\itermone,\itermtwo]
  \end{array}
  }
\end{array}
$$
\vspace{7pt}
$$
\infer[\Ifty]
  {\icontextone;\iconstraintone;\tcontexthree \vdash^{\eqpone}_{\itermfive}\CASE{\termone}{\termtwo}{\termthree}:\typeone}
  {
    \begin{array}{c}
      \icontextone;\iconstraintone;\tcontextone\vdash^{\eqpone}_{\itermthree} \termone:\Nat[\itermone,\itermtwo] \\
      \icontextone;\iconstraintone, \itermone\leq\zero;\tcontexttwo\vdash^{\eqpone}_{\itermfour} \termtwo: \typeone \\
      \icontextone;\iconstraintone,\itermtwo\geq \natit{1};\tcontexttwo\vdash^{\eqpone}_{\itermfour} \termthree: \typeone \\
      \icontextone;\iconstraintone\vdash^{\eqpone}\tcontexthree \tless \tcontextone\uplus\tcontexttwo\\
      \icontextone;\iconstraintone\models^{\eqpone}\itermfive\geq\itermthree+\itermfour
    \end{array}
  }
$$
\vspace{7pt}
$$
\infer[\Recty]
  {\icontextone;\iconstraintone;\tcontexthree\vdash^{\eqpone}_{\itermseven} \REC{\varone}{\termone}:\typethree}
  {
    \begin{array}{c}
      \icontextone,\ivartwo;\iconstraintone,\ivartwo<\itermfive;\tcontextone,
        \varone:\qbang{\ivarone<\itermone}{\typeone}\vdash^{\eqpone}_{\itermthree}\termone:\typetwo\\
      \icontextone;\iconstraintone\vdash^{\eqpone}\sb{\typetwo}{\ivartwo}{\natit{0}}\tless\typethree\\
      \icontextone,\ivarone,\ivartwo;\iconstraintone,\ivarone<\itermone,\ivartwo<\itermfive
        \vdash^{\eqpone}\sb{\typetwo}{\ivartwo}{\Dfscomb{\ivartwo}{\ivartwo+\natit{1}}{\ivarone}{\itermone}+\ivartwo+\natit{1}}\tless\typeone\\
        \icontextone;\iconstraintone\vdash^{\eqpone} \tcontexthree\tless\sum_{\ivartwo<\itermfive}\tcontextone\\
      \icontextone;\iconstraintone\models^\eqpone\Dfscomb{\ivartwo}{\natit{0}}{\natit{1}}{\itermone}\leq\itermfive,\itermsix\\
      \icontextone;\iconstraintone\models^\eqpone\itermseven\geq\itermsix\mnu \natit{1}+\sum_{\ivartwo<\itermfive}\itermthree
    \end{array}
  }
$$    
\vspace{0pt}
\end{center}
\end{minipage}}
\caption{Typing Rules}\label{fig:trule}
\end{figure*}

First of all, observe that the typing rules are syntax-directed: given a term $\termone$, all type
derivations for $\termone$ end with the same typing rule, namely the one corresponding to the
last syntax rule used in building $\termone$. In particular, 
no explicit subtyping rule exists, but subtyping is applied to the context in every rule. 
A syntax-directed type system offers a key advantage:
it allows one to prove the statements about type derivations by induction on the structure
of terms. This greatly simplifies the proof of crucial properties like subject reduction.

Typing rules have premises of three different kinds:
\begin{varitemize}
\item
  Of course, typing a term requires typing its immediate subterms, so
  typing judgements can be rule premises.
\item
  As just mentioned, typing rules allow to subtype the context $\tcontextone$,
  so subtyping judgements can be themselves rule premises.
\item
  The application of typing rules (and also of subtyping rules, see Figure~\ref{fig:tsub}) sometimes depends on the truth of
  some inequalities between index terms in the model induced by
  $\eqpone$.
\end{varitemize}
As a consequence, typing rules can only be applied if some relations
between index terms are  consequences of the constraints in $\iconstraintone$. These assumptions
have a semantic nature, but could of course be verified by any sound formal system. 
Completeness (see Section \ref{sec:relative-completeness}), 
however, only holds if all \emph{true} inequalities can be used as assumptions. 
As a consequence, type inference but also type (derivation) checking are bound to be problematic from a computational
point of view. See Section~\ref{sec:type-checking} for a more thorough discussion on this issue.

As a last remark, note that each rule can be seen as a \emph{decoration} of a rule of ordinary \PCF. More: for every
\PCFld\ type derivation $\tdone$ of $\icontextone;\iconstraintone;\tcontextone\vdash_\itermone^\eqpone\termone:\typeone$
there is a structurally identical derivation in \PCF\ for the same term, i.e. a derivation 
$\TtoNDT{\tdone}\prov\TtoNDT{\tcontextone}\vdash\termone:\TtoNDT{\typeone}$. 
\subsection{An Example}
In this section, we will show how \PCFld\ can give a sensible type to the example
we talked about in the Introduction, namely 
$$
\termdbl=\REC{\funone}{\lambda\varone.\CASE{\varone}{\val{0}}{\suc(\suc(\funone(\pred(\varone))))}}.
$$
First, let us take a look at a subterm of $\termdbl$, namely
$\termone=\CASE{\varone}{\val{0}}{\suc(\suc(\funone(\pred(\varone))))}$.
In plain \PCF, $\termone$ receives the type $\Nat$ in an environment
where $\varone$ has type $\Nat$ and $\funone$ has type $\Nat\arr\Nat$. Presumably,
a \PCFld\ type for $\termone$ can be obtained by decorating in an appropriate way
the type above. In other words, we are looking for a type derivation with
conclusion:
$$
\icontextone;\iconstraintone;\varone:\qbang{\ivarone<\itermone}{\Nat[\itermtwo]},
  \funone:\qbang{\ivartwo<\itermthree}{(\qbang{\ivarthree<\itermfour}\Nat[\itermfive]\lin\Nat[\itermsix])}
  \vdash^\eqpone_\itermseven \termone:\Nat[\itermeight].
$$
But how should we proceed? What we would like, at the end of the day, is being able to
describe how the value of $\termone$ depends on the value of $\varone$, so we could
look for a type derivation in this form:
$$
\ivarfour;\emcon;\varone:\qbang{\itermone}{\Nat[\ivarfour]},
  \funone:\qbang{\ivartwo<\itermthree}{(\qbang{\itermfour}\Nat[\ivarfour\mnu\natit{1}]\lin\Nat[\natit{2}(\ivarfour\mnu\natit{1})])}
  \vdash^\eqpone_\itermseven \termone:\Nat[\natit{2}\ivarfour],
$$
where $\qbangwa{\ivarone<\itermone}$ (respectively, $\qbangwa{\ivarthree<\itermfour}$) 
has been abbreviated into $\qbangwa{\itermone} $ (respectively, $\qbangwa{\itermfour}$) 
because the bound variable $\ivarone$ (respectively, $\ivarthree$) does not appear free in
the underlying type. But how to give values to $\itermone$, $\itermthree$, and $\itermfour$?
One could be tempted to define $\itermone$ simply as $\natit{2}$, since there are two occurrences
of $\varone$ in $\termone$. However, in view of the role played by $\varone$ and $\funone$
in $\termdbl$, $\itermone$ should be rather defined taking into account the number of times
$\varone$ will be copied along the computation of $\termdbl$ on \emph{any} input. A good guess
could be, for example, $\ivarfour+\natit{1}$. Similarly, $\itermfour$ could be $\ivarfour$.
But how about $\itermthree$? How many times $\funone$ is used? If $\ivarfour=0$, then
$\funone$ is not called, while if $\ivarfour>\natit{0}$, the function is called once. In other
words, a guess for $\itermfour$ could be
$\ifgtz(\ivarfour,\natit{0})$. Here we use the infix notation
$>$ for the operator $\ifgtz$ just to improve readability. Let us now try to build a derivation
for 
$$
\ivarfour;\emcon;\varone:\qbang{\ivarfour+\natit{1}}{\Nat[\ivarfour]},
  \funone:\qbang{\ivarfour>\natit{0}}{(\qbang{\ivarfour}\Nat[\ivarfour\mnu\natit{1}]\lin\Nat[\natit{2}(\ivarfour\mnu\natit{1})])}
  \vdash^\eqpone_\natit{0} \termone:\Nat[\natit{2}\ivarfour].
$$
Actually, it has the following shape:

{\footnotesize
$$
\infer
  {\ivarfour;\emcon;\varone:\qbang{\ivarfour+\natit{1}}{\Nat[\ivarfour]},
  \funone:\qbang{\ivarfour>\natit{0}}{(\qbang{\ivarfour}\Nat[\ivarfour\mnu\natit{1}]\lin\Nat[\natit{2}(\ivarfour\mnu\natit{1})])}
  \vdash^\eqpone_\natit{0} \termone:\Nat[\natit{2}\ivarfour]}
  {
    \begin{array}{c}
    \tdone\prov\ivarfour;\emcon;\varone:\qbang{\natit{1}}{\Nat[\ivarfour]}
      \vdash^\eqpone_\natit{0} \varone:\Nat[\ivarfour]\\
    \tdtwo\prov\ivarfour;\ivarfour\leq \natit{0};\varone:\qbang{\ivarfour}{\Nat[\ivarfour]},
      \funone:\qbang{\ivarfour>\natit{0}}{(\qbang{\ivarfour}\Nat[\ivarfour\mnu\natit{1}]\lin\Nat[\natit{2}(\ivarfour\mnu\natit{1})])}
      \vdash^\eqpone_\natit{0} \val{\natit{0}}:\Nat[\natit{2}\ivarfour]\\
    \tdthree\prov\ivarfour;\ivarfour>\natit{0};\varone:\qbang{\ivarfour}{\Nat[\ivarfour]},
      \funone:\qbang{\ivarfour>\natit{0}}{(\qbang{\ivarfour}\Nat[\ivarfour\mnu\natit{1}]\lin\Nat[\natit{2}(\ivarfour\mnu\natit{1})])}
      \vdash^\eqpone_\natit{0} \suc(\suc(\funone(\pred(\varone)))):\Nat[\natit{2}\ivarfour]\\
    \end{array}
  }
$$
}where assignments to types in the form $\qbang{\natit{0}}{\typeone}$ have been omitted from
contexts. Now, $\tdone$ and $\tdtwo$ can be easily built, while $\tdthree$ requires a little effort:
it is the type derivation

{\footnotesize
$$
\infer
  {\ivarfour;\ivarfour>\natit{0};\varone:\qbang{\ivarfour}{\Nat[\ivarfour]},
      \funone:\qbang{\ivarfour>\natit{0}}{(\qbang{\ivarfour}\Nat[\ivarfour\mnu\natit{1}]\lin\Nat[\natit{2}(\ivarfour\mnu\natit{1})])}
      \vdash^\eqpone_\natit{0} \suc(\suc(\funone(\pred(\varone)))):\Nat[\natit{2}\ivarfour]}
  {
    \infer
        {\ivarfour;\ivarfour>\natit{0};\varone:\qbang{\ivarfour}{\Nat[\ivarfour]},
          \funone:\qbang{\ivarfour>\natit{0}}{(\qbang{\ivarfour}\Nat[\ivarfour\mnu\natit{1}]\lin\Nat[\natit{2}(\ivarfour\mnu\natit{1})])}
          \vdash^\eqpone_\natit{0} \suc(\funone(\pred(\varone))):\Nat[\natit{2}\ivarfour\mnu\natit{1}]}
        {
          \infer
              {\ivarfour;\ivarfour>\natit{0};\varone:\qbang{\ivarfour}{\Nat[\ivarfour]},
                \funone:\qbang{\ivarfour>\natit{0}}{(\qbang{\ivarfour}\Nat[\ivarfour\mnu\natit{1}]\lin\Nat[\natit{2}(\ivarfour\mnu\natit{1})])}
                \vdash^\eqpone_\natit{0} \funone(\pred(\varone)):\Nat[\natit{2}(\ivarfour\mnu\natit{1})]}
              {
                \begin{array}{c}
                \tdfour\prov\ivarfour;\ivarfour>\natit{0};
                      \funone:\qbang{\ivarfour>\natit{0}}{(\qbang{\ivarfour}\Nat[\ivarfour\mnu\natit{1}]\lin\Nat[\natit{2}(\ivarfour\mnu\natit{1})])}
                      \vdash^\eqpone_\natit{0} \funone:\qbang{\ivarfour}{\Nat[\ivarfour\mnu\natit{1}]}\lin\Nat[\natit{2}(\ivarfour\mnu\natit{1})]\\
                \tdfive\prov\ivarfour;\ivarfour>\natit{0};\varone:\qbang{\natit{1}}{\Nat[\ivarfour]}
                      \vdash^\eqpone_\natit{0} \pred(\varone):\Nat[\ivarfour\mnu\natit{1}]
                \end{array}
              }
        }
  }
$$
}where $\tdfour$ and $\tdfive$ are themselves easily definable. Summing up, $\termone$ can indeed be given the type
we wanted it to have. As a consequence, we can say that
$$
\ivarfour;\emcon;
  \funone:\qbang{\ivarfour>\natit{0}}{(\qbang{\ivarfour}\Nat[\ivarfour\mnu\natit{1}]\lin\Nat[\natit{2}(\ivarfour\mnu\natit{1})])}
  \vdash^\eqpone_\natit{0} \lambda\varone.\termone:\qbang{\ivarfour+\natit{1}}{\Nat[\ivarfour]}\lin\Nat[\natit{2}\ivarfour].
$$
However, we have only solved half of the problem, since the last step (namely typing the fixpoint) is definitely 
the most complicated. In particular, the rule $\Recty$ requires an index variable $\ivartwo$
which somehow ranges over all recursive calls. A different but related type can be given
to $\lambda\varone.\termone$, namely
$$
\begin{array}{c}
\ivarone,\ivartwo;\ivartwo<\ivarone+\natit{1};
  \funone:\qbang{\ivarone>\ivartwo}{(\qbang{\ivarone\mnu\ivartwo}\Nat[\ivarone\mnu\ivartwo\mnu\natit{1}]\lin\Nat[\natit{2}(\ivarone\mnu\ivartwo\mnu\natit{1})])}\\
\vdash^\eqpone_\natit{0} \lambda\varone.\termone:\qbang{\ivarone\mnu\ivartwo+\natit{1}}{\Nat[\ivarone\mnu\ivartwo]}\lin\Nat[\natit{2}(\ivarone\mnu\ivartwo)].
\end{array}
$$
By the way, this does not require rebuilding the entire type derivation (see the properties in the
forthcoming Section~\ref{sect:prop}). Let us now check whether the judgement above can be the
premise of the rule $\Recty$. Following the notation in the typing rule $\Recty$ we can stipulate that:
\begin{align*}
  \itermone&\equiv\ivarone>\ivartwo;\\
  \itermtwo&\equiv\ivarone;\\
  \itermthree&\equiv\natit{0};\\
  \itermfive&\equiv\ivarone+\natit{1};\\
\end{align*}
and
\begin{align*}
  \typeone&\equiv\qbang{\ivarone\mnu\ivartwo}{\Nat[\ivarone\mnu\ivartwo\mnu\natit{1}]\lin\Nat[\natit{2}(\ivarone\mnu\ivartwo\mnu\natit{1})]};\\
  \typetwo&\equiv\qbang{\ivarone\mnu\ivartwo+\natit{1}}{\Nat[\ivarone\mnu\ivartwo]\lin\Nat[\natit{2}(\ivarone\mnu\ivartwo)]};\\
  \typethree&\equiv\sb{\typetwo}{\ivartwo}{\natit{0}}\equiv\qbang{\ivarone+\natit{1}}{\Nat[\ivarone]\lin\Nat[\natit{2}\ivarone]};\\
  \conone&\equiv\tcontexthree\equiv\emcon.
\end{align*}
We can then conclude that, since $\ivarone<(\ivarone>\ivartwo)$ implies $\ivarone=\natit{0}$:
\begin{align*}
  \ivarone;\emcon&\models\Dfscomb{\ivartwo}{\natit{0}}{\natit{1}}{\itermone}=\ivarone+\natit{1}=\itermtwo;\\
  \ivarone,\ivartwo;\ivarone<(\ivarone>\ivartwo)&\models\Dfscomb{\ivartwo}{\ivartwo+\natit{1}}{\ivarone}{\itermone}=\natit{0};\\
  \ivarone;\emcon&\models\sb{\typetwo}{\ivartwo}{\Dfscomb{\ivartwo}{\ivartwo+\natit{1}}{\ivarone}{\itermone}+\ivartwo+\natit{1}}=\sb{\typetwo}{\ivartwo}{\ivartwo+\natit{1}}=\typeone;
\end{align*}
and, ultimately, that $\ivarone;\emcon;\emcon\vdash^\eqpone_{\ivarone}\termdbl:\typethree$.
\subsection{Properties}\label{sect:prop}
This section is mainly concerned with Subject Reduction. Subject Reduction will only be proved for closed terms, 
since the language is endowed with a weak notion of reduction and, as a consequence, reduction cannot happen in the scope of lambda abstractions.
The system \PCFld\ enjoys some nice properties that are both necessary intermediate steps towards proving subject
reduction and essential ingredients for proving soundness and relative completeness. These properties
permit to manipulate judgements being sure that derivability is preserved.

First of all, the constraints $\iconstraintone$ in a typing judgement
can be made stronger without altering the rest:
\begin{lem}[Constraint Strenghtening]
\label{model-strengthening-typing}
Let $\icontextone;\iconstraintone;\tcontextone\vdash_\itermone\termone:\typeone$
and $\icontextone;\iconstrainttwo\models^{\eqpone}
\iconstraintone$. Then,
$\icontextone;\iconstrainttwo;\tcontextone\vdash_\itermone\termone:\typeone$.
\end{lem}
\begin{proof}
  It follows easily by definition of $\icontextone;\iconstrainttwo\eqpless^{\eqpone} \iconstraintone$.
\end{proof}
Note that a sort of strengthening also holds for weights.
\begin{lem}[Weight Monotonicity]
\label{lem:weight-monotonicity}
Let $\icontextone;\iconstraintone;\tcontextone\vdash_\itermone\termone:\typeone$ and
$\icontextone;\iconstraintone\models^{\eqpone} \itermone\leq \itermtwo$. Then,
$\icontextone;\iconstraintone;\tcontextone\vdash_\itermtwo\termone:\typeone$.
\end{lem}
\begin{proof}
  It follows easily by induction on the derivation proving  
  $\icontextone;\iconstraintone;\tcontextone\vdash_\itermone\termone:\typeone$.
  In particular, observe that all rules altering the weight are designed in such a way 
  as to allow the latter to be lifted up.
\end{proof}
Whenever a parameter in a subtyping judgment needs to be specialized, we can simply substitute
it with an index term.
\begin{lem}[Index Term Substitution Respects Subtyping]
\label{lem:subst-respect-index-ineq-I2Type}
Let $\icontextone,\ivarone;\iconstraintone\vdash\gtypeone
\tless \gtypetwo$ and $\itermone$ be an index term.
Then, $\icontextone;\sb{\iconstraintone}{\ivarone}{\itermone},\iconstrainttwo\vdash \sb{\gtypeone}{\ivarone}{\itermone}
\tless \sb{\gtypetwo}{\ivarone}{\itermone}$ whenever $\icontextone;\iconstrainttwo\models\conv{\itermone}$.
\end{lem}
\begin{proof}
Easy.
\end{proof}
Subtyping can be freely applied both to the context
$\tcontextone$ (contravariantly) and to the type $\typeone$
(covariantly), leaving the rest of the judgement unchanged:
\medskip
\begin{lem}[Subtyping]
\label{lem:proof-strenghtening}
Suppose
$\icontextone;\iconstraintone;\varone_1:\mtypeone_1,\ldots,\varone_n:\mtypeone_n\vdash_\itermone
\termone:\typeone$ and $\icontextone;\iconstraintone\vdash
\mtypetwo_i\tless \mtypeone_i$ for $1\leq i\leq n$ and
$\icontextone;\iconstraintone\vdash \typeone\tless \typetwo$.  Then,
$\icontextone;\iconstraintone;\varone_1:\mtypetwo_1,\ldots,\varone_n:\mtypetwo_n\vdash_\itermone
\termone:\typetwo$.
\end{lem}
\begin{proof}
By induction on the structure of a derivation $\tdone$ for
$$
\icontextone;\iconstraintone;\varone_1:\mtypeone_1,\ldots,\varone_n:\mtypeone_n\vdash_{\itermone} \termone:\typeone.
$$
Let us examine some interesting cases:
\begin{varitemize}
\item
  If $\tdone$ is just
  $$
  \infer[\Axty]
        {\icontextone;\iconstraintone;\tcontextone,\varone:\qbang{\ivarone<\itermone}{\typethree}\vdash_{\itermtwo}^{\eqpone} \varone:\typeone}
        {
          \begin{array}{c}
            \icontextone;\iconstraintone\models^{\eqpone} \natit{0}\leq \itermtwo \qquad
            \icontextone;\iconstraintone\models^{\eqpone} \natit{1}\leq \itermone\\
            \icontextone;\iconstraintone\vdash^{\eqpone}\sb{\typethree}{\ivarone}{\natit{0}}\tless{\typeone}\\
            \icontextone;\iconstraintone\vdash^{\eqpone}\conv{(\qbang{\ivarone<\itermone}{\typethree})} \qquad
            \icontextone;\iconstraintone\vdash^{\eqpone}\conv{\tcontextone}
          \end{array}
        }
  $$
  then, by assumption we have that 
  $\mtypetwo\equiv\qbang{\ivarone<\itermthree}{\typefour}$ and $\icontextone;\iconstraintone\vdash 
  \qbang{\ivarone<\itermthree}{\typefour}\tless\qbang{\ivarone<\itermone}{\typethree}$. 
  Moreover, by assumption we have $\icontextone;\iconstraintone\vdash\typeone\tless \typetwo$.
  From $\icontextone;\iconstraintone\vdash\qbang{\ivarone<\itermthree}{\typefour}\tless\qbang{\ivarone<\itermone}{\typethree}$,
  it follows that $\icontextone;\iconstraintone,\ivarone<\itermthree\vdash\typefour\tless\typethree$
  and that $\icontextone;\iconstraintone\models\itermone\leq\itermthree$. By Lemma~\ref{lem:subst-respect-index-ineq-I2Type},
  $\icontextone;\iconstraintone\vdash\sb{\typefour}{\ivarone}{\natit{0}}\tless\sb{\typethree}{\ivarone}{\natit{0}}$, which
  by transitivity of $\tless$ implies $\icontextone;\iconstraintone\vdash^{\eqpone}\sb{\typefour}{\ivarone}{\natit{0}}\tless{\typetwo}$.
  Now, if $\tcontexttwo$ is a context such that (with a slight abuse of notation)
  $\icontextone;\iconstraintone\vdash^{\eqpone}\tcontexttwo\tless\tcontextone$, then
  $\icontextone;\iconstraintone\vdash^{\eqpone}\conv{\tcontexttwo}$. Summing up,
  $$
  \infer[\Axty]
        {\icontextone;\iconstraintone;\tcontexttwo,\varone:\qbang{\ivarone<\itermthree}{\typefour}\vdash_{\itermtwo}^{\eqpone} \varone:\typetwo}
        {
          \begin{array}{c}
            \icontextone;\iconstraintone\models^{\eqpone} \natit{0}\leq \itermtwo \qquad
            \icontextone;\iconstraintone\models^{\eqpone} \natit{1}\leq \itermthree\\
            \icontextone;\iconstraintone\vdash^{\eqpone}\sb{\typefour}{\ivarone}{\natit{0}}\tless{\typetwo}\\
            \icontextone;\iconstraintone\vdash^{\eqpone}\conv{(\qbang{\ivarone<\itermthree}{\typefour})}\qquad
            \icontextone;\iconstraintone\vdash^{\eqpone}\conv{\tcontexttwo}
          \end{array}
        }
  $$
  \item
    If $\tdone$ is
    $$
    \infer[\Apty]
          {\icontextone;\iconstraintone;\tcontexthree\vdash^{\eqpone}_{\itermfour} \termone\termtwo:\typeone}
          {
            \begin{array}{c}
              \icontextone;\iconstraintone;\tcontextone\vdash^{\eqpone}_{\itermtwo}\termone:\qlin{\ivarone<\itermone}{\typethree}{\typeone}\\
              \icontextone,\ivarone;\iconstraintone,\ivarone<\itermone ;\tcontexttwo\vdash^{\eqpone}_{\itermthree} \termtwo:\typethree\\
              \icontextone;\iconstraintone\vdash^{\eqpone}\tcontexthree \tless \tcontextone\uplus\sum_{\ivarone<\itermone}{\tcontexttwo}\\
              \icontextone;\iconstraintone\models^{\eqpone}\itermfour\geq\itermtwo+\itermone+ {\sum_{\ivarone<\itermone}\itermthree}
            \end{array}
          }
    $$
    but we have $\icontextone;\iconstraintone\vdash^{\eqpone}\typeone\tless\typetwo$
    and $\icontextone;\iconstraintone\vdash^{\eqpone}\tcontextfour\tless\tcontexthree$, then
    by induction hypothesis we can easily conclude that 
    $\icontextone;\iconstraintone;\tcontextone\vdash^{\eqpone}_{\itermtwo}\termone:\qlin{\ivarone<\itermone}{\typethree}{\typetwo}$
    and, by transitivity of $\tless$, that 
    $\icontextone;\iconstraintone\vdash^{\eqpone}\tcontextfour \tless \tcontextone\uplus\sum_{\ivarone<\itermone}{\tcontexttwo}$.
    As a consequence:
    $$
    \infer[\Apty]
          {\icontextone;\iconstraintone;\tcontextfour\vdash^{\eqpone}_{\itermfour} \termone\termtwo:\typetwo}
          {
            \begin{array}{c}
              \icontextone;\iconstraintone;\tcontextone\vdash^{\eqpone}_{\itermtwo}\termone:\qlin{\ivarone<\itermone}{\typethree}{\typetwo}\\
              \icontextone,\ivarone;\iconstraintone,\ivarone<\itermone ;\tcontexttwo\vdash^{\eqpone}_{\itermthree} \termtwo:\typethree\\
              \icontextone;\iconstraintone\vdash^{\eqpone}\tcontextfour \tless \tcontextone\uplus\sum_{\ivarone<\itermone}{\tcontexttwo}\\
              \icontextone;\iconstraintone\models^{\eqpone}\itermfour\geq\itermtwo+\itermone+{\sum_{\ivarone<\itermone}\itermthree}
            \end{array}
          }
    $$    
\end{varitemize}
The other cases are similar.
\end{proof}
Weakening holds for term contexts:
\begin{lem}[Context Weakening]
\label{lem:context-weakening}
Let $\icontextone;\iconstraintone;\tcontextone\vdash_{\itermone}\termone:\typeone$. Then,
$\icontextone;\iconstraintone;\tcontextone,\tcontexttwo\vdash_{\itermone}\termone:\typeone$
whenever $\icontextone;\iconstraintone\vdash\conv{\tcontexttwo}$.
\end{lem}
\begin{proof}
  Easy, by induction on the derivation proving $\icontextone;\iconstraintone;\tcontextone\vdash_{\itermone}\termone:\typeone$.
\end{proof}
Another useful transformation on type derivations is substitution of an index variable for an index term:
\begin{lem}[Index Term Substitution]
\label{lem:I2Term-substitution}
Let $\icontextone,\ivarone;\iconstraintone;\tcontextone\vdash_\itermone \termone :\typeone$.
Then we have $$\icontextone;\sb{\iconstraintone}{\ivarone}{\itermtwo},\iconstrainttwo;\sb{\tcontextone}{\ivarone}{\itermtwo}\vdash_{\sb{\itermone}{\ivarone}{\itermtwo}} \termone:\sb{\typeone}{\ivarone}{\itermtwo}$$
for every $\itermtwo$ such that $\icontextone,\iconstrainttwo\models^{\eqpone}\conv{\itermtwo}$.
\end{lem}
\begin{proof}
By induction on the structure of a derivation $\tdone$ for
$$
\icontextone,\ivarone;\iconstraintone;\tcontextone\vdash_\itermone \termone :\typeone.
$$
Let us examine some cases:
\begin{varitemize}
\item
  If $\tdone$ is just
  $$
  \infer[\Axty]
        {\icontextone,\ivarone;\iconstraintone;\tcontextone,\varone:
          \qbang{\ivartwo<\itermthree}{\typethree}\vdash_{\itermone}^{\eqpone} \varone:\typeone}
        {
          \begin{array}{c}
            \icontextone,\ivarone;\iconstraintone\models^{\eqpone} \natit{0}\leq \itermone \qquad
            \icontextone,\ivarone;\iconstraintone\models^{\eqpone} \natit{1}\leq \itermthree\\
            \icontextone,\ivarone;\iconstraintone\vdash^{\eqpone}\sb{\typethree}{\ivartwo}{\natit{0}}\tless{\typeone}\\
            \icontextone,\ivarone;\iconstraintone\vdash^{\eqpone}\conv{(\qbang{\ivartwo<\itermthree}{\typethree})} \qquad
            \icontextone,\ivarone;\iconstraintone\vdash^{\eqpone}\conv{\tcontextone}
          \end{array}
        }
  $$
  then of course we have that $\icontextone;\sb{\iconstraintone}{\ivarone}{\itermtwo},
  \iconstrainttwo\models^{\eqpone}\natit{0}\leq\sb{\itermone}{\ivarone}{\itermtwo}$
  and that $\icontextone;\sb{\iconstraintone}{\ivarone}{\itermtwo},\iconstrainttwo\models^{\eqpone}
  \natit{1}\leq\sb{\itermthree}{\ivarone}{\itermtwo}$.
  By Lemma \ref{lem:subst-respect-index-ineq-I2Type}, one obtains
  $\icontextone;\sb{\iconstraintone}{\ivarone}{\itermtwo},\iconstrainttwo\vdash^{\eqpone}
  \sb{(\sb{\typethree}{\ivartwo}{\natit{0}})}{\ivarone}{\itermtwo}\tless\sb{\typeone}{\ivarone}{\itermtwo}$.
  Please observe that $\ivartwo$ can be assumed not to occur free in $\itermtwo$, and
  as a consequence $\sb{(\sb{\typethree}{\ivartwo}{\natit{0}})}{\ivarone}{\itermtwo}\equiv
  \sb{(\sb{\typethree}{\ivarone}{\itermtwo})}{\ivartwo}{\natit{0}}$.
  Similarly,
  $\icontextone;\sb{\iconstraintone}{\ivarone}{\itermtwo},\iconstrainttwo\vdash^{\eqpone}
  \conv{(\sb{(\qbang{\ivartwo<\itermthree}{\typethree})}{\ivarone}{\itermtwo})}$
  and $\icontextone;\sb{\iconstraintone}{\ivarone}{\itermtwo},\iconstrainttwo\vdash^{\eqpone}
  \conv{\sb{\tcontextone}{\ivarone}{\itermtwo}}$. Again, 
  $\sb{(\qbang{\ivartwo<\itermthree}{\typethree})}{\ivarone}{\itermtwo}$ is syntactically
  identical to $\qbang{\ivartwo<\sb{\itermthree}{\ivarone}{\itermtwo}}{\sb{\typethree}{\ivarone}{\itermtwo}}$.
  As a consequence:
  $$
  \infer[\Axty]
        {\icontextone;\sb{\iconstraintone}{\ivarone}{\itermtwo},\iconstrainttwo;\sb{\tcontextone}{\ivarone}{\itermtwo},\varone:
          \qbang{\ivartwo<\sb{\itermthree}{\ivarone}{\itermtwo}}{\sb{\typethree}{\ivarone}{\itermtwo}}
          \vdash_{\sb{\itermone}{\ivarone}{\itermtwo}}^{\eqpone} \varone:\sb{\typeone}{\ivarone}{\itermtwo}}
        {
          \begin{array}{c}
            \icontextone;\sb{\iconstraintone}{\ivarone}{\itermtwo},\iconstrainttwo\models^{\eqpone} 
              \natit{0}\leq\sb{\itermone}{\ivarone}{\itermtwo} \qquad
            \icontextone;\sb{\iconstraintone}{\ivarone}{\itermtwo},\iconstrainttwo\models^{\eqpone} 
              \natit{1}\leq\sb{\itermthree}{\ivarone}{\itermtwo}\\
            \icontextone;\sb{\iconstraintone}{\ivarone}{\itermtwo},\iconstrainttwo\vdash^{\eqpone}
              \sb{(\sb{\typethree}{\ivarone}{\itermtwo})}{\ivartwo}{\natit{0}}\tless{\sb{\typeone}{\ivarone}{\itermtwo}}\\
            \icontextone;\sb{\iconstraintone}{\ivarone}{\itermtwo},\iconstrainttwo\vdash^{\eqpone}
              \conv{(\qbang{\ivartwo<\sb{\itermthree}{\ivarone}{\itermtwo}}{\sb{\typethree}{\ivarone}{\itermtwo}})} \qquad
            \icontextone;\sb{\iconstraintone}{\ivarone}{\itermtwo},\iconstrainttwo\vdash^{\eqpone}
              \conv{(\sb{\tcontextone}{\ivarone}{\itermtwo})}
          \end{array}
        }
  $$
\item
  If $\tdone$ is
  $$
  \infer[\Lamty]
        {\icontextone,\ivarone;\iconstraintone;\tcontextone \vdash_{\itermone} \lambda \varone.
          \termone:\qlin{\ivartwo<\itermthree}{\typethree}{\typetwo}}
        {
          \icontextone,\ivarone;\iconstraintone;\tcontextone,\varone :\qbang{\ivartwo<\itermthree}{\typethree}
          \vdash_{\itermone} \termone:\typetwo
        }
  $$
  then, by the induction hypothesis we get
  $$
  \icontextone;\sb{\iconstraintone}{\ivarone}{\itermtwo},\iconstrainttwo;\sb{\tcontextone}{\ivarone}{\itermtwo}, 
  \varone :\qbang{\ivartwo<\sb{\itermthree}{\ivarone}{\itermtwo}}{\sb{\typethree}{\ivarone}{\itermtwo}}
  \vdash_{\sb{\itermone}{\ivarone}{\itermtwo}} \termone:\sb{\typetwo}{\ivarone}{\itermtwo}.
  $$
  As a consequence, we can conclude by 
  $$
  \infer[\Lamty]
        {\icontextone;\sb{\iconstraintone}{\ivarone}{\itermtwo},\iconstrainttwo;\sb{\tcontextone}{\ivarone}{\itermtwo} 
          \vdash_{\sb{\itermone}{\ivarone}{\itermtwo}} \lambda \varone.\termone:\sb{(\qlin{\ivartwo<\itermthree}
                {\typethree}{\typetwo})}{\ivarone}{\itermtwo}}
        {
          \icontextone;\sb{\iconstraintone}{\ivarone}{\itermtwo},\iconstrainttwo;\sb{\tcontextone}{\ivarone}{\itermtwo},
          \varone :\qbang{\ivartwo<\sb{\itermthree}{\ivarone}{\itermtwo}}{\sb{\typethree}{\ivarone}{\itermtwo}}
          \vdash_{\sb{\itermone}{\ivarone}{\itermtwo}} \termone:\sb{\typetwo}{\ivarone}{\itermtwo}
        }
  $$
since $\qlin{\ivartwo<\sb{\itermthree}{\ivarone}{\itermtwo}}{\sb{\typethree}{\ivarone}{\itermtwo}}
{\sb{\typetwo}{\ivarone}{\itermtwo}}\equiv\sb{(\qlin{\ivartwo<\itermthree}{\typethree}{\typetwo})}{\ivarone}{\itermtwo}$.
\end{varitemize}
The other cases are similar.
\end{proof}
A particularly useful instance of Lemma~\ref{lem:I2Term-substitution} is the following:
\begin{lem}[Instantiation]
\label{lem:instantiation}
Let  $\icontextone,\ivarone;\iconstraintone,\ivarone<\itermone\vdash_{\itermthree} \termone:\typeone$. If
$\icontextone;\iconstrainttwo\models_{\eqpone} \itermtwo<\itermone$,
then, $\icontextone;\sb{\iconstraintone}{\ivarone}{\itermtwo},\iconstrainttwo\vdash_{\sb{\itermthree}{\ivarone}{\itermtwo}} 
\termone:\sb{\typeone}{\ivarone}{\itermtwo}$. 
\end{lem}
\begin{proof}
  By Lemma \ref{lem:I2Term-substitution} and Lemma \ref{lem:proof-strenghtening}.
\end{proof}
Moreover a Generation Lemma will be useful.
\begin{lem}[Generation]
\label{lem:generation}$ $
  \begin{varenumerate}
  \item 
    Let $\icontextone;\iconstraintone;\tcontextone\vdash_{\itermthree} \lambda \varone.\termone:\typeone$, then 
    $\typeone=\qlin{\ivarone<\itermone}{\typetwo}{\typethree}$ and 
    $\icontextone;\iconstraintone;\tcontextone,\varone:\qbang{\ivarone<\itermone}{\typetwo}\vdash_{\itermthree}\termone:\typethree$;
  \item  
    Let $\icontextone;\iconstraintone;\tcontextone\vdash_{\itermthree} \val{0}:\Nat[\itermone,\itermtwo]$, then
    $\icontextone;\iconstraintone\models^{\eqpone}\itermone=\natit{0}$;
  \item  
    Let $\icontextone;\iconstraintone;\tcontextone\vdash_{\itermthree} \val{n+1}:\Nat[\itermone,\itermtwo]$, then
    $\icontextone;\iconstraintone\models^{\eqpone}\itermtwo\geq\natit{1}$.
  \end{varenumerate}
\end{lem}
\begin{proof} 
  All the points are immediate by an inspection of the rules.
\end{proof}
We are now ready to embark on a proof of Subject Reduction. As usual, the first step is a Substitution Lemma:
\begin{lem}[Term Substitution]\label{lemma:subst}
\label{lem:t2tsubstitution}
Let $\icontextone,\ivarone;\iconstraintone,\ivarone<\itermone;\emcon\vdash_\itermtwo \termone :\typeone$ and
$\icontextone;\iconstraintone;\varone:\qbang{\ivarone< \itermone}{\typeone},\tcontexttwo\vdash_\itermthree \termtwo:\typetwo$.
Then we have $\icontextone;\iconstraintone;\tcontexttwo\vdash_{\itermfour} \sb{\termtwo}{\varone}{\termone}:\typetwo$ for some $\itermfour$ such that $\icontextone;\iconstraintone\models^{\eqpone} \itermfour\leq {\itermthree+\itermone+\sum_{\ivarone<\itermone}\itermtwo} $.
\end{lem}
\begin{proof}
As usual, this is an induction on the structure of a type derivation for $\termtwo$. All relevant inductive cases require
some manipulation of the type derivation for $\termone$. The previous lemmas give exactly the right degree of ``malleability''. 
Let $\tdone$ be a derivation for
$$
\icontextone;\iconstraintone;\varone:\qbang{\ivarone<\itermone}{\typeone},\tcontexttwo\vdash_{\itermthree} \termtwo:\typetwo.
$$
Let us examine some interesting cases, dependently on the shape of $\tdone$:
\begin{varitemize}
\item
  Consider $\tdone$ to be just
  $$
  \infer[\Axty]
        {\icontextone;\iconstraintone;\tcontexttwo,\varone:
          \qbang{\ivarone<\itermone}{\typeone}\vdash_{\itermthree}\varone:\typetwo}
        {
          \begin{array}{c}
            \icontextone;\iconstraintone\models^{\eqpone} \natit{0}\leq \itermthree \qquad
            \icontextone;\iconstraintone\models^{\eqpone} \natit{1}\leq \itermone\\
            \icontextone;\iconstraintone\vdash^{\eqpone}\sb{\typeone}{\ivarone}{\natit{0}}\tless{\typetwo}\\
            \icontextone;\iconstraintone\vdash^{\eqpone}\conv{(\qbang{\ivarone<\itermone}{\typeone})} \qquad
            \icontextone;\iconstraintone\vdash^{\eqpone}\conv{\tcontexttwo}
          \end{array}
        }
  $$
  Since $\icontextone;\emcon\models\conv{\natit{0}}$, applying Lemma 
  \ref{lem:instantiation} we have 
  $$
  \icontextone;\sb{\iconstraintone}{\ivarone}{\natit{0}};\emptyset\vdash_{\sb{\itermtwo}{\ivarone}{\natit{0}}} 
    \termone :\sb{\typeone}{\ivarone}{\natit{0}}
  $$
  and since $\iconstraintone$  does not contain 
  free occurrences of $\ivarone$ we obtain:
  $$
  \icontextone;\iconstraintone;\emptyset\vdash_{\sb{\itermtwo}{\ivarone}{\natit{0}}} \termone :\sb{\typeone}{\ivarone}{\natit{0}}.
  $$
  Now, by applying Lemma \ref{lem:context-weakening}, Lemma \ref{lem:weight-monotonicity}
  and Lemma~\ref{lem:proof-strenghtening} we can conclude
  $$
  \icontextone;\iconstraintone;\tcontexttwo\vdash_{\itermthree+\itermone+\sum_{\ivarone<\itermone}\itermtwo} 
    \termone:\typetwo
  $$
  since clearly
  $$
  \icontextone;\iconstraintone\models\sb{\itermtwo}{\ivarone}{\natit{0}}\leq\itermthree+\itermone+\sum_{\ivarone<\itermone}\itermtwo.
  $$
\item 
  Let us consider the case  $\tdone$ ends by an instance of the
  $\Apty$ rule. In particular, without loss of generality we can consider a situation as
  the following:
  $$
  \infer[\Apty]
        {\icontextone;\iconstraintone;\varone:\qbang{\ivarone<\itermone}{\typeone}
          \vdash_{\itermten} \termthree\termtwo:\typetwo}
        {
          \begin{array}{c}
            \icontextone;\iconstraintone;\varone:\qbang{\ivarone< \itermthree}{\typefour}
              \vdash_{\itermfive}\termthree:\qlin{\ivartwo<\itermseven}{\typethree}{\typetwo}\\
            \icontextone,\ivartwo;\iconstraintone,\ivartwo<\itermseven;
              \varone:\qbang{\ivarone<\itermfour}{(\sb{\typefour}{\ivarone}
                {\itermthree+\ivarone+\sum_{\ivarfour<\ivartwo}{\sb{\itermfour}{\ivartwo}{\ivarfour}}})}
              \vdash_{\itermsix} \termtwo:\typethree\\
            \icontextone;\iconstraintone\vdash\qbang{\ivarone<\itermone}{\typeone} 
              \tless\qbang{\ivarone<\itermthree+\sum_{\ivartwo<\itermseven}{\itermfour}}{\typefour}\\
            \icontextone;\iconstraintone\models\itermten\geq\itermfive+\itermseven+\sum_{\ivartwo<\itermseven}\itermsix 
          \end{array}
        }
   $$
By definition of subtyping,
$\icontextone;\iconstraintone,\ivarone<\itermone\vdash\typeone\tless\typefour$,
and $\icontextone;\iconstraintone\models^{\eqpone} \itermthree+\itermeight\leq \itermone$,
where $\itermeight\equiv\sum_{\ivartwo<\itermseven}{\itermfour}$.
So, by Lemma \ref{model-strengthening-typing}, we have 
$$
 \icontextone;\iconstraintone,\ivarone<\itermthree+\itermeight;\emptyset\vdash_{\itermtwo}\termone:\typeone
$$
and by Lemma \ref{lem:proof-strenghtening} we have 
$$
 \icontextone;\iconstraintone,\ivarone<\itermthree+\itermeight;\emptyset\vdash_{\itermtwo}\termone:\typefour
$$
(since $\icontextone;\iconstraintone,\ivarone<\itermthree+\itermeight\vdash\typeone\tless\typefour$).
Applying again Lemma \ref{model-strengthening-typing}
we obtain 
$$
 \icontextone;\iconstraintone,\ivarone<\itermthree;\emcon\vdash_{\itermtwo}\termone:\typefour
$$
and by induction hypothesis we get
$$
 \icontextone;\iconstraintone;\emcon\vdash_{\itermtwelve}\sb{\termthree}{\varone}{\termone}:\qlin{\ivartwo<\itermseven}{\typethree}{\typetwo}
$$
with $\icontextone;\iconstraintone\models^{\eqpone}\itermtwelve\leq \itermfive+\itermthree+\sum_{\ivarone<\itermthree}\itermtwo$.
We observe that
$$
\icontextone,\ivartwo,\ivarthree;\iconstraintone,\ivarone\leq\itermthree+\ivarthree+
  \sum_{\ivarfour<\ivartwo}\sb{\itermfour}{\ivartwo}{\ivarfour},\ivartwo<\itermseven,\ivarthree<\itermfour
  \models^\eqpone\ivarone<\itermthree+\itermeight.
$$
By Lemma \ref{model-strengthening-typing} we get
$$
 \icontextone,\ivartwo,\ivarthree;\iconstraintone,\ivarone\leq\itermthree+\ivarthree+
  \sum_{\ivarfour<\ivartwo}\sb{\itermfour}{\ivartwo}{\ivarfour},\ivartwo<\itermseven,\ivarthree<\itermfour;\emptyset
  \vdash_{\itermtwo}\termone:\typefour
$$
and by Lemma \ref{lem:I2Term-substitution} and Lemma \ref{lem:proof-strenghtening}
we obtain 
$$
 \icontextone;\iconstraintone,\ivarone<\itermfour,\ivartwo<\itermseven;\emptyset\vdash_{\itermnine}\termone:\sb{\typefour}{\ivarone}{\itermthree+\ivarone+\sum_{\ivarfour<\ivartwo}{\sb{\itermfour}{\ivartwo}{\ivarfour}}},
$$
where $\itermnine\equiv\sb{\itermtwo}{\ivarone}{\itermthree+\ivarone+\sum_{\ivarfour<\ivartwo}{\sb{\itermfour}{\ivartwo}{\ivarfour}}}$.
By induction hypothesis, we get
$$
 \icontextone;\iconstraintone,\ivartwo<\itermseven;\emcon\vdash_{\itermeleven} \sb{\termtwo}{\varone}{\termone}:\typethree
$$
with $\icontextone;\iconstraintone\models^{\eqpone}\itermeleven\leq\itermsix+ \itermfour+ \sum_{\ivarone<\itermfour}\itermnine$.
And we can conclude as follows:
$$
\infer[\Apty]
      { \icontextone;\iconstraintone;\emcon\vdash_{\itermtwelve+\itermseven+\sum_{\ivartwo<\itermtwo}\itermeleven} 
        \sb{\termthree}{\varone}{\termone}\sb{\termtwo}{\varone}{\termone}:\typetwo
      }{ 
        \begin{array}{c}
          \icontextone;\iconstraintone;\emcon\vdash_{\itermtwelve} \sb{\termthree}{\varone}{\termone}:
          \qlin{\ivartwo<\itermseven}{\typethree}{\typetwo}\\
          \icontextone;\iconstraintone;\emcon\vdash_{\itermeleven} \sb{\termtwo}{\varone}{\termone}:\typethree\\
        \end{array}
      }
$$
Please observe that:
\begin{align*}
\icontextone;\iconstraintone\models^{\eqpone}\itermtwelve+\itermseven+\sum_{\ivartwo<\itermtwo}\itermeleven&
  \leq(\itermfive+ \itermthree+ \sum_{\ivarone<\itermthree}\itermtwo)+\itermseven+
  \sum_{\ivartwo<\itermseven}(\itermsix+ \itermfour+ \sum_{\ivarone<\itermfour}\itermnine)\\
  &\leq(\itermfive+\itermseven+\sum_{\ivartwo<\itermseven}\itermsix)+(\itermthree+\sum_{\ivartwo<\itermseven}\itermfour)+
    (\sum_{\ivarone<\itermthree}\itermtwo+\sum_{\ivartwo<\itermseven}\sum_{\ivarone<\itermfour}\itermnine)\\
  &\leq(\itermfive+\itermseven+\sum_{\ivartwo<\itermseven}\itermsix)+(\itermthree+\sum_{\ivartwo<\itermseven}\itermfour)+
    \sum_{\ivarone<\itermthree+\sum_{\ivartwo<\itermseven}\itermfour}\itermtwo\\
  &\leq\itermten+\itermone+\sum_{\ivarone<\itermone}\itermtwo.
\end{align*}
\end{varitemize}
The other cases are similar.
\end{proof}
\begin{thm}[Subject Reduction]
Let $\icontextone;\iconstraintone;\emcon\vdash_\itermone \termone :\typeone$ and
$\termone\to \termtwo$. Then, $\icontextone;\iconstraintone;\emcon\vdash_\itermtwo \termtwo :\typeone$,
where $\icontextone;\iconstraintone\models\itermtwo\leq\itermone$.
\end{thm}
\begin{proof}
  By induction on the structure of a derivation $\tdone$ for
 $\icontextone;\iconstraintone;\emptyset\vdash_{\itermone} \termone :\typeone$
Let us examine the distinct cases:
\begin{varitemize}
\item
  Suppose $\tdone$ is 
  $$
  \infer[\Apty]{\icontextone;\iconstraintone;\emptyset\vdash_\itermone (\lambda \varone.\termone)\termtwo:\typeone}{
    \begin{array}{c}
      \icontextone;\iconstraintone;\emptyset\vdash_{\itermthree} \lambda \varone.\termone:\qlin{\ivarone<\itermfour}{\typetwo}{\typeone}\\
      \icontextone;\iconstraintone,\ivarone<\itermfour;\emptyset\vdash_{\itermfive} \termtwo:\typetwo\\
      \icontextone;\iconstraintone\models\itermthree+\itermfour+\sum_{\ivarone<\itermfour}\itermfive\leq\itermone
    \end{array}
  }
  $$
  By Lemma \ref{lem:generation}, Point 1, we have 
  $\icontextone;\iconstraintone;\varone:\qbang{\ivarone<\itermfour}{\typetwo} \vdash_{\itermthree} \termone:{\typeone}$.
  Then by Lemma \ref{lem:t2tsubstitution} we can conclude:
  $$
  \icontextone;\iconstraintone;\emptyset\vdash_{\itermtwo} \sb{\termone}{\varone}{\termtwo}:\typeone
  $$
  for $\icontextone;\iconstraintone\models^{\eqpone}\itermtwo\leq\itermthree+\itermfour+\sum_{\ivarone<\itermfour}\itermfive\leq\itermone$.
\item
  Suppose $\tdone$ is 
  $$
  \infer[\Ifty]
  {\icontextone;\iconstraintone;\emptyset \vdash_{\itermone}\CASE{\val{0}}{\termthree}{\termfour}:\typeone}
  {
    \begin{array}{c}
      \icontextone;\iconstraintone;\emptyset\vdash_{\itermthree} \val{0}:\Nat[\itermthree,\itermfour] \\
      \icontextone;\iconstraintone,\itermfour\leq\zero;\emptyset\vdash_{\itermfive} \termthree: \typeone \\
      \icontextone;\iconstraintone,\itermthree\geq 1;\emptyset\vdash_{\itermfive} \termfour: \typeone \\
      \icontextone;\iconstraintone\models\itermthree+\itermfive\leq\itermone
    \end{array}
  }
  $$
  By Lemma \ref{lem:generation}, Point 2, we have $\icontextone;\iconstraintone\models^{\eqpone} \itermfour\leq\zero$. So, by Lemma
  \ref{model-strengthening-typing} we can conclude 
  $\icontextone;\iconstraintone;\emptyset\vdash_{\itermfive} \termthree: \typeone$.
\item
  Suppose $\tdone$ is 
  $$
  \infer[\Ifty]
        {\icontextone;\iconstraintone;\emptyset \vdash_{\itermone}\CASE{\val{n+1}}{\termthree}{\termfour}:\typeone}
        {
          \begin{array}{c}
            \icontextone;\iconstraintone;\emptyset\vdash_{\itermthree} \val{n+1}:\Nat[\itermthree,\itermfour] \\
            \icontextone;\iconstraintone,\itermfour\leq\zero;\emptyset\vdash_{\itermfive} \termthree: \typeone \\
            \icontextone;\iconstraintone,\itermthree\geq 1;\emptyset\vdash_{\itermfive} \termfour: \typeone \\
            \icontextone;\iconstraintone\models\itermthree+\itermfive\leq\itermone
          \end{array}
        }
  $$
  By Lemma \ref{lem:generation}, Point 3, we have $\icontextone;\iconstraintone\models^{\eqpone} \itermthree\geq 1$. So, by Lemma
  \ref{model-strengthening-typing} we have $\icontextone;\iconstraintone;\emptyset\vdash_{\itermfive}\termfour:\typeone$.
\item
  Suppose $\tdone$ is
  $$
  \infer[\Recty]
  {\icontextone;\iconstraintone;\emcon\vdash_{\itermone} \REC{\varone}{\termone}:\typeone}
  {
    \begin{array}{c}
      \icontextone;\iconstraintone\vdash\Dfscomb{\ivartwo}{\natit{0}}{\natit{1}}{\itermtwo}\leq\itermfive,\itermeight\\
      \icontextone,\ivartwo;\iconstraintone,\ivartwo<\itermfive;
      \varone:\qbang{\ivarone<\itermtwo}{\typethree}\vdash_{\itermthree}\termone:\typetwo\\
      \icontextone;\iconstraintone\vdash\sb{\typetwo}{\ivartwo}{\natit{0}}\tless\typeone\\
      \icontextone,\ivarone,\ivartwo;\iconstraintone,\ivarone<\itermtwo,\ivartwo<\itermfive
        \vdash\sb{\typetwo}{\ivartwo}{\Dfscomb{\ivartwo}{\ivartwo+1}{\ivarone}{\itermtwo}+\ivartwo+1}\tless\typethree\\
      \icontextone,\iconstraintone\models\itermeight\mnu 1+\sum_{\ivartwo<\itermfive}\itermthree\leq\itermone
    \end{array}
  }
  $$    
  The index term $\itermtwo$ describes a tree $\tree{\itermtwo}$ (in the sense of forest
  cardinalities, see Section~\ref{sect:itep}) which in turn represents the tree of recursive calls.
  $\tree{\itermtwo}$ looks as follows:
  $$
  \xymatrix
  {
    & & \cdot\ar@{-}[dll]\ar@{-}[dl]\ar@{-}[dr] & \\
    \tree{\itermtwo}^1 & \tree{\itermtwo}^2 & \ldots & \tree{\itermtwo}^{\sb{\itermtwo}{\ivartwo}{\natit{0}}}
  }
  $$
  where $\tree{\itermtwo}^i$ represents the tree of recursive calls triggered by the $i$-th call
  to $\varone$ in $\termone$. We first proceed by giving a type to $\termone$ which somehow corresponds
  to the root of $\tree{\itermtwo}$. This will be done by substituting $\ivartwo$ for $\natit{0}$ in
  the derivation we get as an hypothesis of $\tdone$. Since $\icontextone;\iconstraintone\models_{\eqpone}\natit{0}<\itermfive$,
  by Lemma \ref{lem:instantiation} we have
  $$
  \icontextone;\iconstraintone;
  \varone:\qbang{\ivarone<\sb{\itermtwo}{\ivartwo}{\natit{0}}}{\sb{\typeone}{\ivartwo}{\natit{0}}}\vdash_{\sb{\itermthree}{\ivartwo}{\natit{0}}}
    \termone:\sb{\typetwo}{\ivartwo}{\natit{0}}.
  $$
  From the hypothesis $\icontextone;\iconstraintone\vdash\sb{\typetwo}{\ivartwo}{\natit{0}}\tless\typeone$
  and by the Subtyping Lemma, we obtain 
  $$
  \icontextone;\iconstraintone;
  \varone:\qbang{\ivarone<\sb{\itermtwo}{\ivartwo}{\natit{0}}}{\sb{\typeone}{\ivartwo}{\natit{0}}}\vdash_{\sb{\itermthree}{\ivartwo}{\natit{0}}}
    \termone:\typeone.
  $$
  Our objective now is building \emph{one} type derivation for $\REC{\varone}{\termone}$ that somehow
  reflect the $\sb{\itermtwo}{\ivartwo}{\natit{0}}$ subtrees $\tree{\itermtwo}^1,\ldots,\tree{\itermtwo}^{\sb{\itermtwo}{\ivartwo}{\natit{0}}}$.
  Speaking more formally, we want to prove that:
  \begin{equation}\label{equ:srfix-I}
  \icontextone;\iconstraintone,\ivarone<\sb{\itermtwo}{\ivartwo}{\natit{0}}\vdash_{\itermnine}\REC{\varone}{\termone}:
  \sb{\typeone}{\ivartwo}{\natit{0}}
  \end{equation}
  where
  $$
  \icontextone;\iconstraintone\models\sb{\itermthree}{\ivartwo}{\natit{0}}+\sb{\itermtwo}{\ivartwo}{\natit{0}}+
    \sum_{\ivarone<\sb{\itermtwo}{\ivartwo}{\natit{0}}}\itermnine\leq\itermone.
  $$
  That would immediately lead to the thesis. To reach (\ref{equ:srfix-I}), we proceed by first defining
  two index terms with a quite intuitive informal semantics:
  \begin{varitemize}
  \item
    First of all, we define $\itermsix$ to be
    $\Dfscomb{\ivartwo}{\natit{0}}{\natit{1}}{\sb{\itermtwo}{\ivartwo}{\ivartwo+\natit{1}+\Dfscomb{\ivartwo}{\natit{1}}{\ivarthree}{\itermtwo}}}$.
    Observe that $\ivarthree$ occurs free in $\itermsix$; indeed, $\itermsix$ 
    counts the number of nodes in the tree $\tree{\itermtwo}^\ivarthree$.
  \item
    Another useful index term is $\itermseven$, which is defined to be
    $\natit{1}+\ivartwo+\sum_{\ivarthree<\ivarone}\itermsix$. $\itermseven$ is designed
    as to return the label of a node in $\tree{\itermtwo}^\ivarone$ given
    $\ivarone$ and the offset $\ivartwo$. In other words,
    $\tree{\sb{\itermtwo}{\ivartwo}{\itermseven}}$ is a
    recursion tree isomorphic to $\tree{\itermtwo}^\ivarone$.
  \end{varitemize}
  Now, if we substitute $\ivartwo$ for $\itermseven$
  in one of the premises of $\tdone$, we get
  \begin{multline}\label{equ:srfix-II}
  \qquad\icontextone,\ivarone,\ivartwo;\iconstraintone,\ivarone<\sb{\itermtwo}{\ivartwo}{\natit{0}},\ivartwo<\sb{\itermsix}{\ivarthree}{\ivarone};
  \varone:\qbang{\ivarfour<\sb{\itermtwo}{\ivartwo}{\itermseven}}
         {\sb{\sb{\typethree}{\ivarone}{\ivarfour}}{\ivartwo}{\itermseven}}\\
         \vdash_{\sb{\itermthree}{\ivartwo}{\itermseven}}\termone:\sb{\typetwo}{\ivartwo}{\itermseven}.
  \end{multline}
  Since by Lemma \ref{lemma:dfssum} we have $\sum_{\ivarthree<\ivarfive}\itermsix \kleq\Dfscomb{\ivartwo}{\natit{1}}{\ivarfive}{\itermtwo}$
  we know that
  \begin{equation}\label{equ:srfix-III}
    \Dfscomb{\ivartwo}{\natit{0}}{\natit{1}}{\sb{\itermtwo}{\ivartwo}{\itermseven}}
    \kleq\Dfscomb{\ivartwo}{\natit{0}}{\natit{1}}{\sb{\itermtwo}{\ivartwo}{\natit{1}+\ivartwo+ \sum_{\ivarthree<\ivarone}\itermsix}}
    \kleq\Dfscomb{\ivartwo}{\natit{0}}{\natit{1}}{\sb{\itermtwo}{\ivartwo}{\natit{1}+\ivartwo+ \Dfscomb{\ivartwo}{\natit{1}}{\ivarone}\itermtwo}}
    \kleq\sb{\itermsix}{\ivarthree}{\ivarone}.
  \end{equation}
  Now, consider the problem of determining the index (in $\tree{\itermtwo}$)
  of the $(\ivarfour+1)$-th children of a node of index $\ivartwo$ inside
  $\tree{\itermtwo}^\ivarone$. There are two equivalent ways to
  compute it:
  \begin{varitemize}
  \item
    either you start from $\itermseven$, but then you substitute $\ivartwo$
    by $\ivartwo+\natit{1}+\Dfscomb{\ivartwo}{\ivartwo+\natit{1}}{\ivarfour}{\sb{\itermtwo}{\ivartwo}{\itermseven}}$;
  \item
    or you simply consider $\itermseven+\natit{1}+\Dfscomb{\ivartwo}{\itermseven+\natit{1}}{\ivarfour}{\itermtwo}$. 
  \end{varitemize}
  In the first case, you compute the desired index by merely instantiating $\itermseven$ appropriately, while in the
  second case you use $\itermseven$ without altering it. The observation above can be
  formalized as follows:
  \begin{multline*}
    \qquad\icontextone,\ivarone,\ivartwo,\ivarfour;\iconstraintone,\ivarone<\sb{\itermtwo}{\ivartwo}{\natit{0}},
    \ivartwo<\sb{\itermsix}{\ivarthree}{\ivarone},\ivarfour<
    \sb{\itermtwo}{\ivartwo}{\itermseven} \vdash\\
    \sb{\sb{\typetwo}{\ivartwo}{\itermseven}}
       {\ivartwo}{\ivartwo+\natit{1}+\Dfscomb{\ivartwo}{\ivartwo+\natit{1}}{\ivarfour}{
           \sb{\itermtwo}{\ivartwo}{\itermseven} }}
       \kleq\sb{\typetwo}{\ivartwo}{\itermseven+\natit{1}+\Dfscomb{\ivartwo}{\itermseven+\natit{1}}{\ivarfour}{\itermtwo}}.
  \end{multline*}
  By Lemma \ref{lem:instantiation}, we also obtain: 
  \begin{multline}\label{equ:srfix-IV}
    \qquad\icontextone,\ivarone,\ivartwo,\ivarfour;\iconstraintone,\ivarone<\sb{\itermtwo}{\ivartwo}{\natit{0}},
    \ivartwo<\sb{\itermsix}{\ivarthree}{\ivarone},\ivarfour< \sb{\itermtwo}{\ivartwo}{\itermseven}
    \vdash \\
    \sb{\sb{\typetwo}{\ivartwo}{\itermseven}}
       {\ivartwo}{\ivartwo+\natit{1}+\Dfscomb{\ivartwo}{\ivartwo+\natit{1}}{\ivarfour}{
           \sb{\itermtwo}{\ivartwo}{\itermseven}
       }}
       \tless
       \sb{\sb{\typeone}{\ivarone}{\ivarfour}}{\ivartwo}{\itermseven}.
  \end{multline}
  Now, (\ref{equ:srfix-II}), (\ref{equ:srfix-III}) and (\ref{equ:srfix-IV}) can be put together
  by way of rule $\Recty$, and one then conclude that
  $$
  \icontextone;\iconstraintone,\ivarone<\sb{\itermtwo}{\ivartwo}{\natit{0}};\emcon
  \vdash_{\sb{\itermsix}{\ivarthree}{\ivarone}\mnu\natit{1}+\sum_{\ivartwo<\sb{\itermsix}{\ivarthree}{\ivarone}}\sb{\itermthree}{\ivartwo}{\itermseven}}
  \REC{\varone}{\termone}:
  \sb{\typetwo}{\ivartwo}{\sb{\itermseven}{\ivartwo}{\natit{0}}}.
  $$
  But instantiating one of the hypothesis' of $\tdone$, we obtain 
  $$
  \icontextone,\ivarone;\iconstraintone,\ivarone<\sb{\itermtwo}{\ivartwo}{\natit{0}}
  \vdash\sb{\typetwo}{\ivartwo}{\Dfscomb{\ivartwo}{\natit{1}}{\ivarone}{\itermtwo}+\natit{1}}\tless\sb{\typethree}{\ivartwo}{\natit{0}}.
  $$
  By Lemma~\ref{lemma:dfssum}, we can prove that
  $\Dfscomb{\ivartwo}{\natit{1}}{\ivarone}{\itermtwo}+\natit{1}=\sb{\itermseven}{\ivartwo}{\natit{0}}$. Indeed, this is quite intuitive:
  the index of the root of $\tree{\itermtwo}^\ivarone$ can be computed in two equivalent ways through
  $\itermtwo$ or through $\itermseven$. As a consequence,
  $$
  \icontextone;\iconstraintone,\ivarone<\sb{\itermtwo}{\ivartwo}{\natit{0}};\emcon
  \vdash_{\itermnine}\REC{\varone}{\termone}:\sb{\typeone}{\ivartwo}{\natit{0}},
  $$
  where $\itermnine\equiv\sb{\itermsix}{\ivarthree}{\ivarone}\mnu \natit{1}+
  \sum_{\ivartwo<\sb{\itermsix}{\ivarthree}{\ivarone}}\sb{\itermthree}{\ivartwo}{\itermseven}$.
  But we are done, since
  \begin{align*}
    \icontextone;\iconstraintone&\models\sb{\itermthree}{\ivartwo}{\natit{0}}+\sb{\itermtwo}{\ivartwo}{\natit{0}}+  
      \sum_{\ivarone<\sb{\itermtwo}{\ivartwo}{\natit{0}}}\itermnine\\
    &\equiv\sb{\itermthree}{\ivartwo}{\natit{0}}+\sb{\itermtwo}{\ivartwo}{\natit{0}}+
      \sum_{\ivarone<\sb{\itermtwo}{\ivartwo}{\natit{0}}}(\sb{\itermsix}{\ivarthree}{\ivarone}\mnu\natit{1}+
      \sum_{\ivartwo<\sb{\itermsix}{\ivarthree}{\ivarone}}\sb{\itermthree}{\ivartwo}{\itermseven})\\
    &=(\sb{\itermtwo}{\ivartwo}{\natit{0}}+
      \sum_{\ivarone<\sb{\itermtwo}{\ivartwo}{\natit{0}}}(\sb{\itermsix}{\ivarthree}{\ivarone}\mnu\natit{1}))+
      \sb{\itermthree}{\ivartwo}{\natit{0}}+\sum_{\ivarone<\sb{\itermtwo}{\ivartwo}{\natit{0}}}
      \sum_{\ivartwo<\sb{\itermsix}{\ivarthree}{\ivarone}}\sb{\itermthree}{\ivartwo}{\itermseven}\\
    &\leq \Dfscomb{\ivartwo}{\natit{1}}{\sb{\itermtwo}{\ivartwo}{\natit{0}}}{\itermtwo} + \sb{\itermthree}{\ivartwo}{\natit{0}}+
      \sum_{\ivarone<\sb{\itermtwo}{\ivartwo}{\natit{0}}}
      \sum_{\ivartwo<\sb{\itermsix}{\ivarthree}{\ivarone}}\sb{\itermthree}{\ivartwo}{\itermseven}\\
    &\leq \itermeight\mnu \natit{1}+\sum_{\ivartwo<\itermfive}\itermthree\leq\itermone.
  \end{align*}
\end{varitemize}
This concludes the proof.
\end{proof}
\section{Intensional Soundness}
Subject Reduction already implies an \emph{extensional} notion of soundness for programs: if a term $\termone$ can 
be typed with $\vdash_\itermthree\termone:\Nat[\itermone,\itermtwo]$, then its normal form (if any) is
a natural number between $\semu{\itermone}$ and $\semu{\itermtwo}$. However, Subject Reduction does not tell us whether 
the evaluation of $\termone$ terminates, and in how much time. Has $\itermthree$ anything to do with
the complexity of evaluating $\termone$? The only information that can be extracted from the Subject Reduction
Theorem is that $\itermthree$ does not increase along reduction. 
\begin{figure*}\centering
\fbox{
\begin{minipage}[c]{.97\textwidth}
\begin{center}
{\footnotesize
\begin{tabular}{ccccccc}
\textbf{Term}        & \textbf{Environment} & \textbf{Stack} &     & \textbf{Term} & \textbf{Environment} & \textbf{Stack}\\ \hline
$\termone\termtwo$ &                      
$\envone$ &
$\stone$ &                
$\rt$ &
$\termone$ &
$\envone$ &
$(\termtwo,\envone)\cdot\stone$ \\
$\lambda\varone.\termone$ &                      
$\envone$ &
$\pairone\cdot\stone$ &                
$\rt$ &
$\termone$ &
$\pairone\cdot\envone$ &
$\stone$ \\
$\varone$ &                      
$(\termone_0,\envone_0)\cdots(\termone_n,\envone_n)$ &
$\stone$ &                
$\rt$ &
$\termone_\varone$ &
$\envone_\varone$ &
$\stone$ \\
\hspace{-3mm}$\CASE{\termone}{\termtwo}{\termthree}$ &                      
$\envone$ &
$\stone$ &                
$\rt$ &
$\termone$ &
$\envone$ &
$\casecl{\termtwo}{\termthree}{\envone}\cdot\stone$ \\
$\REC{\varone}{\termone}$ &                      
$\envone$ &
$\stone$ &                
$\rt$ &
$\termone$ &
$(\REC{\varone}{\termone},\envone)\cdot\envone$ &
$\stone$\\
$\val{n}$ &                      
$\envone$ &
$\suc\cdot\stone$ &                
$\rt$ &
$\val{n+1}$ &
$\envone$ &
$\stone$ \\
$\val{n}$ &                      
$\envone$ &
$\pred\cdot\stone$ &                
$\rt$ &
$\val{n\mnu 1}$ &
$\envone$ &
$\stone$ \\
$\val{0}$ &                      
$\envone$ &
$\casecl{\termone}{\termtwo}{\envtwo}\cdot\stone$ &                
$\rt$ &
$\termone$ &
$\envtwo$ &
$\stone$ \\
$\val{n+1}$ &                      
$\envone$ &
$\casecl{\termone}{\termtwo}{\envtwo}\cdot\stone$ &                
$\rt$ &
$\termtwo$ &
$\envtwo$ &
$\stone$ \\
$\suc(\termone)$ &                      
$\envone$ &
$\stone$ &                
$\rt$ &
$\termone$ &
$\envone$ &
$\suc\cdot\stone$ \\
$\pred(\termone)$ &                      
$\envone$ &
$\stone$ &                
$\rt$ &
$\termone$ &
$\envone$ &
$\pred\cdot\stone$ 
\end{tabular}}
\end{center}
\end{minipage}}
\caption{The $\Kld$ Transition Steps.}
\label{fig:kld-machine}
\end{figure*}

In this section, \emph{Intensional Soundness}  (Theorem \ref{thm:int-soundness} below) 
for the type system \PCFld\ will be proved. A Krivine's Machine $\Kld$ for 
\PCF\ programs will be first defined in Section \ref{sec:kam-definition}. Given a 
program (i.e. a closed term of base type), the machine $\Kld$ 
either evaluates it to normal form or diverges. 
A formal connection between the machine $\Kld$ and the type system 
\PCFld\ will be established by means of a weighted typability notion for 
machine configurations, introduced in Section \ref{sec:kam-weights}. 
This notion is the fundamental ingredient to keep track of the number of machine steps.
\subsection{The $\Kld$ Machine}
\label{sec:kam-definition}
The Krivine's Machine has been introduced as a natural device to evaluate pure lambda-terms under a weak-head notion of 
reduction~\cite{Krivine07}. Here, the standard Krivine's Machine is extended to a machine 
$\Kld$ which handles not only abstractions and applications, but also constants, conditionals and fixpoints.

The \emph{configurations} of the machine $\Kld$, ranged over by $\confone,\conftwo,\ldots$, 
are triples $\confone=(\termone,\envone,\stone)$ 
where $\envone$ and $\stone$ are two additional constructions: $\envone$ is an \emph{environment}, 
that is a (possibly empty) finite 
sequence of \emph{closures}; while $\stone$ is a (possibly empty) \emph{stack} of \emph{contexts}. 
Stacks are ranged over by $\stone,\sttwo,\ldots$. 
A \emph{closure}, as usual, is a pair $\pairone=(\termone,\envone)$ where $\termone$ is a
term and $\envone$ is an environment. A \emph{context} is either a closure,  a term $\suc$, a term $\pred$, or 
a triple $(\termtwo, \termthree,\envone)$ where $\termtwo, \termthree$ are terms and $\envone$ is an environment.

The transition steps between configurations of the $\Kld$ machine are given in Figure
\ref{fig:kld-machine}. The transition rules require some comments. First of all, a na\"ive management of name variables is used. 
A more effective description however, could be given by using standard de Bruijn indexes. 
Note that the triple $(\termtwo, \termthree,\envone)$ is used as a context for the conditional construction; moreover, in a 
recursion step, a copy of the recursive term is put in a closure on the top of the current environment.
As usual, the symbol $\rts$ denotes the reflexive and transitive closure of the transition relation $\rt$. 
The relation $\rts$ implements weak-head reduction. Weak-head normal 
form and the normal form coincide for programs. So the machine $\Kld$ is a correct device to evaluate programs. 
For this reason, the notation $\termone\ev \val{n}$ can be used as a shorthand
for $(\termone,\varepsilon,\varepsilon)\rts (\val{n},\envone,\varepsilon)$. Moreover, notations like $\confone\ev^{\natone}$ could also 
be used to stress that $\confone$ reduces to an irreducible configuration in exactly $\natone$ steps. The 
proof of the formal correctness of the abstract machine is outside the scope of this paper, however it should be clear 
that it could be obtained as a simple extension of the original one~\cite{Krivine07}.

Intensional Soundness will be proved by studying how the weight $\itermone$ of any program $\termone$ evolves during
the evaluation of $\termone$ by $\Kld$. This is possible because every reduction step in $\termone$ is decomposed into a number
of transitions in $\Kld$, and this decomposition highlights \emph{when}, precisely, the weight changes. The same
would be more difficult when performing plain reduction on terms. Proving Intensional Soundness this 
way requires, however, to keep track of the types and weights of all objects in a machine configuration. In other words, the
type system should be somehow generalized to an assignment system on \emph{configurations}.

\subsection{Types and Weights for Configurations}
\label{sec:kam-weights}
Assigning types and weights to configurations amounts to somehow keeping track of the nature of all terms
appearing in environments and stacks. This is captured by the rules in Figure~\ref{fig:liftingpcfld}.
\begin{figure*}
\fbox{
\begin{minipage}[c]{.97\textwidth}
\vspace{10pt}
\begin{center}
\textbf{Closures}
\end{center}
\vspace{2pt}
$$
\infer
{\icontextone;\iconstraintone\vdash^\eqpone_{\itermtwo}(\termone,\pairone_1\cdots\pairone_n):\typeone}
{
  \begin{array}{c}
    \icontextone;\iconstraintone;\varone_1:\qbang{\ivarone<\itermone_1}\typetwo_1,\ldots,\varone_n:\qbang{\ivarone<\itermone_n}\typetwo_n
      \vdash_\itermthree\termone:\typeone\\
    \icontextone,\ivarone;\iconstraintone,\ivarone<\itermone_i\vdash^\eqpone_{\itermfour_i}\pairone_i:\typetwo_i\\
    \icontextone;\iconstraintone\models\itermtwo\geq\itermthree+\itermone_1+\ldots+\itermone_n+\sum_{\ivarone<\itermone_1}
      \itermfour_1+\ldots+\sum_{\ivarone<\itermone_n}\itermfour_n.
  \end{array}
}
$$
\vspace{5pt}
\begin{center}
\textbf{Stacks}
\end{center}
\vspace{2pt}
$$
\begin{array}{ccc}
\infer
{\icontextone;\iconstraintone\vdash^\eqpone_\itermtwo\varepsilon:(\typeone,\typetwo)}
{
  \begin{array}{c}
    \icontextone;\iconstraintone\models^\eqpone\itermtwo\geq\natit{0}\\
    \icontextone;\iconstraintone\vdash^\eqpone\typeone\tless\typetwo
  \end{array}
}
&&
\infer
{\icontextone;\iconstraintone\vdash^\eqpone_\itermtwo\pairone\cdot\sttwo:(\qbang{\ivarone<\itermone}\typefour\lin\typethree,\typetwo)}
{
  \begin{array}{c}
    \icontextone;\iconstraintone,\ivarone<\itermone\vdash^\eqpone_{\itermthree}\pairone:\typefour\\
    \icontextone;\iconstraintone\vdash^\eqpone_\itermfour\sttwo:(\typethree,\typetwo)\\
    \icontextone;\iconstraintone\models^\eqpone\itermtwo\geq\itermfour+\sum_{\ivarone<\itermone}\itermthree+\itermone
  \end{array}
}
\end{array}
$$
\vspace{3pt}
$$
\begin{array}{ccc}
\infer
{\icontextone;\iconstraintone\vdash^\eqpone_\itermtwo\suc\cdot\sttwo:(\Nat[\itermone,\itermfive],\typetwo)}
{
  \begin{array}{c}
    \icontextone;\iconstraintone\vdash^\eqpone_\itermtwo\sttwo:(\Nat[\itermthree,\itermfour],\typetwo)\\
    \icontextone;\iconstraintone\vdash^\eqpone \Nat[\itermone+1,\itermfive+1]\tless \Nat[\itermthree, \itermfour]\\
  \end{array}
}
&&
\infer
{\icontextone;\iconstraintone\vdash^\eqpone_\itermtwo\pred\cdot\sttwo:(\Nat[\itermone,\itermfive],\typetwo)}
{
  \begin{array}{c}
    \icontextone;\iconstraintone\vdash^\eqpone_\itermtwo\sttwo:(\Nat[\itermthree,\itermfour],\typetwo)\\
    \icontextone;\iconstraintone\vdash^\eqpone \Nat[\itermone\mnu 1,\itermfive\mnu 1]\tless \Nat[\itermthree, \itermfour]\\
  \end{array}
}
\end{array}
$$
\vspace{3pt}
$$
\infer
{\icontextone;\iconstraintone\vdash^\eqpone_\itermtwo(\termone,\termtwo,\envone)\cdot\sttwo:(\Nat[\itermone,\itermfive],\typetwo)}
{
  \begin{array}{c}
    \icontextone;\iconstraintone,\itermone\leq 0\vdash^\eqpone_\itermthree(\termone,\envone):\typethree\\
    \icontextone;\iconstraintone,\itermfive\geq 1\vdash^\eqpone_\itermthree(\termtwo,\envone):\typethree
  \end{array}
  &
  \begin{array}{c}
    \icontextone;\iconstraintone\vdash^\eqpone_\itermfour\sttwo:(\typethree,\typetwo)\\
    \icontextone;\iconstraintone\models^\eqpone\itermtwo\geq\itermthree+\itermfour
  \end{array}
}
$$
\vspace{5pt}
\begin{center}
\textbf{Configurations}
\end{center}
$$
\infer
{\icontextone;\iconstraintone\vdash^\eqpone_\itermone(\termone,\envone,\stone):\typetwo}
{
  \icontextone;\iconstraintone\vdash^\eqpone_\itermthree(\termone,\envone):\typeone
  &
  \icontextone;\iconstraintone\vdash^\eqpone_\itermtwo\stone:(\typeone,\typetwo)
  &
  \icontextone;\iconstraintone\models^\eqpone\itermone\geq\itermthree+\itermtwo
}
$$
\vspace{-1pt}
\end{minipage}}
\caption{Lifting \PCFld\ Typing to Closures, Stacks and Configurations.}
\label{fig:liftingpcfld}
\end{figure*}
A formal connection between typed terms and typed configurations could be established as expected, 
and such connection could be shown to be preserved by reduction. However, the following lemma 
is everything we need in the sequel:
\begin{lem}
\label{lem:program-typing}
Let $\termone\in \prog$. Then,
$\icontextone;\iconstraintone;\emptyset\vdash^\eqpone_{\itermone} \termone:\typeone$ if and only if 
$\icontextone;\iconstraintone\vdash^\eqpone_\itermone(\termone,\varepsilon,\varepsilon):\typeone$.
\end{lem}
Analogous notions of typability for closures, stacks and configurations
can be given following the simpler type discipline of \PCF\ proper. They can be 
obtained by simplifying those for \PCFld, see Figure~\ref{fig:liftingpcf}.
\begin{figure*}
\fbox{
\begin{minipage}[c]{.97\textwidth}
\vspace{10pt}
\begin{center}
\textbf{Closures}
\end{center}
\vspace{2pt}
$$
\infer
{\vdash(\termone,\pairone_1\cdots\pairone_n):\typeone}
{
  \varone_1:\typetwo_1,\ldots,\varone_n:\typetwo_n\vdash\termone:\typeone
  &
  \vdash\pairone_i:\typetwo_i
}
$$
\vspace{5pt}
\begin{center}
\textbf{Stacks}
\end{center}
\vspace{2pt}
$$
\begin{array}{ccc}
\infer
{\varepsilon:(\typeone,\typeone)}
{}
&&
\infer
{\vdash\pairone\cdot\sttwo:(\typefour\arr\typethree,\typetwo)}
{
  \vdash\pairone:\typefour
  &
  \vdash\sttwo:(\typethree,\typetwo)
}
\end{array}
$$
\vspace{3pt}
$$
\begin{array}{ccccc}
\infer
{\vdash\suc\cdot\sttwo:(\Nat,\typetwo)}
{\vdash\sttwo:(\Nat,\typetwo)}
&&
\infer
{\vdash\pred\cdot\sttwo:(\Nat,\typetwo)}
{\vdash\sttwo:(\Nat,\typetwo)}
&&
\infer
{\vdash(\termone,\termtwo,\envone)\cdot\sttwo:(\Nat,\typetwo)}
{
  \vdash(\termone,\envone):\typethree
  &
  \vdash(\termtwo,\envone):\typethree
  &
  \vdash\sttwo:(\typethree,\typetwo)
}
\end{array}
$$
\vspace{5pt}
\begin{center}
\textbf{Configurations}
\end{center}
$$
\infer
{\vdash(\termone,\envone,\stone):\typetwo}
{
  \vdash(\termone,\envone):\typeone
  &
  \vdash\stone:(\typeone,\typetwo)
}
$$
\vspace{-1pt}
\end{minipage}}
\caption{Extending \PCF\ Typing to Closures, Stacks and Configurations.}
\label{fig:liftingpcf}
\end{figure*}
If $\confone\rt\conftwo$ and $\tdone$ is a derivation of $\vdash\confone:\typeone$, then a derivation
$\tdtwo$ of $\vdash\conftwo:\typeone$ can be easily obtained by manipulating $\tdone$, and
we write $\tdone\rt\tdtwo$.
\subsection{Measure Decreasing and Intensional Soundness}\label{sec:int-sound-theorem}
An important property of Krivine's Machine says that during the evaluation of programs only subterms 
of the initial program are recorded in the environment. This justifies the notion of \emph{size} 
for configurations, denoted $\size{\confone}$, that will be used in the sequel. 
This is defined as  $\size{(\termone,\envone,\stone)}=\size{\termone}+\size{\stone}$. The size
$\size{\stone}$ of a stack $\stone$ is defined as the sum of sizes of its elements,
where $\size{(\termone,\envone)}=\size{\termone}$, $\size{\suc}=\size{\pred}=1$, and
$\size{(\termone,\termtwo,\envone)}=\size{\termone}+\size{\termtwo}$. 
Moreover, another consequence of the same property is the following lemma.

\begin{lem}
\label{lem:size-bound-terms-km}
Let $\termone\in\prog$ and let  $\confone=(\termone,\varepsilon,\varepsilon)$.
Then, for each $\conftwo=(\termtwo,\envone,\stone)$ such that $\confone\rightarrow^*\conftwo$ and for each $\termthree$ occurring in $\envone$ or 
$\stone$, $\size{\termthree}\leq \size{\termone}$.
\end{lem}
\begin{proof}
  Easy, by induction on the length of the reduction
  $\confone\rightarrow^*\conftwo$. In fact, a strengthening of the
  statement is needed for induction to work. In particular, not only $\size{\termthree}\leq \size{\termone}$ for every
  $\termthree$ in $\envone$ and $\stone$, but also for the
  \emph{non-head subterms} of $\termtwo$.
\end{proof}
Intensional Soundness (Theorem \ref{thm:int-soundness}) expresses the fact that for a program $\termone\in\prog$ such that
$\emcon;\emcon;\emcon\vdash_{\itermone}^{\eqpone}\termone:\Nat[\itermtwo,\itermthree]$, the number 
$\semt{\itermone}{\assone}{\eqpone}$ is a good estimate of the number of steps needed to evaluate $\termone$. Moreover, thanks
to Subject Reduction, the numbers $\semt{\itermtwo}{\assone}{\eqpone}$ and $ \semt{\itermthree}{\assone}{\eqpone}$ give an upper 
and a lower bound, respectively, to the result of such an evaluation. This is proved by showing that during reduction a measure, 
expressed as the combination of the weight and the size of a configuration, decreases. In turn, this requires extending some 
of the properties in Section~\ref{sect:prop} from terms to configurations. 
As an example, substitution holds on configurations, too:
\begin{lem}
\label{lem:weigth-typability-properties}
If $\icontextone,\ivarone;\iconstraintone\vdash^\eqpone_\itermfour(\termone,\envone):\typeone$,
then $\icontextone;\sb{\iconstraintone}{\ivarone}{\itermtwo},\iconstrainttwo\vdash^\eqpone_{\sb{\itermfour}{\ivarone}{\itermtwo}}
(\termone,\envone):\sb{\typeone}{\ivarone}{\itermtwo}$ for every $\itermtwo$
such that $\icontextone,\iconstrainttwo\models^{\eqpone}\conv{\itermtwo}$.
\end{lem}
\begin{proof}
  By induction on the proof of $\icontextone,\ivarone;\iconstraintone\vdash^\eqpone_\itermfour(\termone,\envone):\typeone$,
  using Lemma~\ref{lem:I2Term-substitution}.
\end{proof}
Moreover, type derivations for closures can be ``split'', exactly as  terms:
\begin{lem}
  \label{lem:typable-separating-sum}
  Let $\icontextone;\iconstraintone\vdash^\eqpone \qbang{\ivarone<\itermone}{\typeone}\tless \qbang{\ivarone<\itermtwo+\itermthree}{\typetwo}$
  and let  $\icontextone,\ivarone;\iconstraintone,\ivarone <\itermone\vdash^\eqpone_\itermfour(\termone,\envone):\typeone$
  Then, both $(\icontextone,\ivarone;\iconstraintone,\ivarone <\itermtwo)\vdash^\eqpone_\itermfour(\termone,\envone):\typetwo$
  and $\icontextone,\ivarone;\iconstraintone,\ivarone <\itermthree
  \vdash^\eqpone_{\sb{\itermfour}{\ivarone}{\itermtwo+\ivarone}}(\termone,\envone):\sb{\typetwo}{\ivarone}{\itermtwo+\ivarone}$.
\end{lem}
The key step towards Intensional Soundness is the following:
\begin{lem}[Weighted Subject Reduction]
\label{lem:weight-decrease}
Suppose that $(\termone,\varepsilon,\varepsilon)\rightarrow^*\conftwo\rightarrow\confthree$ and let
$\conftwo$ be such that $\icontextone;\iconstraintone\vdash^\eqpone_\itermone\conftwo:\typeone$.
Then $\icontextone;\iconstraintone\vdash^\eqpone_\itermtwo\confthree:\typeone$, and one of the following holds:
\begin{varenumerate}
\item
  $\icontextone;\iconstraintone\models\itermone=\itermtwo$ but
  $\size{\conftwo}>\size{\confthree}$;
\item
  $\icontextone;\iconstraintone\models\itermone>\itermtwo$ and
  $\size{\confthree}<\size{\conftwo}+\size{\termone}$.
\end{varenumerate}
\end{lem}
\begin{proof}
The proof is by cases on the reduction $\conftwo\rt\confthree$. 
Condition $1$ can be shown to apply to all the cases but the one 
in which $\conftwo=(\varone,\envone,\stone)$. 
In that one, weight decreasing relies on 
the side condition in the typing rule for variables, while the bound 
on the size increasing comes from Lemma \ref{lem:size-bound-terms-km}.
We just present some cases, the others can be obtained analogously:
\begin{varitemize}
\item 
  Consider the case $\conftwo\equiv(\CASE{\termfour}{\termtwo}{\termthree},\envone,\stone)$.
  We want to prove Point 1, namely
  that $\confthree\equiv(\termone,\envone,(\termtwo,\termthree,\envone)\cdot\stone)$ is 
  such that $\icontextone;\iconstraintone\vdash_\itermtwo\confthree:\typeone$ where
  $\icontextone;\iconstraintone\models\itermone=\itermtwo$ and $\size{\conftwo}>\size{\confthree}$.
  The latter is immediate:
  \begin{align*}
  \size{\conftwo}&=1+\size{\termfour}+\size{\termtwo}+\size{\termthree}+\size{\stone}>\size{\termfour}+
    (\size{\termtwo}+\size{\termthree})+\size{\stone}\\
  &=\size{\termfour}+\size{(\termtwo,\termthree,\envone)\cdot\stone}=\size{\confthree}.
  \end{align*}
  Let us consider the former. By inspecting a proof of $\icontextone;\iconstraintone\vdash^\eqpone_\itermone\conftwo:\typeone$, we
  can easily derive the following judgments (where $\envone\equiv\pairone_1,\ldots,\pairone_n$):

  {\footnotesize
  \begin{align}
    \icontextone;\iconstraintone;\varone_1:\qbang{\ivarone<\itermthree_1^\termfour}\typethree_1,\ldots,\varone_n:
      \qbang{\ivarone<\itermthree_n^\termfour}\typethree_n
      &\vdash_{\itermone_\termfour}\termfour:\Nat[\itermfour,\itermfive];\label{equ:ifw}\\
    \icontextone;\iconstraintone,\itermfive\leq \natit{0};\varone_1:\qbang{\ivarone<\itermthree_1^{\termtwo\termthree}}
      {\sb{\typethree_1}{\ivarone}{\itermthree_1^{\termfour}+\ivarone}},
      \ldots,\varone_n:\qbang{\ivarone<\itermthree_n^{\termtwo\termthree}}{\sb{\typethree_n}{\ivarone}{\itermthree_n^{\termfour}+\ivarone}}
      &\vdash_{\itermone_{\termtwo\termthree}}\termtwo:{\typetwo};\label{equ:ifu}\\
    \icontextone;\iconstraintone,\itermfour\geq\natit{1};\varone_1:\qbang{\ivarone<\itermthree_1^{\termtwo\termthree}}
      {\sb{\typethree_1}{\ivarone}{\itermthree_1^{\termfour}+\ivarone}},
      \ldots,\varone_n:\qbang{\ivarone<\itermthree_n^{\termtwo\termthree}}{\sb{\typethree_n}{\ivarone}{\itermthree_n^{\termfour}+\ivarone}}
      &\vdash_{\itermone_{\termtwo\termthree}}\termthree:{\typetwo};\label{equ:ifv}\\
    \icontextone,\ivarone;\iconstraintone,\ivarone<\itermthree_i&\vdash_{\itermone_{\pairone_i}}\pairone_i:\typethree_i;\label{equ:ifclos}\\
    \icontextone;\iconstraintone&\vdash_{\itermone_{\stone}}\stone:(\typetwo,\typeone).\label{equ:ifstack}
  \end{align}}
where
  \begin{align}
    \icontextone;\iconstraintone&\vdash\qbang{\ivarone<\itermthree_i}\typetwo_i\tless 
      \qbang{\ivarone<\itermthree^{\termfour}_i}\typethree_i\uplus\qbang{\ivarone<\itermthree^{\termtwo\termthree}_i}
      {\sb{\typethree_i}{\ivarone}{\itermthree_i^\termfour+\ivarone}};\label{equ:iftless}\\
    \icontextone;\iconstraintone&\models\itermone\geq\itermone_\termfour+\itermone_{\termtwo\termthree}+\itermthree_1+\ldots+\itermthree_n+
      \sum_{\ivarone<\itermthree_1}\itermone_{\pairone_1}+\ldots+\sum_{\ivarone<\itermthree_n}\itermone_{\pairone_n}+\itermone_\stone.\label{equ:ifwei}
  \end{align}
  By Lemma~\ref{lem:typable-separating-sum} applied to (\ref{equ:ifclos}) and exploiting (\ref{equ:iftless}), we obtain that
  \begin{align*}
    \icontextone,\ivarone;\iconstraintone,\ivarone<\itermthree^\termfour_i&\vdash_{\itermone_{\pairone_i}}\pairone_i:\typethree_i;\\
    \icontextone,\ivarone;\iconstraintone,\ivarone<\itermthree^{\termtwo\termthree}_i&
      \vdash_{\sb{\itermone_{\pairone_i}}{\ivarone}{\itermthree^\termfour+\ivarone}}\pairone_i:
      \sb{\typethree_i}{\ivarone}{\itermthree^\termfour+\ivarone}.
  \end{align*}
  By way of (\ref{equ:ifw}), (\ref{equ:ifu}) and (\ref{equ:ifv}), we obtain
  \begin{align*}
    \icontextone;\iconstraintone,\itermfive\leq \zero&\vdash_{\itermone_{(\termfour,\envone)}}(\termfour,\envone):\Nat[\itermfour,\itermfive];\\
    \icontextone;\iconstraintone,\itermfour\geq\natit{1}&\vdash_{\itermone_{(\termtwo\termthree,\envone)}}(\termtwo,\envone):\typetwo;\\
    \icontextone;\iconstraintone&\vdash_{\itermone_{(\termtwo\termthree,\envone)}}(\termthree,\envone):\typetwo;
  \end{align*}
  where
  \begin{align*}
    \itermone_{(\termfour,\envone)}&\equiv\itermone_\termfour+\itermthree^\termfour_1+\ldots+\itermthree^\termfour_n+
      \sum_{\ivarone<\itermthree^\termfour_1}\itermone_{\pairone_1}+\ldots+\sum_{\ivarone<\itermthree^\termfour_n}\itermone_{\pairone_n};\\
    \itermone_{(\termtwo\termthree,\envone)}&\equiv\itermone_{\termtwo\termthree}+\itermthree_1^{\termtwo\termthree}+\ldots+
      \itermthree_n^{\termtwo\termthree}+\sum_{\ivarone<\itermthree_1^{\termtwo\termthree}}\sb{\itermone_{\pairone_1}}{\ivarone}{\itermthree_1^\termfour+\ivarone}
      \sum_{\ivarone<\itermthree_n^{\termtwo\termthree}}\sb{\itermone_{\pairone_n}}{\ivarone}{\itermthree_n^\termfour+\ivarone}.
  \end{align*}
  So, by definition and by (\ref{equ:ifstack})  we have that $\icontextone;\iconstraintone\vdash_{\itermone_{\termtwo\termthree}+\itermone_\stone}
  (\termtwo,\termthree,\envone)\cdot\stone:(\Nat[\itermfour,\itermfive],\termone)$.
  Thus, we can conclude
  that $\icontextone;\iconstraintone\vdash_\itermone\confthree:\typeone$ (since from (\ref{equ:ifwei}),
  it easily follows that $\icontextone;\iconstraintone\models\itermone\geq\itermone_{(\termfour,\envone)}+
    \itermone_{(\termtwo\termthree,\envone)}+\itermone_\stone$).
\item 
  Consider the case $\conftwo\equiv(\lambda\varone.\termtwo,\envone,\pairone\cdot\stone)$. 
  We want to prove Point 1, namely that $\confthree=(\termtwo,\pairone\cdot\envone,\stone)$ is 
  such that $\icontextone;\iconstraintone\vdash_\itermtwo\confthree:\typeone$ where
  $\icontextone;\iconstraintone\models\itermone=\itermtwo$ and $\size{\conftwo}>\size{\confthree}$.
  The latter is immediate, so let us consider the former. By inspecting a proof of 
  $\icontextone;\iconstraintone\vdash^\eqpone_\itermone\conftwo:\typeone$, we
  can easily derive the following judgments (where $\envone\equiv\pairone_1,\ldots,\pairone_n$), in
  particular using the Generation Lemma: 
  \begin{align}
    \icontextone;\iconstraintone;\varone_1:\qbang{\ivarone<\itermthree_1}\typethree_1,\ldots,\varone_n:
      \qbang{\ivarone<\itermthree_n}\typethree_n,\varone:\qbang{\ivarone<\itermfour}\typefour
      &\vdash_{\itermone_\termtwo}\termtwo:\typetwo;\label{equ:lamu}\\
    \icontextone,\ivarone;\iconstraintone,\ivarone<\itermthree_i&\vdash_{\itermone_{\pairone_i}}\pairone_i:\typethree_i;\label{equ:lamclos}\\
    \icontextone,\ivarone;\iconstraintone,\ivarone<\itermfour&\vdash_{\itermone_{\pairone}}\pairone:\typefour;\label{equ:lamcloss}\\
    \icontextone;\iconstraintone&\vdash_{\itermone_{\stone}}\stone:(\typetwo,\typeone).\label{equ:lamstack}
  \end{align}
  Moreover:
  $$
  \icontextone;\iconstraintone\vdash\itermone\geq\itermone_\termtwo+\itermthree_1+\ldots+\itermthree_n+
  \sum_{\ivarone<\itermthree_1}\itermone_{\pairone_1}+\ldots+\sum_{\ivarone<\itermthree_n}\itermone_{\pairone_n}+
  \itermfour+\sum_{\ivarone<\itermfour}\itermone_{\pairone}+\itermone_{\stone}.
  $$
  From (\ref{equ:lamu}), (\ref{equ:lamclos}) and (\ref{equ:lamcloss}), we obtain
  $\icontextone;\iconstraintone;\emptyset\vdash_{\itermone_{\pairone\cdot\envone}}(\termtwo,\pairone\cdot\envone):\typetwo$,
  where
  $$
  \itermone_{\pairone\cdot\envone}\equiv\itermone_\termtwo+\itermthree_1+\ldots+\itermthree_n+
  \sum_{\ivarone<\itermthree_1}\itermone_{\pairone_1}+\ldots+\sum_{\ivarone<\itermthree_n}\itermone_{\pairone_n}+
  \itermfour+\sum_{\ivarone<\itermfour}\itermone_{\pairone}.
  $$
  This, together with (\ref{equ:lamstack}) easily yields the thesis.
\item 
  Consider the case $\conftwo\equiv(\val{n},\envone,\suc\cdot\stone)$. 
  Again, we want to prove Point 1, that is $\confthree=(\val{n+1},\envone,\stone)$ is 
  such that $\icontextone;\iconstraintone\vdash_\itermtwo\confthree:\typeone$, where
  $\icontextone;\iconstraintone\models\itermone=\itermtwo$ and $\size{\conftwo}>\size{\confthree}$.
  The latter is easy: 
  $$
  \size{\conftwo}=\size{\val{n}}+\size{\suc\cdot\stone}=2+\size{\stone}+1>1+\size{\stone}=\size{\val{n+1}}+\size{\stone}=\size{\confthree},
  $$  
  so we consider the former. By inspecting a proof of 
  $\icontextone;\iconstraintone\vdash^\eqpone_\itermone\conftwo:\typeone$, we
  can easily derive the following judgments (where $\envone\equiv\pairone_1,\ldots,\pairone_n$) in
  particular using the Generation Lemma: 
  \begin{align}
    \icontextone;\iconstraintone;\varone_1:\qbang{\ivarone<\itermthree_1}\typethree_1,\ldots,\varone_n:
      \qbang{\ivarone<\itermthree_n}\typethree_n
      &\vdash_{\itermone_{\val{n}}}\val{n}:\Nat[\itermfour,\itermfive];\label{equ:natn}\\
    \icontextone,\ivarone;\iconstraintone,\ivarone<\itermthree_i&\vdash_{\itermone_{\pairone_i}}\pairone_i:\typethree_i;\label{equ:natclos}\\
     \icontextone;\iconstraintone&\vdash_{\itermone_{\stone}}\stone:(\Nat[\itermsix,\itermseven],\typeone).\label{equ:natstack}
  \end{align}
  Moreover:
  \begin{align}
    \icontextone;\iconstraintone&\models\itermone\geq\itermone_{\val{n}}+\itermthree_1+\ldots+\itermthree_n+
      \sum_{\ivarone<\itermthree_1}\itermone_{\pairone_1}+\ldots+\sum_{\ivarone<\itermthree_n}\itermone_{\pairone_n}+
      \itermone_{\stone};\label{equ:natweight}\\
    \icontextone;\iconstraintone&\vdash\Nat[\itermfour+1,\itermfive+1]\tless\Nat[\itermsix,\itermseven].\label{equ:natsub}
  \end{align}
  From (\ref{equ:natn}) and (\ref{equ:natsub}), we get
  $$
  \icontextone;\iconstraintone;\varone_1:\qbang{\ivarone<\itermthree_1}\typethree_1,\ldots,\varone_n:
      \qbang{\ivarone<\itermthree_n}\typethree_n
      \vdash_{\itermone_{\val{n}}}\val{n+1}:\Nat[\itermsix,\itermseven].
  $$
  This, together with (\ref{equ:natclos}), allows us to reach
  $\icontextone;\iconstraintone\vdash_{\itermone_{(\val{n+1},\envone)}}(\val{n+1},\envone):\Nat[\itermsix,\itermseven]$,
  where
  $$
  \itermone_{(\val{n+1},\envone)}\equiv\itermone_{\val{n}}+\itermthree_1+\ldots+\itermthree_n+
    \sum_{\ivarone<\itermthree_1}\itermone_{\pairone_1}+\ldots+\sum_{\ivarone<\itermthree_n}\itermone_{\pairone_n}.
  $$
  By (\ref{equ:natstack}), the thesis can be easily reached.
\item 
  Consider the case $\conftwo=(\REC{\varone}{\termtwo},\envone,\stone)$.  
  Yet another time, we want to prove Point 1, that is 
  $\confthree=(\termtwo,(\REC{\varone}{\termtwo},\envone)\cdot\envone,\stone)$ is 
  such that $\icontextone;\iconstraintone\vdash_\itermtwo\confthree:\typeone$, where
  $\icontextone;\iconstraintone\models\itermone=\itermtwo$ and $\size{\conftwo}>\size{\confthree}$.
  The latter is easy, as usual: 
  $$
  \size{\conftwo}=\size{\REC{\varone}{\termtwo}}+\size{\stone}>\size{\termtwo}+\size{\stone}=\size{\confthree},
  $$  
  so we consider the former. By inspecting a proof of 
  $\icontextone;\iconstraintone\vdash^\eqpone_\itermone\conftwo:\typeone$, we
  can easily derive the following judgments (where $\envone\equiv\pairone_1,\ldots,\pairone_n$):

  {
  \begin{align}
    \icontextone,\ivartwo;\iconstraintone,\ivartwo<\itermfour;\varone_1:\qbang{\ivarone<\itermthree_1}\typethree_1,\ldots,\varone_n:
      \qbang{\ivarone<\itermthree_n}\typethree_n,\varone:\qbang{\ivarone<\itermfive}\typefour
      &\vdash_{\itermone_{\termtwo}}\termtwo:\typetwo;\label{equ:fixu}\\
    \icontextone,\ivarone;\iconstraintone,\ivarone<\itermsix_i&\vdash_{\itermone_{\pairone_i}}\pairone_i:\typesix_i;\label{equ:fixclos}\\
     \icontextone;\iconstraintone&\vdash_{\itermone_{\stone}}\stone:(\typefive,\typeone).\label{equ:fixstack}
  \end{align}}
  Moreover:

  {\footnotesize
  \begin{align*}
    \icontextone;\iconstraintone&\vdash\sb{\typetwo}{\ivartwo}{\natit{0}}\tless\typefive;\\
    \icontextone,\ivarone,\ivartwo;\iconstraintone,\ivarone<\itermfive,\ivartwo<\itermfour
    &\vdash\sb{\typetwo}{\ivartwo}{\Dfscomb{\ivartwo}{\ivartwo+\natit{1}}{\ivarone}{\itermfive}+\ivartwo+\natit{1}}\tless\typefour;\\
    \icontextone;\iconstraintone&\vdash\qbang{\ivarone<\itermsix_i}\typesix_i\tless
      \sum_{\ivartwo<\itermfour}\qbang{\ivarone<\itermthree_i}\typethree_i;\\
    \icontextone;\iconstraintone&\models\Dfscomb{\ivartwo}{\natit{0}}{\natit{1}}{\itermfive}\leq\itermfour,\itermseven;\\
    \icontextone;\iconstraintone&\models\itermone\geq\itermseven\mnu\natit{1}+\sum_{\ivartwo<\itermfour}\itermone_{\termtwo}+
      \itermsix_1+\ldots+\itermsix_n+\sum_{\ivarone<\itermsix_1}\itermone_{\pairone_1}+\ldots+
      \sum_{\ivarone<\itermsix_n}\itermone_{\pairone_n}+\itermone_{\stone}.
  \end{align*}}
  By manipulations of the indices similar to the one used in the proof of Subject Reduction, we can derive
  the following from (\ref{equ:fixu}), given the judgments above:
  \begin{align*}
    \icontextone;\iconstraintone;\tcontextone,
      \varone:\qbang{\ivarone<\sb{\itermfive}{\ivartwo}{\natit{0}}}{\sb{\typefour}{\ivartwo}{\natit{0}}}
      &\vdash_{\sb{\itermone_{\termtwo}}{\ivartwo}{\natit{0}}}\termtwo:\typefive;\\
    \icontextone;\iconstraintone,\ivarone<\sb{\itermfive}{\ivartwo}{\natit{0}};\tcontexttwo
      &\vdash_{\sb{\itermeight}{\ivarthree}{\ivarone}
      \mnu \natit{1}+\sum_{\ivartwo<\sb{\itermeight}{\ivarthree}{\ivarone}}\sb{\itermone_\termtwo}{\ivartwo}{\itermnine}} 
      \REC{\varone}{\termtwo}:\sb{\typefour}{\ivartwo}{\natit{0}}.
  \end{align*}
  In the equations above,
  \begin{align*}
    \itermeight&\equiv\Dfscomb{\ivartwo}{\natit{0}}{\natit{1}}{\sb{\itermfive}{\ivartwo}{\ivartwo+\natit{1}+\Dfscomb{\ivartwo}{\natit{0}}{\ivarthree}{\itermfive}}};\\
    \itermnine&\equiv \natit{1}+\ivartwo+\sum_{\ivarthree<\ivarone}\itermeight;
  \end{align*}
  and $\tcontextone,\tcontexttwo$ can be chosen in such a way as to guarantee:
  \begin{align*}
    \icontextone;\iconstraintone\vdash&\varone_1:\qbang{\ivarone<\itermsix_1}\typesix_1,\ldots,\varone_n:\qbang{\ivarone<\itermsix_n}\typesix_n
    \teq\sum_{\ivarone<\sb{\itermfive}{\ivartwo}{\natit{0}}}\tcontexttwo\uplus\tcontextone\\
    &\tless\varone_1:\sum_{\ivartwo<\itermfour}\qbang{\ivarone<\itermthree_1}\typethree_1,\ldots,\varone_n:\sum_{\ivartwo<\itermfour}
      \qbang{\ivarone<\itermthree_n}\typethree_n.
  \end{align*}
  So we have that
  $$
  \icontextone;\iconstraintone\vdash_{\itermone_{(\termtwo,(\REC{\varone}{\termtwo},\envone)\cdot\envone)}}(\termtwo,(\REC{\varone}{\termtwo},\envone)\cdot\envone):
  \typefive,
  $$
  where
  \begin{align*}
    \itermone_{(\termtwo,(\REC{\varone}{\termtwo},\envone)\cdot\envone)}\equiv&\;
     \sb{\itermone_\termtwo}{\ivartwo}{\natit{0}}+\sb{\itermfive}{\ivartwo}{\natit{0}}+\sum_{\ivarone<\sb{\itermfive}{\ivartwo}{\natit{0}}}
     (\sb{\itermeight}{\ivarthree}{\ivarone}\mnu \natit{1}+\sum_{\ivartwo<\sb{\itermeight}{\ivarthree}{\ivarone}}
      \sb{\itermone_\termtwo}{\ivartwo}{\itermseven})\\ 
       &\qquad+\itermsix_1+\ldots+\itermsix_n+\sum_{\ivarone<\itermsix_1}\itermone_{\pairone_1}+\ldots+\sum_{\ivarone<\itermsix_n}\itermone_{\pairone_n}.
  \end{align*}
  The value of $\itermone_{(\termtwo,(\REC{\varone}{\termtwo},\envone)\cdot\envone)}$ can then 
  be proved to be equal or smaller than
  $$
  \itermseven\mnu\natit{1}+\sum_{\ivartwo<\itermfour}\itermone_\termtwo+\itermsix_1+\ldots+\itermsix_n+
  \sum_{\ivarone<\itermsix_1}\itermone_{\pairone_1}+\ldots+\sum_{\ivarone<\itermsix_n}\itermone_{\pairone_n},
  $$
  under the hypotheses in $\icontextone$. This immediately yields the thesis, given (\ref{equ:fixstack}).
\item 
  Consider the case $\conftwo=(\varone_m,((\termone_0,\envone_0),\ldots,(\termone_n,\envone_n)),\stone)$.
  We want to prove Point 2, that is $\confthree=(\termone_{m},\envone_{m},\stone)$ is 
  such that $\icontextone;\iconstraintone\vdash_\itermtwo\confthree:\typeone$, where
  $\icontextone;\iconstraintone\models\itermone>\itermtwo$ and 
  $\size{\confthree}<\size{\conftwo}+\size{\termone}$. The latter is immediate by Lemma \ref{lem:size-bound-terms-km},
  so we consider the former. By inspecting a proof of 
  $\icontextone;\iconstraintone\vdash^\eqpone_\itermone\conftwo:\typeone$, we
  can easily derive the following judgments
  \begin{align}
    \icontextone;\iconstraintone;\varone_1:\qbang{\ivarone<\itermthree_1}\typethree_1,\ldots,\varone_n:
      \qbang{\ivarone<\itermthree_n}\typethree_n&\vdash_{\itermone_{\varone_m}}\varone_m:\typetwo\label{equ:varx};\\
    \icontextone,\ivarone;\iconstraintone,\ivarone<\itermthree_i&\vdash_{\itermone_{(\termone_i,\envone_i)}}(\termone_i,\envone_i):
      \typethree_i\label{equ:varclosu};\\
    \icontextone;\iconstraintone&\vdash_{\itermone_{\stone}}\stone:(\typetwo,\typeone).\label{equ:varstack}
  \end{align}
  Moreover:
  \begin{align}
    \icontextone;\iconstraintone&\models\itermthree_m\geq\natit{1};\label{equ:varder}\\
    \icontextone;\iconstraintone&\vdash\sb{\typethree_m}{\ivarone}{\natit{0}}\tless\typetwo;\label{equ:varsub}\\
    \icontextone;\iconstraintone&\vdash\itermone\geq\itermone_{\varone_m}+\itermthree_1+\ldots+\itermthree_n+
      \sum_{\ivarone<\itermthree_1}\itermone_{(\termone_1,\envone_1)}+\ldots+\sum_{\ivarone<\itermthree_n}\itermone_{(\termone_n,\envone_n)}+
      \itermone_{\stone}.\label{equ:varweight}
  \end{align}
  From (\ref{equ:varclosu}) where $i=m$, (\ref{equ:varder}), and (\ref{equ:varsub}), one obtains that
  $\icontextone;\iconstraintone\vdash_{\sb{\itermone_{(\termone_m,\envone_m)}}{\ivarone}{\natit{0}}}(\termone_m,\envone_m):\typetwo\label{equ:varclos}$
  and, by (\ref{equ:varstack}), that 
  $$
  \icontextone;\iconstraintone\vdash_{\sb{\itermone_{(\termone_m,\envone_m)}}{\ivarone}{\natit{0}}+\itermone_\stone}\confthree:\typeone.
  $$
  But from (\ref{equ:varweight}) and (\ref{equ:varder}) one easily infer that
  $$
  \icontextone;\iconstraintone\models\itermone>\sb{\itermone_{(\termone_m,\envone_m)}}{\ivarone}{\natit{0}}+\itermone_\stone,
  $$
  that is the thesis.
\end{varitemize}
This concludes the proof.
\end{proof}
It is worth noticing that if $\iconstraintone$ is inconsistent,  
the inequality $\icontextone;\iconstraintone\models\itermone>\itermtwo$ 
in Lemma~\ref{lem:weight-decrease}, Point 2, does not necessary imply that weight strictly decreases. Indeed, Intensional
Soundness only holds in presence of a consistent set of constraints:

\begin{thm}[Intensional Soundness]
\label{thm:int-soundness}
Let $\vdash_{\itermone} \termone :\Nat[\itermtwo,\itermthree]$ 
and $\termone\ev^{n}\val{m}$.
Then, $n\leq \size{\termone} \cdot (\semu{\itermone}+1)$.
\end{thm}
\begin{proof}
  By induction on $n$, making essential use of Lemma \ref{lem:weight-decrease} and Lemma \ref{lem:size-bound-terms-km}.
\end{proof}
Please observe that an easy consequence of Theorem~\ref{thm:int-soundness} is intensional soundness \emph{for functions}. As an
example, if $\ivarone;\emcon;\emcon\vdash_{\itermone} \termone :\qbang{\ivartwo<\itermtwo}
\Nat[\ivarone]\lin\Nat[\itermthree,\itermfour]$, then the complexity of evaluating $\termone\;\val{n}$ is
at most $(\size{\termone\;\val{n}})\cdot(\semu{\sb{\itermone}{\ivarone}{\natit{n}}}+1)$. Observe, however, that $\size{\termone\;\val{n}}$
does not depend on $n$, since $\size{\val{n}}=1$.
\section{Relative Completeness}
\label{sec:relative-completeness}
This section is devoted to proving \emph{relative completeness} for the type system \PCFld.
In fact, \emph{two} relative completeness theorems will be presented. The first one 
(Theorem \ref{thm:program-relative-completeness}) states relative completeness \emph{for programs}:
for each \PCF\ program $\termone$ that evaluates to a numeral $\val{\num}$  there is a type derivation  
in \PCFld\ whose index terms capture both the number of reduction steps and the value of $\val{\num}$.
The second one (Theorem \ref{thm:function-relative-completeness}) states relative completeness \emph{for functions}:
for each \PCF\ term $\termone:\Nat\arr\Nat$ computing a \emph{total} function $\funone$ in time expressed by 
a function $\funtwo$ there exists a type derivation  in \PCFld\ whose index terms capture both the extensional 
behavior $\funone$ and the intensional property embedded into $\funtwo$.

Relative completeness does not hold in general. Indeed, it holds only when  the underlying equational program 
$\eqpone$ is \emph{universal}, i.e. when it is sufficiently expressive as to encode all total computable 
functions. A universal equational program is introduced in Section \ref{sec:universal-eq-prog}.

Relative completeness for programs will be proved using a weighted form of 
\emph{Subject Expansion} (Theorem \ref{thm:weighted-subject-expansion}) similar
to the one holding in intersection type theories. This will be proved in 
Section \ref{sec:subject-expansion}. The proof of relative completeness for
functions needs a further step: a \emph{uniformization} result (Lemma \ref{lem:uniformizing-typing}) 
relying on the properties of the universal model. This is the subject of Section \ref{sec:uniformization}.
\subsection{Universal Equational Program}
\label{sec:universal-eq-prog}
Since the class of equational programs is clearly recursively enumerable, it can be put
in one-to-one correspondence with natural numbers, using a coding scheme $\ulcorner \cdot \urcorner$ \emph{\`a la G\"odel}. 
Such a coding, as usual, can be used to define a \emph{universal equational program} $\eqpun$ that is able to simulate 
all equational programs (including itself). 

Let $\gn{\eqpone}{\funsymone}$ be the natural number coding an equational program
$\eqpone$ and a function symbol $\funsymone$ among the ones defined in it. This  can be easily 
computed from (a description of) $\eqpone$ and $\funsymone$.
 A signature $\sigun$  containing just the symbol $\emfs$ of arity $0$ and the symbols $\pairfs$ and 
$\evalfs$ of arity $2$ (plus some auxiliary symbols) is sufficient to define the universal program $\eqpun$.
For each $\funsymone$ of arity $\natone$, the equational program $\eqpun$ satisfies
$$
\semt{\evalfs(\gn{\eqpone}{\funsymone},\pairing_{\natone}(x_1,\ldots,x_\natone))}{\assone}{\eqpun}=
\semt{\funsymone(x_1,\ldots,x_\natone)}{\assone}{\eqpone},
$$
where $\pairing_\natone(\termone_1,\ldots,\termone_\natone)$ is
defined by induction on $\natone$:
\begin{align*}
  \pairing_0&\equiv\emfs;\\
  \pairing_{\natone+1}(\termone_1,\ldots,\termone_{\natone+1})&\equiv\pairfs(\pairing_{\natone}(\termone_1,\ldots,\termone_{\natone}),\termone_{\natone+1}).
\end{align*}
This way, $\eqpun$ acts as an interpreter for any equational program. Such a universal
program $\eqpun$ can be defined as a \emph{finite} sequence of equations, similarly to what
happens in the construction of, e.g., universal Turing machines.

The universal equational program  $\eqpun$ enjoys some nice properties 
which are crucial when proving Subject Expansion. The following lemma says, for example, 
that sums and bounded sums can always be formed (modulo $\cong$) whenever index terms are built and reasoned
about using the universal program:
\medskip 
\begin{lem}
\label{lem:obtaining-sum-equivalence}
  \begin{varenumerate}
  \item
    For every $\mtypeone$ and $\mtypetwo$ such that 
    $\icontextone;\iconstraintone\vdash^\eqpun\conv{\mtypeone}$,
    $\icontextone;\iconstraintone\vdash^\eqpun\conv{\mtypetwo}$, and
    $\TtoNDT{\mtypeone}=\TtoNDT{\mtypetwo}$,
    there are $\mtypethree$ and $\mtypefour$
    such that $\icontextone;\iconstraintone\vdash^{\eqpun}\mtypethree\teq
    \mtypeone$, $\icontextone;\iconstraintone\vdash^{\eqpun}\mtypefour\teq
    \mtypetwo$ and  $\mtypethree\uplus\mtypefour$ is defined.
  \item
    For every $\mtypeone$ and $\itermone$ such
    that $\icontextone,\ivarone;\iconstraintone,\ivarone<\itermone\vdash^{\eqpun}\conv{\mtypeone}$
    and $\icontextone;\iconstraintone\vdash^{\eqpun}\conv{\itermone}$,
    there is $\mtypetwo$ such that $\icontextone,\ivarone;\iconstraintone,\ivarone<\itermone\vdash^{\eqpun}\mtypetwo\teq
    \mtypeone$ and $\sum_{\ivarone<\itermone}\mtypetwo$ is defined.
  \end{varenumerate}
\end{lem}
\begin{proof}
These are inductions on the structure of the involved formulas. Actually, it
is convenient to enrich the statements above (which only deals with \emph{modal}
types) with similar statements involving \emph{basic} types, this way facilitating
the inductive argument.
\end{proof}

\subsection{Subject Expansion and Relative Completeness for Programs }
\label{sec:subject-expansion}
\emph{Weighted Subject Expansion} (Theorem \ref{thm:weighted-subject-expansion} below)
says that typing is preserved while weights increase by at most one along any
$\Kld$ expansion step. This is somehow the converse of Weighted Subject Reduction.
Weighted Subject Expansion, however, does not hold in general but only when the underlying
equational program is universal.

In order to prove Weighted Subject Expansion, only typing that carry precise information should be considered.
As an example, we write $\icontextone;\iconstraintone\pvdash_{\itermone}\confone:\typeone$ if we can derive
$\icontextone;\iconstraintone\vdash_{\itermone}\confone:\typeone$ by \emph{precise} type derivations. 
The type of a precisely-typable configuration, in other words, carries exact information about the value of the objects at hand.
One can easily extend the above notation to type derivations for closures and stacks. Recall that a precise type derivation is a type derivation 
such that all premises in the form $\typeone\tless\typetwo$ (respectively, in the form $\itermone\leq\itermtwo$) are actually required
to be in the form $\typeone\cong\typetwo$ (respectively, $\itermone=\itermtwo$).

Furthermore, only specific typing transformations should be considered, namely those that leave the weight information
unaltered. In order to achieve this, some properties of precise typability for the $\Kld$ machine should be exploited. 
As an example, if a closure $\icontextone;\iconstraintone\pvdash_{\itermone}(\termone,\envone):\typeone$,
then $\icontextone;\iconstraintone\pvdash_{\itermtwo}(\termone,\envone):\typetwo$
whenever $\typetwo$ and $\itermtwo$ such that $\icontextone;\iconstraintone\vdash \typeone\teq\typetwo$
and $\icontextone;\iconstraintone\models\itermone =\itermtwo$. This is a natural variation on the
Subtyping Lemma for terms (Lemma~\ref{lem:proof-strenghtening}). 

Finally, it is worth noticing that  by considering an inconsistent set of constraints $\iconstraintone$,
it is possible to make any closure $(\termone,\envone)$  typable with type $\typeone$ (in the sense of \PCF) to be also
typable in the sense of \PCFld: $\icontextone;\iconstraintone\pvdash_\itermone(\termone,\envone):\typetwo$
whenever $\TtoNDT{\typetwo}=\typeone$ and for every index term $\itermone$.
This says that inconsistent sets cover a role similar to the $\omega$-rule in intersection type systems.

The following two lemmas will be useful in the sequel, and allow to ``join'' apparently uncorrelated typing
judgements into one:
\begin{lem}
\label{lem:typable-unifying-sum}
Let $\sbstone$ be the substitution $\sbst{\ivarone}{\ivarone+\itermone}$.
Suppose that $\tdone\prov\icontextone,\ivarone;\iconstraintone,\ivarone<\itermone\pvdash_\itermfour\pairone:\typeone$,
that $\tdtwo\prov\icontextone,\ivarone;\sbt{\iconstraintone}{\sbstone},\ivarone <\itermtwo\pvdash_{\sbt{\itermfour}{\sbstone}}
\pairone:\sbt{\typeone}{\sbstone}$, and that $\TtoNDT{\tdone}=\TtoNDT{\tdtwo}$.
Then, $\icontextone,\ivarone;\iconstraintone,\ivarone <\itermone+\itermtwo\pvdash_\itermfour\pairone:\typeone$.
\end{lem}
\begin{proof}
  By simultaneous induction on $\tdone$ and $\tdtwo$. We make essential use of the implicit
  assumption about the universality of the underlying equational program.
\end{proof}
\begin{lem}
\label{lem:typable-unifying-bounded-sum}
Let $\sbstone$ be the substitution $\sbst{\ivarthree}{\sum_{\ivarthree<\ivarone}\sb{\itermtwo}{\ivarone}{\ivarthree} + \ivartwo}$.
Suppose that $\tdone\prov\icontextone,\ivarone,\ivartwo;\sbt{\iconstraintone}{\sbstone},\ivarone<\itermone,\ivartwo <\itermtwo
\pvdash_{\sbt{\itermfour}{\sbstone}}\pairone:\sbt{\typeone}{\sbstone}$.
Then, $\icontextone,\ivarone;\iconstraintone,\ivarthree <\sum_{\ivarone<\itermone}\itermtwo\pvdash_\itermfour\pairone:\typeone$.
\end{lem}
\begin{proof}
  By induction on the derivation $\tdone$, again using the properties of a universal equational program.  
\end{proof}
But there are even other ways to turn two typing derivations into a more general one, again relying on the
semantic nature of \PCFld:
\begin{lem}
Suppose that $\tdone\prov\icontextone;\iconstraintone,\itermone\leq\itermtwo\pvdash_\itermthree\pairone:\typeone$,
that $\tdtwo\prov\icontextone;\iconstraintone,\itermone>\itermtwo\pvdash_{\itermthree}
\pairone:\typeone$, and that $\TtoNDT{\tdone}=\TtoNDT{\tdtwo}$.
Then, $\icontextone;\iconstraintone\pvdash_\itermthree\pairone:\typeone$.
\end{lem}
It is now time to state Weighted Subject Expansion, since all the necessary ingredients have been introduced:
\begin{thm}[Weighted Subject Expansion]
\label{thm:weighted-subject-expansion}
Suppose that $\tdone\prov\icontextone;\iconstraintone\pvdash_\itermone\conftwo:\typeone$
and that $\tdtwo\rt\TtoNDT{\tdone}$, where $\tdtwo\prov\vdash\confone:\TtoNDT{\typeone}$.
Then $\tdthree\prov\icontextone;\iconstraintone\pvdash_\itermtwo\confone:\typeone$, 
where $\icontextone;\iconstraintone\models \itermtwo\leq \itermone+\natit{1}$ and 
$\TtoNDT{\tdthree}=\tdtwo$. Moreover, $\tdthree$
can be effectively computed from $\tdone$ and $\tdtwo$.
\end{thm}
\begin{proof}
The proof is by cases on the shape of the reduction $\confone\to\conftwo$. We just present some cases, the others can be obtained analogously.
\begin{varitemize}
  \item 
    Consider the case 
    $$
    \confone\equiv(\val{0},\envone,(\termone,\termtwo,\envtwo)\cdot\stone)\to (\termone, \envtwo,\stone)\equiv\conftwo.
    $$  
    By assumption we have that $\confone$ is typable in \PCF\ and that
    $\icontextone;\iconstraintone\pvdash_\itermone\conftwo:\typeone$. So, we have that
    \begin{align*}
      \icontextone;\iconstraintone&\pvdash_{\itermone_{(\termone,\envtwo)}}(\termone,\envtwo):\typetwo;\\
      \icontextone;\iconstraintone&\pvdash_{\itermone_\stone}\stone:(\typetwo,\typeone);\\
      \icontextone;\iconstraintone&\models\itermone=\itermone_{(\termone,\envtwo)}+\itermone_\stone;
    \end{align*}
    for some $\itermone_{(\termone,\envtwo)}$ and $\itermone_\stone$. We clearly also have 
    that $\icontextone;\iconstraintone,\natit{0}\leq \natit{0}\pvdash_{\itermone_{(\termone,\envtwo)}}(\termone,\envtwo):\typetwo$.
    $\iconstraintone,\natit{1}\leq \natit{0}$ is an inconsistent set of constraints, and since  $\confone$ is typable in \PCF\ (as remarked above),
    we also have that $\icontextone;\iconstraintone,\natit{1}\leq\natit{0}\pvdash_{\itermone_{(\termone,\envtwo)}}(\termtwo,\envtwo):\typetwo$.
    This implies, in particular, that $\icontextone,\iconstraintone\pvdash_{\itermone}(\termone,\termtwo,\envtwo)\cdot\stone:(\Nat[\natit{0}],\typeone)$.
    Now, assume that $\envone=(\termone_1,\envone_1)\cdot\ldots\cdot(\termone_n,\envone_n)$ where for every $1\leq i\leq n$,
    $(\termone_i,\envone_i)$ is typable in \PCF. Since $\iconstraintone,\ivarone<\natit{0}$ is inconsistent, we have that
    $$
    \icontextone,\ivarone;\iconstraintone,\ivarone<\natit{0}\pvdash_{\natit{0}}(\termone_i,\envone_i):\typethree_i
    $$
    for some $\typethree_i$. By Lemma \ref{lem:context-weakening} we can build a derivation for
    $$
    \icontextone;\iconstraintone; 
    \varone_1:\qbang{\ivarone<\natit{0}}{\typethree_1},\ldots, \varone_n:\qbang{\ivarone<\natit{0}}{\typethree_n}
    \pvdash_{\natit{0}} \val{0}:\Nat[\natit{0}].
    $$
    So, we have that 
    $$
    \icontextone;\iconstraintone\pvdash_{\natit{0}}(\val{0},\envone):\Nat[\natit{0}].
    $$
    Summing up, we obtain that 
    $$
    \icontextone;\iconstraintone\pvdash_\itermone\confone:\typeone,
    $$
    from which the thesis easily follows, since $\icontextone;\iconstraintone\models\itermone\leq\itermone+\natit{1}$.
  \item 
    Consider the case 
    $$
    \confone\equiv(\lambda\varone.\termone,\envone,\pairone\cdot\stone)\to (\termone, \pairone\cdot\envone,\stone)\equiv\conftwo.
    $$  
    By assumption we have that $\confone$ is typable in \PCF\ and that 
    $\icontextone;\iconstraintone\pvdash_\itermone\conftwo:\typeone$. So, we have that 
    \begin{align*}
      \icontextone;\iconstraintone;\varone_1:\qbang{\ivarone<\itermthree_1}{\typetwo_1},\ldots,\varone_n:\qbang{\ivarone<\itermthree_n}{\typetwo_n}
        &\pvdash_{\itermone_\termone}\termone:\typethree;\\
      \icontextone,\ivarone;\iconstraintone,\ivarone<\itermthree_i&\pvdash_{\itermone_{\pairone_i}}\pairone_i:\typetwo_i;\\
      \icontextone;\iconstraintone&\pvdash_{\itermone_\stone}\stone:(\typethree,\typeone);
    \end{align*}
    where: 
    $$
      \icontextone;\iconstraintone\models\itermone=\itermone_\termone+\itermthree_1+\ldots+\itermthree_n+\sum_{\ivarone<\itermthree_1}\itermone_{\pairone_1}+
        \ldots+\sum_{\ivarone<\itermthree_n}\itermone_{\pairone_n}+\itermone_\stone.
    $$
    For simplicity and without loosing any generality, we can consider the case where 
    $\pairone\cdot\envone\equiv\pairone_1\ldots\pairone_n$ with $\varone\equiv\varone_1$ and $\pairone\equiv\pairone_1$. So, in particular we 
    can build a derivation ending as follows:
    $$
    \infer
    {
      \icontextone;\iconstraintone;\varone_2:\qbang{\ivarone<\itermthree_2}{\typetwo_2},\ldots,\varone_n:\qbang{\ivarone<\itermthree_n}{\typetwo_n}
      \pvdash_{\itermone_\termone}\lambda\varone_1.\termone:\qbang{\ivarone<\itermthree_1}{\typetwo_1}\lin\typethree
    }
    {
      \icontextone;\iconstraintone;\varone_1:\qbang{\ivarone<\itermthree_1}{\typetwo_1},\ldots,\varone_n:\qbang{\ivarone<\itermthree_n}{\typetwo_n}
        \pvdash_{\itermone_\termone}\termone:\typethree
    }
    $$
    and thus we have that $\icontextone;\iconstraintone\pvdash_{\itermone_{(\lambda\varone.\termone,\envone)}}(\lambda\varone.\termone,\envone):
    \qbang{\ivarone<\itermthree_1}{\typetwo_1}\lin\typethree$, where
    $$
    \itermone_{(\lambda\varone.\termone,\envone)}\equiv\itermone_\termone+\itermthree_2+\ldots+\itermthree_n+\sum_{\ivarone<\itermthree_2}\itermone_{\pairone_2}+
        \ldots+\sum_{\ivarone<\itermthree_n}\itermone_{\pairone_n}.
    $$
    Further, we have that 
    $$
    \icontextone;\iconstraintone\pvdash_{\itermone_\stone+\itermthree_1+\sum_{\ivarone<\itermthree_1}\itermone_{\pairone_1}}\pairone_1\cdot\stone:
    (\qbang{\ivarone<\itermthree_1}{\typetwo_1}\lin\typethree,\typeone)
    $$
    and, as an easy consequence, that
    $$
    \icontextone;\iconstraintone\pvdash_{\itermone_{(\lambda\varone.\termone,\envone)}+\itermone_\stone+\itermthree_1+\sum_{\ivarone<\itermthree_1}\itermone_{\pairone_1}}\confone:\typeone.
    $$
    This easily leads to the conclusion, since
    \begin{align*}
    \icontextone;\iconstraintone\models\itermone&=\itermone_\termone+\itermthree_1+\ldots+\itermthree_n+\sum_{\ivarone<\itermthree_1}\itermone_{\pairone_1}+
        \ldots+\sum_{\ivarone<\itermthree_n}\itermone_{\pairone_n}+\itermone_\stone\\
        &=\itermone_{(\lambda\varone.\termone,\envone)}+\itermone_\stone+\itermthree_1+\sum_{\ivarone<\itermthree_1}\itermone_{\pairone_1}.
    \end{align*}
 \item 
   Consider the case 
   $$
   \confone\equiv(\REC{\varone}{\termone},\envone,\stone)\to (\termone, (\REC{\varone}{\termone},\envone)\cdot\envone,\stone)\equiv\conftwo.
   $$  
   By assumption we have that $\confone$ is typable in \PCF\ and that 
    $\icontextone;\iconstraintone\pvdash_\itermone\conftwo:\typeone$. So, we have that 
    \begin{align}
      \icontextone;\iconstraintone;\varone_1:\qbang{\ivarone<\itermthree_1}{\typetwo_1},\ldots,\varone_n:\qbang{\ivarone<\itermthree_n}{\typetwo_n}
        &\pvdash_{\itermone_\termone}\termone:\typethree;\label{equ:fixtfirst}\\
      \icontextone,\ivarone;\iconstraintone,\ivarone<\itermthree_i&\pvdash_{\itermone_{\pairone_i}}\pairone_i:\typetwo_i;\label{equ:fixcifirst}\\
      \icontextone;\iconstraintone&\pvdash_{\itermone_\stone}\stone:(\typethree,\typeone);
    \end{align}
    where:
    \begin{equation}\label{equ:fixweightI}
      \icontextone;\iconstraintone\models\itermone=\itermone_\termone+\itermthree_1+\ldots+\itermthree_n+\sum_{\ivarone<\itermthree_1}\itermone_{\pairone_1}+
        \ldots+\sum_{\ivarone<\itermthree_n}\itermone_{\pairone_n}+\itermone_\stone.
    \end{equation}
    For simplicity and without losing any generality, we can consider
    the case where 
$(\REC{\varone}{\termone},\envone)\cdot\envone\equiv\pairone_1\ldots\pairone_n$
with 
    $\varone\equiv\varone_1$ and $(\REC{\varone}{\termone},\envone)\equiv\pairone_1$. As a consequence, we can conclude that:
    \begin{align}
      \icontextone,\ivarone;\iconstraintone,\ivarone<\itermthree_1;\tcontextone&\pvdash_{\itermone_{\REC{\varone}{\termone}}}
      \REC{\varone}{\termone}:\typetwo_1;\label{equ:fixfix}\\
      \icontextone,\ivarone,\ivartwo;\iconstraintone,\ivarone<\itermthree_1,\ivartwo<\itermfour_i&\pvdash_{\itermtwo_{\pairone_i}}\pairone_i:\typethree_i;\label{equ:fixcisecond}
    \end{align}
    where $\tcontextone\equiv\varone_2:\qbang{\ivartwo<\itermfour_2}{\typethree_2},\ldots,\varone_n:\qbang{\ivartwo<\itermfour_n}{\typethree_n}$, and
    \begin{equation}\label{equ:fixweightII}
    \icontextone,\ivarone;\iconstraintone,\ivarone<\itermthree_1\models\itermone_{\pairone_1}=\itermone_{\REC{\varone}{\termone}}+
      \itermfour_2+\ldots+\itermfour_n+\sum_{\ivartwo<\itermfour_2}\itermtwo_{\pairone_2}+\ldots+\sum_{\ivartwo<\itermfour_n}\itermtwo_{\pairone_n}.
    \end{equation}
    Our objective now is to prove that
    \begin{equation}\label{equ:fixconcl}
    \icontextone,\iconstraintone\pvdash_{\itermone_{(\REC{\varone}{\termone},\envone)}}(\REC{\varone}{\termone},\envone):\typethree,
    \end{equation}
    where $\icontextone,\iconstraintone\models\itermone_{(\REC{\varone}{\termone},\envone)}=\itermone\mnu\itermone_\stone$.
    The thesis easily follows from (\ref{equ:fixconcl}). To do that, we proceed by spelling out what the premises of (\ref{equ:fixfix})
    are. They are:
    \begin{equation}\label{equ:fixtsecond}
      \icontextone,\ivarone,\ivartwo;\iconstraintone,\ivarone<\itermthree_1,\ivartwo<\Dfscomb{\ivartwo}{\natit{0}}{\natit{1}}{\itermeight};
        \varone:\qbang{\ivarthree<\itermeight}{\sb{\typefour}{\ivartwo}{\Dfscomb{\ivartwo}{\ivartwo+1}{\ivarthree}{\itermeight}+\ivartwo+\natit{1}}},
        \tcontexttwo\pvdash_{\itermtwo_\termone}\termone: \typefour,
    \end{equation}
    and the following two:
    \begin{align*}
      \icontextone,\ivarone;\iconstraintone,\ivarone<\itermthree_1&\pvdash\typetwo_1\teq\sb{\typefour}{\ivartwo}{\natit{\natit{0}}};\\
      \icontextone,\ivarone;\iconstraintone,\ivarone<\itermthree_1&\pvdash\tcontextone\teq\sum_{\ivartwo<\Dfscomb{\ivartwo}{\natit{0}}{\natit{1}}{\itermeight}}\tcontexttwo;
    \end{align*}
    where $\itermeight$ and $\itermtwo_\termone$ are index terms such that
    \begin{equation}\label{equ:fixweightIII}
      \icontextone,\ivarone;\iconstraintone,\ivarone<\itermthree_1\models\itermone_{\REC{\varone}{\termone}}=
         \Dfscomb{\ivartwo}{\natit{0}}{\natit{1}}{\itermeight}\mnu\natit{1}+\sum_{\ivartwo<\Dfscomb{\ivartwo}{\natit{0}}{\natit{1}}{\itermeight}}\itermtwo_\termone.
    \end{equation}
    Now, consider an index term $\itermseven$ such that
    $$
    \icontextone;\iconstraintone\models\Dfscomb{\ivartwo}{\natit{0}}{\natit{1}}{\itermseven}=\natit{1}+\sum_{\ivarone<\itermthree_1}\Dfscomb{\ivartwo}{\natit{0}}{\natit{1}}{\itermeight}
    $$
    Such an index term can be easily defined from $\itermeight$ and $\itermthree_1$, given that the underlying equational program
    is assumed to be universal. For the same reasons, one can define types $\typefive$ and $\typesix$, a type context $\tcontexthree$ and an index
    term $\itermnine$ such that the following holds (where $\sbstone$ is $\sbst{\ivartwo}{\natit{1}+\sum_{\ivarone<\ivarone}\Dfscomb{\ivartwo}{\natit{0}}{\natit{1}}{\itermeight}+\ivartwo}$):
    
    {\footnotesize
    \begin{align*}
      \icontextone;\iconstraintone&\pvdash\sb{\typesix}{\ivartwo}{\natit{0}}=\typethree;
        &
      \icontextone,\ivarone,\ivartwo;\iconstraintone,\ivarone<\itermthree_1,\ivartwo<\Dfscomb{\ivartwo}{\natit{0}}{\natit{1}}{\itermeight}&\pvdash 
        \sbt{\typesix}{\sbstone}=\typefour;\\
      \icontextone;\iconstraintone&\pvdash\sb{\typefive}{\ivartwo}{\natit{0}}=\typetwo_1;
        &
      \icontextone,\ivarone,\ivartwo,\ivarthree;\iconstraintone,\ivarone<\itermthree_1, \ivartwo<\Dfscomb{\ivartwo}{\natit{0}}{\natit{1}}{\itermeight},\ivarthree<\itermeight,&\pvdash 
         \sbt{\typefive}{\sbstone}=\sb{\typefour}{\ivartwo}{\Dfscomb{\ivartwo}{\ivartwo+1}{\ivarthree}{\itermeight}+\ivartwo+\natit{1}};\\
      \icontextone;\iconstraintone&\pvdash\sb{\itermnine}{\ivartwo}{\natit{0}}=\itermone_\termone;
        &
      \icontextone,\ivarone,\ivartwo;\iconstraintone,\ivarone<\itermthree_1,\ivartwo<\Dfscomb{\ivartwo}{\natit{0}}{\natit{1}}{\itermeight}&\models
         \sbt{\itermnine}{\sbstone}=\itermtwo_\termone;\\
      \icontextone;\iconstraintone&\pvdash\sb{\tcontexthree}{\ivartwo}{\natit{0}}\teq\tcontextone;
        &
      \icontextone,\ivarone,\ivartwo;\iconstraintone,\ivarone<\itermthree_1,\ivartwo<\Dfscomb{\ivartwo}{\natit{0}}{\natit{1}}{\itermeight}&\pvdash
         \sbt{\tcontexthree}{\sbstone}\teq\tcontexttwo.   
    \end{align*}}

    \noindent
    This is possible since the type derivations for (\ref{equ:fixtfirst}) and (\ref{equ:fixtsecond}) have exactly the same \PCF\ skeleton.
    By transforming them according to the equations above, one can merge them into one with conclusion:
    $$
    \icontextone,\ivartwo;\iconstraintone,\ivartwo<\Dfscomb{\ivartwo}{\natit{0}}{\natit{1}}{\itermseven};\varone:\qbang{\ivarone<\itermseven}{\typefive},\tcontexthree\pvdash_{\itermnine}\termone: \typesix.
    $$
    So, by using again the $\Recty$ rule we obtain:
    $$
    \icontextone;\iconstraintone;\sum_{\ivartwo<\Dfscomb{\ivartwo}{\natit{0}}{\natit{1}}{\itermseven}}\tcontexthree\pvdash_{\Dfscomb{\ivartwo}{\natit{0}}{\natit{1}}{\itermseven}
      \mnu\natit{1}+\sum_{\ivartwo<\Dfscomb{\ivartwo}{\natit{0}}{\natit{1}}{\itermseven}}\itermnine}\REC{\varone}{\termone}: \typethree.
    $$
    We are not at (\ref{equ:fixconcl}), however: it is still necessary to type $\envone$ appropriately.
    But note that we have:
    $$
    \icontextone,\iconstraintone\pvdash
    \sum_{\ivartwo< \Dfscomb{\ivartwo}{\natit{0}}{\natit{1}}{\itermseven}}\tcontexthree=\tcontextone\uplus \sum_{\ivartwo< \Dfscomb{\ivartwo}{\natit{0}}{\natit{1}}{\itermseven}\mnu 1}\tcontexttwo=
    \tcontextone\uplus\sum_{\ivarone<\itermthree_1}\sum_{\ivartwo< \Dfscomb{\ivartwo}{\natit{0}}{\natit{1}}{\itermeight}}\tcontexttwo=
    \tcontextone\uplus\sum_{\ivarone<\itermthree_1}\tcontextone.
    $$
    So we can find types $\typeseven_2,\ldots,\typeseven_n$ such that
    $$
    \sum_{\ivartwo< \Dfscomb{\ivartwo}{\natit{0}}{\natit{1}}{\itermseven}}\tcontexthree = 
      \varone_2:\qbang{\ivarone<\itermthree_2+\sum_{\ivarone<\itermthree_1}\itermfour_2}{\typeseven_2},\ldots, 
      \varone_n:\qbang{\ivarone<\itermthree_n+\sum_{\ivarone<\itermthree_1}\itermfour_n}{\typeseven_n},
    $$
    where for every $2\leq i\leq n$,
    \begin{align*}
      \icontextone,\ivarone;\iconstraintone,\ivarone<\itermthree_i&\pvdash\typeseven_i\teq\typetwo_i;\\
      \icontextone,\ivarone;\iconstraintone,\ivarone<\itermthree_1,\ivartwo<\itermfour_i&\pvdash\sb{\typeseven_i}{\ivarone}{\itermthree_i+\ivartwo+\sum_{\ivarone<\ivarone}\itermfour_i}\teq\typethree_i.
    \end{align*}
    Similarly, one can define index terms $\itermten_2,\ldots,\itermten_n$ such that
    \begin{align*}
      \icontextone,\ivarone;\iconstraintone,\ivarone<\itermthree_i&\models\itermten_i=\itermone_{\pairone_i};\\
      \icontextone,\ivarone;\iconstraintone,\ivarone<\itermthree_1,\ivartwo<\itermfour_i&\models\sb{\itermten_i}{\ivarone}{\itermthree_i+\ivartwo+\sum_{\ivarone<\ivarone}\itermfour_i}=\itermtwo_{\pairone_i}.
    \end{align*}
    By relabelling the type derivations of (\ref{equ:fixcifirst}) and (\ref{equ:fixcisecond}) (which are structurally equal) according to the types and index terms introduced above,
    one obtains:
    $$
      \icontextone,\ivarone;\iconstraintone,\ivarone<\itermthree_1+\sum_{\ivarone<\itermthree_1}\itermfour_i\pvdash_{\itermten_i}\pairone_i:\typeseven_i;
    $$
    From this it follows that $\icontextone;\iconstraintone\pvdash_{\itermone_{(\REC{\varone}{\termone},\envone)}}(\REC{\varone}{\termone},\envone):\typethree$,
    where
    \begin{align*}
      \itermone_{(\REC{\varone}{\termone},\envone)}&\equiv\Big ( \Dfscomb{\ivartwo}{\natit{0}}{\natit{1}}{\itermseven}\mnu\natit{1}+\sum_{\ivartwo<\Dfscomb{\ivartwo}{\natit{0}}{\natit{1}}{\itermseven}}\itermnine \Big)+
      \Big (\itermthree_2+\sum_{\ivarone<\itermthree_1}\itermfour_2+ \cdots+ \itermthree_n+\sum_{\ivarone<\itermthree_1}\itermfour_n+\\
      &\qquad\sum_{\ivarone<(\itermthree_2+\sum_{\ivarone<\itermthree_1}\itermfour_2)}\itermten_2 + \cdots+ \sum_{\ivarone<(\itermthree_n+\sum_{\ivarone<\itermthree_1}\itermfour_n)} \itermten_n\Big).
    \end{align*}
    Let us separately analyze the two thunks in which the expression above can be decomposed. On the one hand we have that:
    \begin{align*}
       \icontextone;\iconstraintone\models&\Dfscomb{\ivartwo}{\natit{0}}{\natit{1}}{\itermseven}\mnu\natit{1}+\sum_{\ivartwo<\Dfscomb{\ivartwo}{\natit{0}}{\natit{1}}{\itermseven}}\itermnine=
          \sum_{\ivarone<\itermthree_1}\Dfscomb{\ivartwo}{\natit{0}}{\natit{1}}{\itermeight}+\itermone_\termone+\sum_{\ivarone<\itermthree_1}\sum_{\ivartwo<\Dfscomb{\ivartwo}{\natit{0}}{\natit{1}}{\itermeight}}\itermtwo_{\termone}\\
          &\qquad=\sum_{\ivarone<\itermthree_1}\itermone_{\REC{\varone}{\termone}}+\itermthree_1+\itermone_\termone.
    \end{align*}
    On the other hand, let us observe that
    \begin{align*}
        \icontextone;\iconstraintone\models&\sum_{\ivarone<(\itermthree_2+\sum_{\ivarone<\itermthree_1}\itermfour_2)}\itermten_2 + \cdots+ \sum_{\ivarone<(\itermthree_n+\sum_{\ivarone<\itermthree_1}\itermfour_n)} \itermten_n\\
          &=\sum_{\ivarone<\itermthree_2}\itermone_{\pairone_2}+\sum_{\ivarone<\itermthree_2}\sum_{\ivartwo<\itermfour_2}\itermtwo_{\pairone_2}+\ldots+
            \sum_{\ivarone<\itermthree_n}\itermone_{\pairone_n}+\sum_{\ivarone<\itermthree_n}\sum_{\ivartwo<\itermfour_n}\itermtwo_{\pairone_n}.
    \end{align*}
    Combining the equations above with (\ref{equ:fixweightI}), (\ref{equ:fixweightII}) and (\ref{equ:fixweightIII}), 
    one easily reaches $\icontextone;\iconstraintone\models\itermone_{(\REC{\varone}{\termone},\envone)}=\itermone\mnu\itermone_{\stone}$,
    which is the thesis.
  \item 
    Consider the case 
    $$
    \confone\equiv(\CASE{\termfour}{\termtwo}{\termthree},\envone,\stone)
    \to (\termfour, \envone,(\termtwo,\termthree,\envone)\cdot\stone)\equiv\conftwo.
    $$
    By assumption we have that $\confone$ is typable in \PCF\ and that 
    $\icontextone;\iconstraintone\pvdash_\itermone\conftwo:\typeone$. So, we have that 
    \begin{align}
      \icontextone;\iconstraintone;\varone_1:\qbang{\ivarone<\itermthree_1}{\typetwo_1},\ldots,\varone_n:\qbang{\ivarone<\itermthree_n}{\typetwo_n}
      &\pvdash_{\itermone_\termfour}\termfour:\Nat[\itermfour];\label{equ:ifwvar}\\
      \icontextone,\ivarone;\iconstraintone,\ivarone<\itermthree_i&\pvdash_{\itermone_{\pairone_i}}\pairone_i:\typetwo_i;\label{equ:ifcloI}\\
      \icontextone;\iconstraintone,\itermfour\leq\natit{0}&\pvdash_{\itermone_{(\termtwo,\termthree,\envone)}}(\termtwo,\envone):\typethree;\label{equ:ifbranchI}\\
      \icontextone;\iconstraintone,\natit{1}\leq\itermfour&\pvdash_{\itermone_{(\termtwo,\termthree,\envone)}}(\termthree,\envone):\typethree;\label{equ:ifbranchII}\\
      \icontextone;\iconstraintone&\pvdash_{\itermone_\stone}\stone:(\typethree,\typeone);
    \end{align}
    where $\envone\equiv\pairone_1\ldots\pairone_n$. Moreover:
    \begin{equation}
      \icontextone;\iconstraintone\models\itermone=\itermone_\termfour+\itermthree_1+\ldots+\itermthree_n+\sum_{\ivarone<\itermthree_1}\itermone_{\pairone_1}+
        \ldots+\sum_{\ivarone<\itermthree_n}\itermone_{\pairone_n}+\itermone_{(\termtwo,\termthree,\envone)}+\itermone_\stone.
    \end{equation}
    By further spelling out (\ref{equ:ifbranchI}) and (\ref{equ:ifbranchII}), we obtain the following:
    \begin{align}
      \icontextone;\iconstraintone,\itermfour\leq\natit{0};\varone_1:\qbang{\ivarone<\itermfour_1}{\typefour_1},\ldots,\varone_n:\qbang{\ivarone<\itermfour_n}{\typefour_n}
      &\pvdash_{\itermone_\termtwo}\termtwo:\typethree;\label{equ:ifuvar}\\
      \icontextone,\ivarone;\iconstraintone,\itermfour\leq\natit{0},\ivarone<\itermfour_i&\pvdash_{\itermtwo_{\pairone_i}}\pairone_i:\typefour_i;\label{equ:ifcloII}\\
      \icontextone;\iconstraintone,\natit{1}\leq\itermfour;\varone_1:\qbang{\ivarone<\itermfive_1}{\typefive_1},\ldots,\varone_n:\qbang{\ivarone<\itermfive_n}{\typefive_n}
      &\pvdash_{\itermone_\termthree}\termthree:\typethree;\label{equ:ifvvar}\\
      \icontextone,\ivarone;\iconstraintone,\natit{1}\leq\itermfour,\ivarone<\itermfive_i&\pvdash_{\itermsix_{\pairone_i}}\pairone_i:\typefive_i;\label{equ:ifcloIII}
    \end{align}
    where
    \begin{align*}
      \icontextone;\iconstraintone,\itermfour\leq\natit{0}&\models\itermone_{(\termtwo,\termthree,\envone)}=\itermone_\termtwo+\itermfour_1+\ldots+\itermfour_n+
         \sum_{\ivarone<\itermfour_1}\itermtwo_{\pairone_1}+\ldots+\sum_{\ivarone<\itermfour_n}\itermtwo_{\pairone_n};\\
      \icontextone;\iconstraintone,\natit{1}\leq\itermfour&\models\itermone_{(\termtwo,\termthree,\envone)}=\itermone_\termthree+\itermfive_1+\ldots+\itermfive_n+
         \sum_{\ivarone<\itermfive_1}\itermsix_{\pairone_1}+\ldots+\sum_{\ivarone<\itermfive_n}\itermsix_{\pairone_n}.
    \end{align*}
    Please notice how the type derivations for (\ref{equ:ifcloI}), (\ref{equ:ifcloII}) and (\ref{equ:ifcloIII}) are structurally identical, i.e., their \PCF\ counterparts are the same.
    Now, let us build index terms $\itermseven_1,\ldots,\itermseven_n$, $\itermeight_{\pairone_1},\ldots,\itermeight_{\pairone_n}$, $\itermone_{\termtwo\termthree}$ and types 
    $\typesix_1,\ldots,\typesix_n$ such that:
    \begin{align*}
      \icontextone;\iconstraintone,\itermfour\leq\natit{0}&\models\itermseven_i=\itermfour_i;\\
      \icontextone;\iconstraintone,\natit{1}\leq\itermfour&\models\itermseven_i=\itermfive_i;\\
      \icontextone;\iconstraintone,\itermfour\leq\natit{0}&\models\itermone_{\termtwo\termthree}=\itermone_\termtwo;\\
      \icontextone;\iconstraintone,\natit{1}\leq\itermfour&\models\itermone_{\termtwo\termthree}=\itermone_\termthree;\\
      \icontextone;\iconstraintone,\ivarone<\itermthree_i&\models\itermeight_{\pairone_i}=\itermone_{\pairone_i};\\
      \icontextone;\iconstraintone,\itermfour\leq\natit{0},\ivarone<\itermfour_i&\models\sb{\itermeight_{\pairone_i}}{\ivarone}{\ivarone+\itermthree_i}=\itermtwo_{\pairone_i};\\
      \icontextone;\iconstraintone,\natit{1}\leq\itermfour,\ivarone<\itermfive_i&\models\sb{\itermeight_{\pairone_i}}{\ivarone}{\ivarone+\itermthree_i}=\itermsix_{\pairone_i};\\
      \icontextone;\iconstraintone,\ivarone<\itermthree_i&\pvdash\typesix_i\teq\typetwo_i;\\
      \icontextone;\iconstraintone,\itermfour\leq\natit{0},\ivarone<\itermfour_i&\pvdash\sb{\typesix_i}{\ivarone}{\ivarone+\itermthree_i}\teq\typefour_i;\\
      \icontextone;\iconstraintone,\natit{1}\leq\itermfour,\ivarone<\itermfive_i&\pvdash\sb{\typesix_i}{\ivarone}{\ivarone+\itermthree_i}\teq\typefive_i.
    \end{align*}
    As a consequence, one can rewrite (\ref{equ:ifwvar}), (\ref{equ:ifuvar}) and (\ref{equ:ifvvar}) as follows:

    {\footnotesize
    \begin{align*}
      \icontextone;\iconstraintone;\varone_1:\qbang{\ivarone<\itermthree_1}{\typesix_1},\ldots,\varone_n:\qbang{\ivarone<\itermthree_n}{\typesix_n}
      &\pvdash_{\itermone_\termfour}\termfour:\Nat[\itermfour];\\
      \icontextone;\iconstraintone,\itermfour\leq\natit{0};\varone_1:\qbang{\ivarone<\itermseven_1}{\sb{\typesix_1}{\ivarone}{\ivarone+\itermthree_1}},\ldots,\varone_n:
        \qbang{\ivarone<\itermseven_n}{\sb{\typesix_n}{\ivarone}{\ivarone+\itermthree_n}}&\pvdash_{\itermone_{\termtwo\termthree}}\termtwo:\typethree;\\
      \icontextone;\iconstraintone,\natit{1}\leq\itermfour;\varone_1:\qbang{\ivarone<\itermseven_1}{\sb{\typesix_1}{\ivarone}{\ivarone+\itermthree_1}},\ldots,\varone_n:
        \qbang{\ivarone<\itermseven_n}{\sb{\typesix_n}{\ivarone}{\ivarone+\itermthree_n}}&\pvdash_{\itermone_{\termtwo\termthree}}\termthree:\typethree;
    \end{align*}}

    \noindent from which one obtains
    
    {\footnotesize
    $$
    \icontextone;\iconstraintone;\varone_1:\qbang{\ivarone<\itermthree_1+\itermseven_1}{\typesix_1},\ldots,\varone_n:\qbang{\ivarone<\itermthree_n+\itermseven_n}{\typesix_n}
      \pvdash_{\itermone_\termfour+\itermone_{\termtwo\termthree}}\CASE{\termfour}{\termtwo}{\termthree}:\typethree.
    $$}
    
    \noindent Similarly, one obtains that
    $$
    \icontextone,\ivarone;\iconstraintone,\ivarone<\itermthree_i+\itermseven_i\pvdash_{\itermeight_{\pairone_i}}\pairone_i:\typesix_i;
    $$
    and, as a consequence, that $\icontextone;\iconstraintone\pvdash_{\itermone_{\confone}}\confone:\typeone$, where
    $$
    \itermone_{\confone}\equiv\itermone_\termfour+\itermone_{\termtwo\termthree}+\itermthree_1+\itermseven_1+\ldots+\itermthree_n+\itermseven_n+
      \sum_{\ivarone<\itermthree_1+\itermseven_1}\itermeight_{\pairone_1}+\ldots+\sum_{\ivarone<\itermthree_n+\itermseven_n}\itermeight_{\pairone_n}.
    $$
    But observe that
    \begin{align*}
    \icontextone;\iconstraintone,\itermfour\leq\natit{0}\models&\;\itermone_{\confone}=\itermone_\termfour+\itermone_\termtwo+\itermthree_1+\ldots+\itermthree_n+
        \sum_{\ivarone<\itermthree_1}\itermeight_{\pairone_1}+\ldots+\sum_{\ivarone<\itermthree_n}\itermeight_{\pairone_n}\\
         &\qquad\qquad
         +\itermseven_1+\ldots+\itermseven_n+\sum_{\ivarone<\itermseven_1}\sb{\itermeight}{\ivarone}{\ivarone+\itermthree_1}+\cdots
         +
           \sum_{\ivarone<\itermseven_n}\sb{\itermeight}{\ivarone}{\ivarone+\itermthree_n}\\
         &\hspace{16pt}=\itermone_\termfour+\itermone_\termtwo+\itermthree_1+\ldots+\itermthree_n+
        \sum_{\ivarone<\itermthree_1}\itermone_{\pairone_1}+\ldots+\sum_{\ivarone<\itermthree_n}\itermone_{\pairone_n}\\
         &\qquad\qquad
         +\itermfour_1+\ldots+\itermfour_n+\sum_{\ivarone<\itermfour_1}\itermtwo_{\pairone_1}+\cdots +
           \sum_{\ivarone<\itermfour_n}\itermtwo_{\pairone_n}\\
         &\hspace{16pt}=\itermone_\termfour+\itermthree_1+\ldots+\itermthree_n+
        \sum_{\ivarone<\itermthree_1}\itermone_{\pairone_1}+\ldots+\sum_{\ivarone<\itermthree_n}\itermone_{\pairone_n}+\itermone_{(\termtwo,\termthree,\envone)}=\itermone.
    \end{align*}
    Similarly, one can prove that $\icontextone;\iconstraintone,\natit{1}\leq\itermfour\models\itermone_{\confone}=\itermone$.
    Summing up, we get $\icontextone;\iconstraintone\models\itermone_{\confone}=\itermone$, which is the thesis.
  \item 
    Consider the case 
    $$
    \confone\equiv(\varone_m,((\termone_0,\envone_0),\ldots,(\termone_n,\envone_n)),\stone)
    \to (\termone_m, \envone_m,\stone)\equiv\conftwo.
    $$  
    By assumption we have that $\confone$ is typable in \PCF\ and that 
    $\icontextone;\iconstraintone\pvdash_\itermone\conftwo:\typeone$. So, we have that 
    \begin{align}
      \icontextone;\iconstraintone&\pvdash_{\itermone_{(\termone_m,\envone_m)}}(\termone_m,\envone_m):\typetwo;\\
      \icontextone;\iconstraintone&\pvdash_{\itermone_\stone}\stone:(\typetwo,\typeone);
    \end{align}
    where $\icontextone;\iconstraintone\models\itermone={\itermone_{(\termone_m,\envone_m)}}+\itermone_\stone$.
    Any closure $(\termone_i,\envone_i)$ (where $1\leq i\leq n$ but $i\neq m$) can be typed as
    follows:
    $$
    \icontextone;\iconstraintone,\ivarone<\natit{0}\pvdash_{\natit{0}}(\termone_i,\envone_i):\typethree_i
    $$
    for some type $\typethree_i$. This is because all these closures are by hypothesis typable in \PCF\ and,
    moreover, $\iconstraintone,\ivarone<\natit{0}$ is inconsistent. For obvious reasons, 
    $$
    \icontextone;\iconstraintone,\ivarone<\natit{1}\pvdash_{\itermone_{(\termone_m,\envone_m)}}(\termone_m,\envone_m):\typetwo.
    $$
    Finally, we can build the following type derivation
    $$
    \infer
    {
      \icontextone;\iconstraintone;\varone_1:\qbang{\ivarone_1<\natit{0}}
                  {\typethree_1},\ldots,\varone_m:\qbang{\ivarone<\natit{1}}{\typetwo} ,\ldots, \varone_n:\qbang{\ivarone_n<\natit{0}}{\typethree_n}
         \pvdash_{\zero}\varone_m:\typetwo 
    }
    {
      \icontextone;\iconstraintone\pvdash\sb{\typetwo}{\ivartwo}{\natit{0}}\tless\typetwo
    }
    $$
    But all this implies that $\icontextone;\iconstraintone\pvdash_{\itermone_\confone}\confone:\typeone$
    where $\icontextone;\iconstraintone\models\itermone_\confone=\itermone+\natit{1}$, which implies the thesis.
\end{varitemize}
This concludes the proof.
\end{proof}

Relative completeness for programs is a direct consequence of Weighted Subject Expansion:

\begin{thm}[Relative Completeness for Programs]
\label{thm:program-relative-completeness}
Let $\termone$ be a \PCF\ program such that $\termone\ev^{n}\val{m}$. Then, there exist two index terms $\itermone$ and $\itermtwo$ 
such that $\semt{\itermone}{}{\eqpun}\leq n$ and $\semt{\itermtwo}{}{\eqpun}=m$ and such that the term 
$\termone$ is typable in \PCFld\ as $\vdash_{\itermone}^{\eqpun}\termone:\Nat[\itermtwo]$.
\end{thm}
\begin{proof}
  By induction on $n$ using Weighted Subject Expansion
  and  Lemma \ref{lem:program-typing}.
\end{proof}

\subsection{Uniformization and Relative Completeness for Functions}
\label{sec:uniformization}
It is useful to recall that by \emph{relative completeness for functions} we mean the following:
for each \PCF\ term $\termone$ computing a total function $f$ 
in time expressed by a function $g$ there exists a type derivation  
in \PCFld\
whose index terms capture both the extensional functional 
behavior $f$ and the intensional property $g$. 
Anticipating on what follows, and using an intuitive notation, this can be expressed 
by a typing judgement like
$$
\ivarone;\emcon;\varone:\Nat[\ivarone] \vdash_{\texttt{g}(\ivarone)} \termone:\Nat [\texttt{f}(\ivarone)].
$$
In order to show this form of relative completeness, a \emph{uniformization} result for
type derivations needs to be proved. 

Suppose that $\{\tdone\}_{n\in\mathbb{N}}$
is a sufficiently ``regular'' (i.e. recursively enumerable) family of type derivations 
such that any $\tdone_n$ is mapped
by $\TtoNDT{\cdot}$ to the \emph{same} \PCF\ type derivation. Uniformization
tells us that with the hypothesis above, there is a \emph{single} type derivation 
$\tdone$ which captures the whole family $\{\tdone_n\}_{n\in\mathbb{N}}$.
In other words, uniformization is an extreme form of polymorphism.
Note that, for instance, uniformization does not hold in intersection types, 
where \emph{uniform typing} permits only to define small classes of functions 
\cite{Leivant:acm:lfp:1990,LICS99*109,
journals/mscs/BucciarelliPS03}.

More formally, a family $\{\tdone_\natone\}_{\natone\in\NN}$ of type derivations is
said to be \emph{recursively enumerable} if there is a computable function $\funone$
which, on input $\natone$, returns (an encoding of) $\tdone_\natone$. Similarly,
recursively enumerable families of index terms, types and modal types can be defined.

It is easy to turn ``uniform families'' of semantic entailments into one compact form:
\begin{lem}\label{lemma:unifsem}
\begin{varenumerate}
\item
  If for every $\natone\in\NN$ it holds that
  $\icontextone;\sb{\iconstraintone}{\ivarone}{\natit{\natone}}\models^\eqpone\sb{\itermone}{\ivarone}{\natit{\natone}}\kleq\sb{\itermtwo}{\ivarone}{\natit{\natone}}$,
  then $\icontextone,\ivarone;\iconstraintone\models^\eqpone\itermone\kleq\itermtwo$.
\item
  If for every $\natone\in\NN$ it holds that
  $\icontextone;\sb{\iconstraintone}{\ivarone}{\natit{\natone}}\models^\eqpone\sb{\itermone}{\ivarone}{\natit{\natone}}\leq\sb{\itermtwo}{\ivarone}{\natit{\natone}}$,
  then $\icontextone,\ivarone;\iconstraintone\models^\eqpone\itermone\leq\itermtwo$.
\end{varenumerate}
\end{lem}
\begin{proof}
This is just an trivial consequence of the way semantic entailment is defined.
Suppose, for example, that for every $\natone\in\NN$ the following holds
$\icontextone;\sb{\iconstraintone}{\ivarone}{\natit{\natone}}\models^\eqpone
\sb{\itermone}{\ivarone}{\natit{\natone}}\kleq\sb{\itermtwo}{\ivarone}{\natit{\natone}}$.
Now, what should we do to prove $\icontextone,\ivarone;\iconstraintone\models^\eqpone\itermone\kleq\itermtwo$?
We should prove that for every value of the variables in
$\icontextone,\ivarone$ satisfying $\iconstraintone$, $\itermone$ and $\itermtwo$
are equal in the sense of Kleene. But this is just what the hypothesis ensures.
\end{proof}

Before embarking on the proof of uniformization for type derivations, it makes sense
to prove the same result for index terms and types, respectively.
\begin{lem}[Uniformizing Index Terms]\label{lemma:unindtms}
Suppose that:
\begin{varenumerate}
\item
  $\{\itermone_\natone\}_{\natone\in\NN}$ is recursively enumerable,
  where for every $\natone\in\NN$, $\itermone_\natone$ is an index term on a 
  signature $\sigun$;
\item
  There is a finite set of variables $\icontextone=\ivarone_1,\ldots,\ivarone_\nattwo$ 
  such that any variables appearing in any $\itermone_\natone$ is in $\icontextone$
\end{varenumerate}
Then there is a term $\itermone$ on the signature $\sigun$ such that 
$\icontextone;\emcon\vdash^\eqpun\sb{\itermone}{\ivarone}{\natit{\natone}}\kleq\itermone_\natone$
for every $\natone$.
\end{lem}
\begin{proof}
Consider the function $\funone:\NN^{\nattwo+1}\rightarrow\NN$ defined as follows:
$$
(x_0,x_1,\ldots,x_\nattwo)\mapsto\semt{\itermone_{x_0}}{[\ivarone_1\leftarrow x_1,\ldots,\ivarone_\natone\leftarrow x_\nattwo]}{\eqpun}.
$$
An algorithm computing $\funone$ can be defined as follows:
\begin{varitemize}
\item
  From $x_0$, compute $\itermone_{x_0}$. Again, this can be done
  effectively.
\item
  Evaluate $\itermone_{x_0}$ where the variables $\ivarone_1,\ldots,\ivarone_\natone$
  takes values $x_1,\ldots,x_\natone$, respectively.
\end{varitemize}
In other words, $\funone$ is computable. Thus, the existence
of a term $\itermone$ like the one required is a consequence of
the universality of the equational program $\eqpun$.
\end{proof}
Observe how the index terms in $\{\itermone_\natone\}_{\natone\in\NN}$
need not be defined for all values of the variables occurring in them.
More: their domains of definition can all be different. The way $\itermone$
is defined, however, ensures that $\semu{\sb{\itermone}{\ivarone}{\natit{n}}}$ is
defined iff $\semu{\itermone_n}$ is defined. Uniformizing types requires
a little more care:
\begin{lem}[Uniformizing Types and Modal Types]\label{lemma:unityp}
Suppose that $\{\tdone_\natone\}_{\natone\in\NN}$ is recursively enumerable and that:
\begin{varenumerate} 
  \item
    for every $\natone\in\NN$, $\tdone_\natone\prov\icontextone;\iconstraintone_\natone\vdash^{\eqpun}\conv{\typeone_\natone}$;
  \item
    for every $\natone,\nattwo\in\NN$, $\TtoNDT{\typeone_\natone}=\TtoNDT{\typeone_\nattwo}$;
  \item
    every $\iconstraintone_\natone$ have the form
    $\itermone^\natone_1\leq\itermtwo^\natone_1,\ldots,\itermone^\natone_\nattwo\leq\itermtwo^\natone_\nattwo$,
    where $\nattwo$ does not depend on $\natone$.
\end{varenumerate}
Then there is one type $\typeone$ such that:
\begin{varenumerate}
\item
  $\icontextone,\ivarone;\iconstraintone\vdash^{\eqpun}\conv{\typeone}$;
\item  
  $\iconstraintone=\itermone_1\leq\itermtwo_1,\ldots,\itermone_\nattwo\leq\itermtwo_\nattwo$;
\item
  for every $1\leq\natthree\leq\nattwo$, both
  $\icontextone;\emcon\vdash^\eqpun\sb{\itermone_\natthree}{\ivarone}{\natit{\natone}}\kleq\itermone_\natthree^\natone$ and
  $\icontextone;\emcon\vdash^\eqpun\sb{\itermtwo_\natthree}{\ivarone}{\natit{\natone}}\kleq\itermtwo_\natthree^\natone$;
\item
  for every $\natone\in\NN$, it holds that $\icontextone;\sb{\iconstraintone}{\ivarone}{\natit{\natone}}
  \vdash^{\eqpun}\sb{\typeone}{\ivarone}{\natit{\natone}}\cong\typeone_\natone$.
\end{varenumerate}
Moreover, the same statement holds for modal types.
\end{lem}
\begin{proof}
The proof goes by induction on the structure of the type $\TtoNDT{\typeone_0}$
and of the modal type $\TtoNDT{\mtypeone_0}$. An essential ingredient in the proof is,
of course, Lemma~\ref{lemma:unindtms}. Suppose, as an example, that $\TtoNDT{\typeone_0}\equiv\Nat$.
This implies that there are index terms $\itermthree_\natone,\itermfour_\natone$ such that,
for every $\natone\in\NN$,
$$
\typeone_n\equiv\Nat[\itermthree_\natone,\itermfour_\natone].
$$
Now, let $\itermone_1,\itermtwo_1,\ldots,\itermone_\nattwo,\itermtwo_\nattwo,\itermthree,\itermfour$
be the index terms obtained from the families 
$$
\{\itermone_1^\natone\}_{\natone\in\NN},\{\itermtwo_1^\natone\}_{\natone\in\NN},\ldots,\{\itermone_\nattwo^\natone\}_{\natone\in\NN},
\{\itermtwo_\nattwo^\natone\}_{\natone\in\NN},\{\itermthree_\natone\}_{\natone\in\NN},\{\itermfour_\natone\}_{\natone\in\NN}
$$
through Lemma~\ref{lemma:unindtms}. Let $\iconstraintone$ be just
$\itermone_1\leq\itermtwo_1,\ldots,\itermone_\nattwo\leq\itermtwo_\nattwo$ and let $\typeone$ be
$\Nat[\itermthree,\itermfour]$.
From $\tdone_\natone\prov\icontextone;\iconstraintone_\natone\vdash^{\eqpun}\conv{\typeone_\natone}$,
it follows that
\begin{align}
  \icontextone;\sb{\iconstraintone}{\ivarone}{\natit{n}}&\models^\eqpun\conv{\sb{\itermthree}{\ivarone}{\natit{n}}}\label{equ:unindI};\\
  \icontextone;\sb{\iconstraintone}{\ivarone}{\natit{n}}&\models^\eqpun\conv{\sb{\itermfour}{\ivarone}{\natit{n}}}\label{equ:unindII}.
\end{align}
By Lemma~\ref{lemma:unifsem}, it follows that
\begin{align*}
  \icontextone,\ivarone;\iconstraintone&\models^\eqpun\conv{\itermthree};\\
  \icontextone,\ivarone;\iconstraintone&\models^\eqpun\conv{\itermfour};
\end{align*}
which implies $\icontextone,\ivarone;\iconstraintone\vdash^{\eqpun}\conv{\typeone}$.
From (\ref{equ:unindI}) and $\icontextone;\emcon\vdash^\eqpun\sb{\itermthree}{\ivarone}{\natit{n}}\kleq\itermthree_\natone$,
it follows that 
$$
\icontextone;\sb{\iconstraintone}{\ivarone}{\natit{n}}\models^\eqpun\sb{\itermthree}{\ivarone}{\natit{n}}=\itermthree_\natone.
$$
Similarly, from ~\ref{equ:unindII} one obtains 
$$
\icontextone;\sb{\iconstraintone}{\ivarone}{\natit{n}}\models^\eqpun\sb{\itermfour}{\ivarone}{\natit{n}}=\itermfour_\natone.
$$
As a consequence, $\icontextone;\sb{\iconstraintone}{\ivarone}{\natit{\natone}}
\vdash^{\eqpun}\sb{\typeone}{\ivarone}{\natit{\natone}}\cong\typeone_\natone$.
\end{proof}
Now that we are able to unify a denumerable family of types into one, we have all the necessary tools
to turn a family of judgements into one. For \emph{subtyping} judgments, the task is relatively simple,
because types and index terms occurring inside any subtyping derivation also occur in its conclusion:
\begin{lem}[Uniformizing Subtyping Judgments]\label{lemma:unifrel}
  If for every $\natone\in\NN$ it holds that
  $\icontextone;\sb{\iconstraintone}{\ivarone}{\natit{\natone}}
  \vdash^\eqpone\sb{\typeone}{\ivarone}{\natit{\natone}}\tless
  \sb{\typetwo}{\ivarone}{\natit{\natone}}$,
  then $\icontextone,\ivarone;\iconstraintone\vdash^\eqpone\typeone\tless\typetwo$.
\end{lem}
\begin{proof}
This is an induction on the structure of a proof of $\typeone$.
If, as an example, $\typeone\equiv\Nat[\itermone,\itermtwo]$, then
$\typetwo\equiv\Nat[\itermthree,\itermfour]$. From the hypothesis,
we know that
\begin{align*}
  \icontextone;\sb{\iconstraintone}{\ivarone}{\natit{\natone}}
  &\vdash^\eqpone\sb{\itermthree}{\ivarone}{\natit{\natone}}\leq
  \sb{\itermone}{\ivarone}{\natit{\natone}};\\
  \icontextone;\sb{\iconstraintone}{\ivarone}{\natit{\natone}}
  &\vdash^\eqpone\sb{\itermtwo}{\ivarone}{\natit{\natone}}\leq
  \sb{\itermfour}{\ivarone}{\natit{\natone}}.
\end{align*} 
By Lemma~\ref{lemma:unifsem}, we can conclude that
\begin{align*}
  \icontextone;\iconstraintone&\vdash^\eqpone\itermthree\leq\itermone;\\
  \icontextone;\iconstraintone&\vdash^\eqpone\itermtwo\leq\itermfour;
\end{align*} 
which immediately yields the thesis.
\end{proof}
In \emph{typing} judgments, on the other hand, there can be types and index terms
which occur in the derivation, but not in its conclusion --- think about how applications
are typed. We then need to impose some further constraints on the kind of (type derivation)
families which we can unify:
\begin{lem}[Uniformizing Typing Judgments]
\label{lem:uniformizing-typing}
If for every $\natone\in\NN$ it holds that
$\tdone_\natone\prov\icontextone;\sb{\iconstraintone}{\ivarone}{\natit{\natone}};\sb{\tcontextone}{\ivarone}{\natit{\natone}}
\vdash^{\eqpun}_{\sb{\itermone}{\ivarone}{\natit{\natone}}}\termone:\sb{\typeone}{\ivarone}{\natit{\natone}}$, where
$\{\tdone_\natone\}_{\natone\in\NN}$ is recursively enumerable and such that $\TtoNDT{\tdone_\natone}=\TtoNDT{\tdone_\nattwo}$ for every $\natone,\nattwo\in\NN$,
then $\icontextone,\ivarone;\iconstraintone;\tcontextone\vdash^{\eqpun}_{\itermone}\termone:\typeone$.
\end{lem}
\begin{proof}
  The proof goes by induction on the structure of $\termone$. Some interesting cases:
  \begin{varitemize}
  \item
    Suppose that $\termone$ is a variable $\varone$. Then $\tdone_\natone$ has the following shape:
    $$
    \infer[\Axty]
    {\icontextone;\sb{\iconstraintone}{\ivarone}{\natit{\natone}};\sb{\tcontexttwo}{\ivarone}{\natit{\natone}},\varone:\qbang{\ivartwo<\sb{\itermone}{\ivarone}{\natit{\natone}}}{\sb{\typeone}{\ivarone}{\natit{\natone}}}\vdash_{\sb{\itermtwo}{\ivarone}{\natit{\natone}}}^{\eqpun} \varone:\sb{\typetwo}{\ivarone}{\natit{\natone}}}
    {
      \begin{array}{c}
        \icontextone;\sb{\iconstraintone}{\ivarone}{\natit{\natone}}\models^{\eqpun} \natit{0}\leq \sb{\itermtwo}{\ivarone}{\natit{\natone}} \qquad
        \icontextone;\sb{\iconstraintone}{\ivarone}{\natit{\natone}}\models^{\eqpun} \natit{1} \leq \sb{\itermone}{\ivarone}{\natit{\natone}}\\
        \icontextone;\sb{\iconstraintone}{\ivarone}{\natit{\natone}}\vdash^{\eqpun}\sb{\sb{\typeone}{\ivarone}{\natit{\natone}}}{\ivartwo}{\natit{0}}\tless{\sb{\typetwo}{\ivarone}{\natit{\natone}}}\\
        \icontextone;\sb{\iconstraintone}{\ivarone}{\natit{\natone}}\vdash^{\eqpun}\conv{(\qbang{\ivarone<\sb{\itermone}{\ivarone}{\natit{\natone}}}{\typeone})}\qquad
        \icontextone;\sb{\iconstraintone}{\ivarone}{\natit{\natone}}\vdash^{\eqpun}\conv{\sb{\tcontexttwo}{\ivarone}{\natit{\natone}}}
      \end{array}
    }
    $$
    Notice that $\sb{\sb{\typeone}{\ivarone}{\natit{\natone}}}{\ivartwo}{\natit{0}}$ is literally the same
    as $\sb{\sb{\typeone}{\ivartwo}{\natit{0}}}{\ivarone}{\natit{\natone}}$. Lemma~\ref{lemma:unifsem} and Lemma~\ref{lem:uniformizing-typing}
    allow us to derive the following
    \begin{align*}
      \icontextone,\ivarone;\iconstraintone&\models^{\eqpun} \natit{0}\leq \itermtwo;\\
      \icontextone,\ivarone;\iconstraintone&\models^{\eqpun} \natit{1} \leq \itermone;\\
      \icontextone,\ivarone;\iconstraintone&\vdash^{\eqpun}\sb{\typeone}{\ivartwo}{\natit{0}}\tless\typetwo;\\
      \icontextone,\ivarone;\iconstraintone&\vdash^{\eqpun}\conv{(\qbang{\ivarone<\itermone}{\typeone})};\\
      \icontextone,\ivarone;\iconstraintone&\vdash^{\eqpun}\conv{\tcontexttwo};
    \end{align*}
    from which the thesis easily follows.
  \item
    Suppose that $\termone$ is $\termtwo\termthree$. Then the derivations in $\{\tdone_\natone\}_{\natone\in\NN}$ have
    the following shape: 
    $$
    \infer[\Apty]
    {\icontextone;\sb{\iconstraintone}{\ivarone}{\natit{\natone}};\sb{\tcontexthree}{\ivarone}{\natit{\natone}}\vdash^{\eqpun}_{\sb{\itermfour}{\ivarone}{\natit{\natone}}} 
       \termone\termtwo:\sb{\typetwo}{\ivarone}{\natit{\natone}}}
    {
      \begin{array}{c}
        \icontextone;\sb{\iconstraintone}{\ivarone}{\natit{\natone}};\tcontextone_n\vdash^{\eqpun}_{\itermtwo_n}\termone:\qlin{\ivartwo<\itermone_\natone}
          {\typeone_\natone}{\sb{\typetwo}{\ivarone}{\natit{n}}}\\
        \icontextone,\ivartwo;\sb{\iconstraintone}{\ivarone}{\natit{\natone}},\ivartwo<\itermone_\natone ;\tcontexttwo_\natone\vdash^{\eqpun}_{\itermthree_\natone} \termtwo:\typeone_\natone\\
        \icontextone;\sb{\iconstraintone}{\ivarone}{\natit{\natone}}\vdash^{\eqpun}\sb{\tcontexthree}{\ivarone}{\natit{\natone}}\tless \tcontextone_\natone\uplus
           \sum_{\ivartwo<\itermone_\natone}{\tcontexttwo_\natone}\\
        \icontextone;\sb{\iconstraintone}{\ivarone}{\natit{\natone}}\models^{\eqpun}\sb{\itermfour}{\ivarone}{\natit{\natone}}\geq\itermtwo_\natone+\itermone_\natone
           +\sum_{\ivartwo<\itermone_\natone}\itermthree_\natone
      \end{array}
    }
    $$
    By Lemma~\ref{lemma:unindtms} and Lemma~\ref{lemma:unityp}, there are index terms $\itermone,\itermtwo,\itermthree$ and a type $\typeone$, and
    typing contexts $\tcontextone$ and $\tcontexttwo$ such that the following holds:
    \begin{align*}
      \icontextone;\emcon&\models^{\eqpun}\sb{\itermone}{\ivarone}{\natit{\natone}}\kleq\itermone_\natone;\\
      \icontextone;\emcon&\models^{\eqpun}\sb{\itermtwo}{\ivarone}{\natit{\natone}}\kleq\itermtwo_\natone;\\
      \icontextone,\ivartwo;\ivartwo<\sb{\itermone}{\ivarone}{\natit{\natone}}&\models^{\eqpun}\sb{\itermthree}{\ivarone}{\natit{\natone}}\kleq\itermthree_\natone;\\
      \icontextone,\ivartwo;\iconstraintone,\ivartwo<\itermone&\vdash^{\eqpun}\conv{\typeone};\\
      \icontextone,\ivartwo;\sb{\iconstraintone}{\ivarone}{\natit{\natone}},\ivartwo<\sb{\itermone}{\ivarone}{\natit{\natone}}&
        \vdash^{\eqpun}\sb{\typeone}{\ivarone}{\natit{\natone}}\cong\typeone_\natone;\\
      \icontextone,\ivartwo;\iconstraintone,\ivartwo<\itermone&\vdash^{\eqpun}\conv{\tcontextone};\\
      \icontextone,\ivartwo;\sb{\iconstraintone}{\ivarone}{\natit{\natone}},\ivartwo<\sb{\itermone}{\ivarone}{\natit{\natone}}&
        \vdash^{\eqpun}\sb{\tcontextone}{\ivarone}{\natit{\natone}}\cong\tcontextone_\natone;\\
      \icontextone,\ivartwo;\iconstraintone,\ivartwo<\itermone&\vdash^{\eqpun}\conv{\tcontexttwo};\\
      \icontextone,\ivartwo;\sb{\iconstraintone}{\ivarone}{\natit{\natone}},\ivartwo<\sb{\itermone}{\ivarone}{\natit{\natone}}&
        \vdash^{\eqpun}\sb{\tcontexttwo}{\ivarone}{\natit{\natone}}\cong\tcontexttwo_\natone.
    \end{align*}
    From the above, we first of all obtain
    $$
    \icontextone;\sb{\iconstraintone}{\ivarone}{\natit{\natone}}\models^{\eqpun}\sb{\itermfour}{\ivarone}{\natit{\natone}}\geq\sb{\itermtwo}{\ivarone}{\natit{\natone}}+
    \sb{\itermone}{\ivarone}{\natit{\natone}}+\sum_{\ivartwo<\sb{\itermone}{\ivarone}{\natit{\natone}}}\sb{\itermthree}{\ivarone}{\natit{\natone}},
    $$
    that by Lemma~\ref{lemma:unifsem} becomes
    $$
    \icontextone,\ivarone;\iconstraintone\models^{\eqpun}\itermfour\geq\itermtwo+\itermone+\sum_{\ivartwo<\itermone}\itermthree.
    $$
    Analogously, this time through Lemma~\ref{lemma:unifrel}, one easily reach
    $$
    \icontextone,\ivarone;\iconstraintone\vdash^\eqpun\tcontexthree\tless\tcontextone\uplus\sum_{\ivartwo<\itermone}\tcontexttwo.
    $$
    Again, one can reach
    \begin{align*}
        \icontextone;\sb{\iconstraintone}{\ivarone}{\natit{\natone}};\sb{\tcontextone}{\ivarone}{\natit{\natone}}&\vdash^{\eqpun}_{\sb{\itermtwo}{\ivarone}{\natit{\natone}}}
           \termone:\qlin{\ivartwo<\sb{\itermone}{\ivarone}{\natit{\natone}}}
          {\sb{\typeone}{\ivarone}{\natit{\natone}}}{\sb{\typetwo}{\ivarone}{\natit{n}}};\\
        \icontextone,\ivartwo;\sb{\iconstraintone}{\ivarone}{\natit{\natone}},\ivartwo<\sb{\itermone}{\ivarone}{\natit{\natone}} ;\sb{\tcontexttwo}{\ivarone}{\natit{\natone}}&
           \vdash^{\eqpun}_{\sb{\itermthree}{\ivarone}{\natit{\natone}}} \termtwo:\sb{\typeone}{\ivarone}{\natit{\natone}};      
    \end{align*}
    to which one can apply the induction hypothesis. The thesis easily follows.
  \end{varitemize}
This concludes the proof.
\end{proof}

Uniformization is the key to prove relative completeness
for functions from relative completeness for programs:

\begin{thm}[Relative Completeness for Functions]
\label{thm:function-relative-completeness}
Suppose that $\termone$ is a \PCF\ term such that  $\vdash \termone:\Nat\arr\Nat$. Moreover,
suppose that there are two (total and computable) functions $\funone,\funtwo:\NN\rightarrow\NN$ such that
$\termone\;\val{n}\ev^{\funtwo(n)}\val{\funone(n)}$.
Then there are terms $\itermone,\itermtwo,\itermthree$ with 
$\semu{\itermone+\itermtwo}\leq g$  and $\semu{\itermthree}=f$, 
such that
$$
\ivarone;\emcon;\emcon\vdash_\itermone^\eqpun\termone :\qbang{\ivartwo<\itermtwo}
\Nat[\ivarone]\lin\Nat[\itermthree].
$$
\end{thm}
\begin{proof}
A consequence of relative completeness for 
programs (Theorem \ref{thm:program-relative-completeness})
and Lemma \ref{lem:uniformizing-typing}. Indeed,
a type derivation for 
$\ivarone;\emcon;\emcon\vdash_\itermone\termone :\qbang{\ivartwo<\itermtwo}\Nat[\ivarone]\lin\Nat[\itermthree]$
can be obtained simply by uniformizing all type derivations $\tdone_\natone$
for programs in the form $\termone\val{n}$. In turn, those
type derivations can be built effectively by way of Subject
Expansion.
\end{proof}

\section{On the Undecidability of Type Checking}
\label{sec:type-checking}
As we have seen in the last two sections, \PCFld\ is not only sound,
but complete: all true typing judgements involving programs
can be derived, and this can be indeed lifted to first-order functions,
as explained in Section~\ref{sec:uniformization}.

There is a price to pay, however. Checking a type derivation
for correctness is undecidable in general, simply because
it can rely on semantic assumptions in the form of inequalities
between index terms, or on subtyping judgements, which themselves
rely on the properties of the underlying equational program $\eqpone$.
If $\eqpone$ is sufficiently involved, e.g. if we work with $\eqpun$,
there is no hope to find a decidable complete type checking procedure.
In this sense, \PCFld\ is a non-standard type system.

Indeed, \PCFld\ is not actually a type system, but rather
a \emph{framework} in which various distinct type systems
can be defined. Concrete type systems can be developed along
two axes: on the one hand by concretely instantiating $\eqpone$, on the
other by choosing specific and sound formal systems for the verification
of semantic assumptions. This way sound and possibly decidable
type systems can be derived.
Even if completeness can only be achieved if $\eqpone$ is
universal, soundness holds for every equational program $\eqpone$.
Choosing a simple equational program $\eqpone$ results in
a (incomplete) type system for which the problem of 
checking the inequalities can be much easier, if not
decidable. And even if $\eqpone$ remains universal, assumptions
could be checked using techniques such as abstract interpretation or 
theorem proving.

By the way, the just described phenomenon is not peculiar to \PCFld. Unsurprisingly,
program logics have similar properties, since the rule
$$
\infer
  {\{p\}P\{q\}}
  {p\Rightarrow r & \{r\}P\{s\} & s\Rightarrow q}
$$
is part of most relatively complete Hoare-Floyd logics and, of course,
the premises $p\Rightarrow r$ and $s\Rightarrow q$ have to be taken
semantically for completeness to hold.

\section{\PCFld\ and Implicit Computational Complexity}

One of the original motivations for the studies which lead to the definition of \PCFld\ came
from Implicit Computational Complexity. There, one aims at giving characterizations of 
complexity classes which can often be turned into type systems or static analysis methodologies
for the verification of resource usage of programs. Historically \cite{Hofmann:2000:PLC:346048.346051,MarionHDR}, what prevented most
ICC techniques to find concrete applications along this line was their poor expressive power:
the class of programs which can be recognized as being efficient by (tools derived from) ICC
systems is often very small and does not include programs corresponding to
natural, well-known algorithms. This is true despite the fact that ICC systems are \emph{extensionally} 
complete --- they capture complexity classes seen as classes of \emph{functions}.
The kind of Intensional Completeness enjoyed by \PCFld\ is much stronger: all \PCF\ programs with a
certain complexity can be proved to be so by deriving a typing judgement for them.

Of course, \PCFld\ is not at all an implicit system: bounds appear everywhere!
On the other hand, \PCFld\ allows to analyze the time complexity of higher-order
functional programs directly, without translating them into low level programs.
In other words, \PCFld\ can be viewed as an abstract framework where to experiment 
new implicit computational complexity techniques.

\section{Related Work}\label{sect:RW}
Other type systems can be proved to satisfy completeness properties similar
to the ones enjoyed by \PCFld. 

The first example that comes to mind is the one of intersection types. In intersection
type disciplines, the class of strongly and weakly normalizable lambda terms 
can be captured~\cite{Dezani-Giovannetti-deLiguoro:Tokyo98}. Recently, these results have been refined in such a way that
the actual complexity of reduction of the underlying term can be read from its
type derivation~\cite{DBLP:journals/corr/abs-0905-4251,bernadetleng11}. What intersection types lack is the possibility to analyze
the behavior of a functional term in one single type derivation --- all function
calls must be typed separately \cite{Leivant:acm:lfp:1990,LICS99*109,
journals/mscs/BucciarelliPS03}. This is in contrast with 
Theorem~\ref{thm:function-relative-completeness} which  gives a unique type derivation for every \PCF\ program computing
a total function on the natural numbers.

Another example of type theories which enjoy completeness properties are refinement
type theories~\cite{DBLP:conf/pldi/FreemanP91}, as shown in~\cite{Denney:1998:RTS:647321.721207}. Completeness, however, only holds
in a logical sense: any property which is true in all Henkin models can be
captured by refinement types. The kind of completeness we obtain here is clearly more
operational: the result of evaluating a program and the time complexity
of the process can both be read off from its type.

As already mentioned in the Introduction, linear logic has been a
great source of inspiration for the authors. Actually, it is not a
coincidence that linear logic was a key ingredient in the development
of one of the earliest fully-abstract game models for \PCF. Indeed,
\PCFld\ can be seen as a way to internalize history-free game
semantics~\cite{Abramsky2000} into a type system. And already
\BLL\ and \QBAL, both precursors of \PCFld, have been designed being
greatly inspired by the geometry of interaction. \PCFld\ is a way to
study the extreme consequences of this idea, when bounds are not only
polynomials but arbitrary first-order total functions on natural
numbers.

\bibliographystyle{abbrv}
\bibliography{main}

\begin{thebibliography}{10}

\bibitem{Abramsky2000}
S.~Abramsky, R.~Jagadeesan, and P.~Malacaria.
\newblock Full abstraction for {PCF}.
\newblock {\em I \& C}, 163(2):409--470, 2000.

\bibitem{ABO09-Book3}
K.~R. Apt, F.~S. de~Boer, and E.-R. Olderog.
\newblock {\em Verification of Sequential and Concurrent Programs}.
\newblock T. in Comp. Sci. Springer-Verlag, 2009.

\bibitem{BaaderNipkow}
F.~Baader and T.~Nipkow.
\newblock {\em Term Rewriting and All That}.
\newblock Cambridge University Press, 1998.

\bibitem{BaillotGaboardiMogbil09esop}
P.~Baillot, M.~Gaboardi, and V.~Mogbil.
\newblock A polytime functional language from light linear logic.
\newblock In {\em {ESOP}}, volume 6012 of {\em LNCS}, pages 104--124. Springer,
  2010.

\bibitem{BaillotTerui}
P.~Baillot and K.~Terui.
\newblock Light types for polynomial time computation in lambda calculus.
\newblock {\em I \& C}, 207(1):41--62, 2009.

\bibitem{conf/csl/BartheGR08}
G.~Barthe, B.~Gr{\'e}goire, and C.~Riba.
\newblock Type-based termination with sized products.
\newblock In {\em {CSL}}, volume 5213 of {\em LNCS}, pages 493--507. Springer,
  2008.

\bibitem{bernadetleng11}
A.~Bernadet and S.~Lengrand.
\newblock Complexity of strongly normalising $\lambda$-terms via non-idempotent
  intersection types.
\newblock In {\em FOSSACS}, volume 6604 of {\em LNCS}, pages 88--107. Springer,
  2011.

\bibitem{LICS99*109}
A.~Bucciarelli, S.~D. Lorenzis, A.~Piperno, and I.~Salvo.
\newblock Some computational properties of intersection types.
\newblock In {\em {LICS}}, pages 109--118. IEEE Comp. Soc., 1999.

\bibitem{journals/mscs/BucciarelliPS03}
A.~Bucciarelli, A.~Piperno, and I.~Salvo.
\newblock Intersection types and lambda-definability.
\newblock {\em MSCS}, 13(1):15--53, 2003.

\bibitem{Cook78}
S.~A. Cook.
\newblock Soundness and completeness of an axiom system for program
  verification.
\newblock {\em SIAM J. on Computing}, 7:70--90, 1978.

\bibitem{CraryWeirich00}
K.~Crary and S.~Weirich.
\newblock Resource bound certification.
\newblock In {\em ACM POPL}, pages 184--198, 2000.

\bibitem{TOCL2009b}
U.~Dal~Lago.
\newblock Context semantics, linear logic and computational complexity.
\newblock {\em ACM TOCL}, 10(4), 2009.

\bibitem{conf/tlca/LagoH09}
U.~Dal~Lago and M.~Hofmann.
\newblock Bounded linear logic, revisited.
\newblock {\em LMCS}, 6(4), 2010.

\bibitem{DBLP:journals/corr/abs-0905-4251}
D.~de~Carvalho.
\newblock Execution time of lambda-terms via denotational semantics and
  intersection types.
\newblock {\em CoRR}, abs/0905.4251, 2009.

\bibitem{Denney:1998:RTS:647321.721207}
E.~Denney.
\newblock Refinement types for specification.
\newblock In {\em IFIP-PROCOMET}, pages 148--166, 1998.

\bibitem{Dezani-Giovannetti-deLiguoro:Tokyo98}
M.~Dezani-Ciancaglini, E.~Giovannetti, and U.~de' Liguoro.
\newblock {{I}ntersection Types, Lambda-models and {B}\"ohm Trees}.
\newblock In {\em ``Theories of Types and Proofs''}, volume~2, pages 45--97.
  Math. Soc. of Japan, 1998.

\bibitem{DBLP:conf/pldi/FreemanP91}
T.~Freeman and F.~Pfenning.
\newblock Refinement types for {ML}.
\newblock In {\em PLDI}, pages 268--277, 1991.

\bibitem{Girard87tcs}
J.-Y. Girard.
\newblock Linear logic.
\newblock {\em Theor. Comp. Sci.}, 50:1--102, 1987.

\bibitem{GSS92}
J.-Y. Girard, A.~Scedrov, and P.~Scott.
\newblock {Bounded linear logic}.
\newblock {\em Theor. Comp. Sci.}, 97(1):1--66, 1992.

\bibitem{conf/icfp/Grobauer01}
B.~Grobauer.
\newblock Cost recurrences for {DML} programs.
\newblock In {\em ICFP}, pages 253--264, 2001.

\bibitem{gunter92mit}
C.~A. Gunter.
\newblock {\em Semantics of Programming Languages: Structures and Techniques}.
\newblock Found. of Comp. Series. MIT Press, 1992.

\bibitem{HAH11}
J.~Hoffmann, K.~Aehlig, and M.~Hofmann.
\newblock {Multivariate Amortized Resource Analysis}.
\newblock In {\em ACM POPL}, pages 357--370, 2011.

\bibitem{conf/lics/Hofmann99a}
M.~Hofmann.
\newblock Linear types and non-size-increasing polynomial time computation.
\newblock In {\em LICS}, pages 464--473. IEEE Comp. Soc., 1999.

\bibitem{Hofmann:2000:PLC:346048.346051}
M.~Hofmann.
\newblock Programming languages capturing complexity classes.
\newblock {\em ACM SIGACT News}, 31:31--42, 2000.

\bibitem{HOAA_POPL10}
S.~Jost, K.~Hammond, H.-W. Loid, and M.~Hofmann.
\newblock {Static Determination of Quantitative Resource Usage for Higher-Order
  Programs}.
\newblock In {\em ACM POPL}, Madrid, Spain, 2010.

\bibitem{conf/lics/KobayashiO09}
N.~Kobayashi and C.-H.~L. Ong.
\newblock A type system equivalent to the modal mu-calculus model checking of
  higher-order recursion schemes.
\newblock In {\em {LICS}}, pages 179--188. IEEE Comp. Soc., 2009.

\bibitem{Krivine07}
J.-L. Krivine.
\newblock A call-by-name lambda-calculus machine.
\newblock {\em Higher-Order and Symbolic Computation}, 20(3):199--207, 2007.

\bibitem{Leivant:acm:lfp:1990}
D.~Leivant.
\newblock Discrete polymorphism.
\newblock In {\em ACM LFP}, pages 288--297. ACM Press, 1990.

\bibitem{MarionHDR}
J.-Y. Marion.
\newblock {\em Complexit\'{e} implicite des calculs, de la th{\'e}orie \`a la
  pratique}.
\newblock Habilitation thesis, Universit\'e {N}ancy 2, 2000.

\bibitem{Odifreddi}
P.~Odifreddi.
\newblock {\em Classical Recursion Theory: the Theory of Functions and Sets of
  Natural Numbers}.
\newblock Number 125 in Studies in Logic and the Foundations of Mathematics.
  North-Holland, 1989.

\bibitem{plotkin77tcs}
G.~D. Plotkin.
\newblock {LCF} considerd as a programming language.
\newblock {\em Theor. Comp. Sci.}, 5:225--255, 1977.

\bibitem{SabelfeldMyers03}
A.~Sabelfeld and A.~C. Myers.
\newblock Language-based information-flow security.
\newblock {\em IEEE JSAC}, 21(1):5--19, 2003.

\bibitem{journals/jcs/VolpanoIS96}
D.~M. Volpano, C.~E. Irvine, and G.~Smith.
\newblock A sound type system for secure flow analysis.
\newblock {\em JCS}, 4(2/3):167--188, 1996.

\bibitem{LICS01*231}
H.~Xi.
\newblock Dependent types for program termination verification.
\newblock In {\em {LICS}}, pages 231--246. IEEE Comp. Soc., 2001.

\bibitem{DBLP:journals/jfp/Xi07}
H.~Xi.
\newblock Dependent ml an approach to practical programming with dependent
  types.
\newblock {\em J. of Funct. Progr.}, 17(2):215--286, 2007.

\bibitem{POPL99*214}
H.~Xi and F.~Pfenning.
\newblock Dependent types in practical programming.
\newblock In {\em ACM POPL}, pages 214--227, 1999.

\end{thebibliography}

\end{document}